\documentclass[conference]{IEEEtran}

\usepackage{calrsfs}
\DeclareMathAlphabet{\pazocal}{OMS}{zplm}{m}{n}
\usepackage{epsfig,endnotes,xspace,hyperref,url,diagbox}
\AtBeginDocument{}
\usepackage{amsmath,amsthm,amssymb}
\usepackage{tikz}
\usepackage{enumerate}
\usepackage{varioref}
\usepackage{graphicx}
\usepackage{epstopdf}
\usepackage{color}
\usepackage{xcolor}
\usepackage{multirow}
\usepackage{subfigure}
\usepackage{comment}
\usepackage[utf8]{inputenc}
\usepackage{mathtools}
\usepackage{wrapfig}
\usepackage{algorithm, algorithmic}
\usepackage{balance}
\usepackage{enumitem}

\definecolor{revision}{RGB}{0,0,255}

\newcommand{\revisionstart}{\begin{color}{revision}}
\newcommand{\revisionend}{~\!\!\end{color}}

\newtheorem{definition}{Definition}

\newtheorem{theorem}{Theorem}

\newtheorem{lemma}{Lemma}

\newcommand{\recht}{\operatorname}

\newcommand{\EV}[1]{\ensuremath{\mathbb{E}\left[\, #1 \,\right]}}

\renewcommand{\Pr}[1]{\ensuremath{\mathsf{Pr}\left[#1\right]}\xspace}

\newcommand{\says}[2]{\noindent\textcolor{purple}{\textbf{#1 says: }}\textcolor{blue}{#2}\xspace}
\newcommand{\mypara}[1]{\vspace*{0.05in}\noindent\textbf{#1} \xspace}

\newcommand{\Domain}{\ensuremath{D}\xspace}

\newcommand{\perturb}{\ensuremath{\Psi}\xspace}
\newcommand{\aggregate}{\ensuremath{\Phi}\xspace}

\newcommand{\olh}{\ensuremath{\mathsf{OLH}}\xspace}

\newcommand{\oue}{\ensuremath{\mathsf{OUE}}\xspace}
\newcommand{\grr}{\ensuremath{\mathsf{GRR}}\xspace}

\newcommand{\fo}{\ensuremath{\mathsf{FO}}\xspace}
\newcommand{\FO}{\ensuremath{\mathsf{FO}}\xspace}
\newcommand{\tuple}[1]{\ensuremath{\langle #1 \rangle}}

\begin{document}

\title{Locally Differentially Private Frequency Estimation with Consistency}

\author{\IEEEauthorblockN{Tianhao Wang$^1$, Milan Lopuhaä-Zwakenberg$^2$, Zitao Li$^1$, Boris Skoric$^2$, Ninghui Li$^1$}
	\IEEEauthorblockA{
		$^1$Purdue University, $^2$Eindhoven University of Technology\\
		\{tianhaowang, li2490, ninghui\}@purdue.edu, \{m.a.lopuhaa, b.skoric\}@tue.nl}
}
\IEEEoverridecommandlockouts
\makeatletter\def\@IEEEpubidpullup{6.5\baselineskip}\makeatother
\IEEEpubid{\parbox{\columnwidth}{
		Network and Distributed Systems Security (NDSS) Symposium 2020\\
		23-26 February 2020, San Diego, CA, USA\\
		ISBN 1-891562-61-4\\
		https://dx.doi.org/10.14722/ndss.2020.24157\\
		www.ndss-symposium.org
	}
	\hspace{\columnsep}\makebox[\columnwidth]{}}

\maketitle

\begin{abstract}
	Local Differential Privacy (LDP) protects user privacy from the data collector. LDP protocols have been increasingly deployed in the industry. A basic building block is frequency oracle (\fo) protocols, which estimate frequencies of values. While several \fo protocols have been proposed, the design goal does not lead to optimal results for answering many queries. In this paper, we show that adding post-processing steps to \fo protocols by exploiting the knowledge that all individual frequencies should be non-negative and they sum up to one can lead to significantly better accuracy for a wide range of tasks, including frequencies of individual values, frequencies of the most frequent values, and frequencies of subsets of values. We consider 10 different methods that exploit this knowledge differently. We establish theoretical relationships between some of them and conducted extensive experimental evaluations to understand which methods should be used for different query tasks.
\end{abstract}

\pagestyle{plain}

\section{Introduction}

Differential privacy (DP)~\cite{tcc:DworkMNS06} has been accepted as the \textit{de facto} standard for data privacy.
Recently, techniques for satisfying DP in the local setting, which we call {LDP}, have been studied and deployed.
In this setting, there are many users and one
aggregator.  The aggregator does not see the actual private data of each individual.  Instead, each user sends randomized information to the aggregator, who attempts to infer the data distribution based on that.
LDP techniques have been deployed by companies like
Apple~\cite{url:apple}, Google~\cite{ccs:ErlingssonPK14}, Microsoft~\cite{nips:DingKY17}, and Alibaba~\cite{sigmod:wang2019answering}.
Examples of use cases include collecting users' default browser homepage and search engine, in order to understand the unwanted or malicious hijacking of user settings; or frequently typed emoji's and words, to help with keyboard typing recommendation.

The fundamental tools in LDP are mechanisms to estimate frequencies of values.  Existing research~\cite{ccs:ErlingssonPK14,stoc:BassilyS15,uss:WangBLJ17,aistats:AcharyaSZ18,tiot:YeB18} has developed frequency oracle (\fo) protocols, where the aggregator can estimate the frequency of any chosen value in the specified domain (fraction of users reporting that value).
While these protocols were designed to provide unbiased estimations of individual frequencies while minimizing the estimation variance~\cite{uss:WangBLJ17}, they can perform poorly for some tasks.  In~\cite{infocom:JiaG18}, it is shown that when one wants to query the frequency of all values in the domain, one can obtain significant accuracy improvement by exploiting the belief that the distribution likely follows power law.
Also, some applications naturally require querying the sums of frequencies for values in a subset.  For example, with the estimation of each emoji's frequency, one may be interested in understanding what categories of emoji's are more popular and need to issue subset frequency queries.
For another example, in~\cite{ccs:ZhangWLHC18}, multiple attributes are encoded together and reported using LDP, and recovering the distribution for each attribute separately requires computing the frequencies of sets of encoded values.
For frequencies of a subset of values, simply summing up the estimations of all values is far from optimal, especially when the input domain is large.

We note that the problem of answering queries using information obtained from the frequency oracle protocols is an estimation problem.  Existing methods such as those in~\cite{uss:WangBLJ17} do not utilize any prior knowledge of the distribution to be estimated.  Due to the significant amount of noise needed to satisfy LDP, the estimations for many values may be negative.  Also, some LDP protocols may result in the total sum of frequencies to be different from one.  In this paper, we show that one can develop better estimation methods by exploiting the universal fact that all frequencies are non-negative and they sum up to 1.

Interestingly, when taking advantage of such prior knowledge, one introduces biases in the estimations.   For example, when we impose the non-negativity constraint, we are introducing positive biases in the estimation as a side effect.  Essentially, when we exploit prior beliefs, the estimations will be biased towards the prior beliefs.  These biases can cause some queries to be much more inaccurate.  For example, changing all negative estimations to zero improves accuracy for frequency estimations of individual values.  However, the introduced positive biases accumulate for range queries.  Different methods to utilize the prior knowledge introduces different forms of biases, and thus have different impacts for different kinds of queries.

In this paper, we consider 10 different methods, which utilizes prior knowledge differently.  Some methods enforce only non-negativity; some other methods enforce only that all estimations sum to 1; and other methods enforce both.  These methods can also be combined with the ``Power'' method in~\cite{infocom:JiaG18} that exploits power law assumption.

We evaluate these methods on three tasks, frequencies of individual values, frequencies of the most frequent values, and frequencies of subsets of values.  We find that there is no single method that out-performs other methods for all tasks.  A method that exploits only non-negativity performs the best for individual values; a method that exploits only the summing-to-one constraint performs the best for frequent values; and a method that enforces both can be applied in conjunction with Power to perform the best for subsets of values.

To summarize, the main contributions of this paper are threefold:
\begin{itemize}
	\item We introduced the consistency properties as a way to improve accuracy for \fo protocols under LDP, and summarized 10 different post-processing methods that exploit the consistency properties differently.
	\item We established theoretical relationships between Constrained Least Squares and Maximum Likelihood Estimation, and analyze which (if any) estimation biases are introduced by these methods.
	\item We conducted extensive experiments on both synthetic and real-world datasets, the results improved the understanding on the strengths and weaknesses of different approaches.
\end{itemize}

\mypara{Roadmap.}
In Section~\ref{sec:problem}, we give the problem definition, followed by the background information on \fo in Section~\ref{sec:fo}.
We present the post-processing methods in Section~\ref{sec:approach}.
Experimental results are presented in~\ref{sec:exp}.
Finally we discuss related work in Section~\ref{sec:related} and
provide concluding remarks in Section~\ref{sec:conc}.
\section{Problem Setting}
\label{sec:problem}

We consider the setting where there are many \emph{users} and one \emph{aggregator}.  Each user possesses a value ${v}$ from a finite domain $\Domain$, and the aggregator wants to learn the distribution of values among all users, in a way that protects the privacy of individual users.  More specifically, the aggregator wants to estimate, for each value  $v\in \Domain$, the fraction of users having $v$ (the number of users having $v$ divided by the population size).    Such protocols are called \textit{frequency oracle ($\fo$)} protocols under Local Differential Privacy (LDP), and they are the key building blocks of other LDP tasks.

\mypara{Privacy Requirement.}
An \fo protocol is specified by a pair of algorithms: $\perturb$ is used by each user to perturb her input value, and $\aggregate$ is used by the aggregator.  Each user sends $\perturb(v)$ to the aggregator.  The formal {privacy} requirement is that the algorithm $\perturb(\cdot)$ satisfies the following property:

\begin{definition}[$\epsilon$-Local Differential Privacy] \label{def:dlp}
	An algorithm $\perturb(\cdot)$ satisfies $\epsilon$-local differential privacy ($\epsilon$-LDP), where $\epsilon \geq 0$,
	if and only if for any input $v, v' \in \Domain$, we have
	\begin{equation*}
		\forall{y\in\! \Psi(D)}:\; \Pr{\perturb(v) = y} \leq e^{\epsilon}\, \Pr{\perturb(v') = y},\end{equation*}
	where $\Psi(D)$ is discrete and denotes the set of all possible outputs of $\perturb$.
\end{definition}

Since a user never reveals $v$ to the aggregator and reports only $\perturb(v)$, the user's privacy is still protected even if the aggregator is malicious.

\mypara{Utility Goals.}
The aggregator uses $\aggregate$, which takes the vector of all reports from users as the input, and produces $\mathbf{\tilde{f}} = \tuple{\tilde{f}_v}_{v\in D}$, the estimated frequencies of the $v \in D$ (i.e., the fraction of users who have input value $v$).  As $\perturb$ is a randomized function, the resulting $\mathbf{\tilde{f}}$ becomes inaccurate.

In existing work, the design goal for $\perturb$ and $\aggregate$ is that the estimated frequency for each $v$ is unbiased, and the variance of the estimation is minimized.  As we will show in this paper, these may not result in the most accurate answers to different queries.

In this paper, we consider three different query scenarios 1) query the frequency of every value in the domain, 2) query the aggregate frequencies of subsets of values, and 3) query the frequencies of the most frequent values.  For each value or set of values, we compute its estimate and the ground truth, and calculate their difference, measured by {Mean of Squared Error} (MSE).

\mypara{Consistency.}
We will show that the utility of existing mechanisms can be improved by enforcing the following consistency requirement.

\begin{definition}[Consistency]
	\label{def:consistency}
	The estimated frequencies are consistent if and only if the following two conditions are satisfied:
	\begin{enumerate}
		\item The estimated frequency of each value is non-negative.
		\item The sum of the estimated frequencies is $1$.
	\end{enumerate}
\end{definition}

\section{Frequency Oracle Protocols}
\label{sec:fo}

We review the state-of-the-art frequency oracle protocols.  We utilize the generalized view from~\cite{uss:WangBLJ17} to present the protocols, so that our post-processing procedure can be applied to all of them.

\subsection{Generalized Random Response (\grr)}

This \fo protocol generalizes the \emph{randomized response} technique~\cite{jasa:Warner65}.  Here each user with private value $v\in D$ sends the true value $v$ with probability $p$, and with probability $1-p$ sends a randomly chosen $v'\in D\setminus\{v\}$.
Suppose the domain $\Domain$ contains $d=|\Domain|$ values, the perturbation function is formally defined as
\begin{align}
	\forall_{y \in D}\;\Pr{\perturb_{\grr(\epsilon, d)}(v) \!=\! y}  \!=\! \left\{
	\begin{array}{lr}
		\!p\!=\!\frac{e^\epsilon}{e^\epsilon + d - 1}, & \mbox{if} \; y = v    \\
		\!q\!=\! \frac{1}{e^\epsilon + d - 1},         & \mbox{if} \; y \neq v \\
	\end{array}\label{eq:grr}
	\right.
\end{align}
This satisfies $\epsilon$-LDP since $\frac{p}{q}=e^\epsilon$.

From a population of $n$ users, the aggregator receives a length-$n$ vector $\textbf{y}=\langle y_1,y_2,\cdots,y_n\rangle$, where $y_i \in \Domain$ is the reported value of the $i$-th user.
The aggregator counts the number of times each value $v$ appears in $\textbf{y}$ and produces a length-$d$ vector $\mathbf{c}$ of natural numbers.
Observe that the components of $\mathbf{c}$ sum up to $n$, i.e., $\sum_{v \in \Domain} c_v = n$.
The aggregator then obtains the estimated frequency vector $\mathbf{\tilde{f}}$ by scaling each component of $\mathbf{c}$ as follows: \begin{align*}
	\tilde{f}_v = & \frac{\frac{c_v}{n} - q}{p - q}
	= \frac{\frac{c_v}{n} - \frac{1}{e^\epsilon + d - 1}}{\frac{e^\epsilon - 1}{e^\epsilon + d - 1}}
\end{align*}

As shown in~\cite{uss:WangBLJ17}, the estimation variance of \grr grows linearly in $d$; hence the accuracy deteriorates fast when the domain size $d$ increases.
This motivated the development of other \fo protocols.

\subsection{Optimized Local Hashing (\olh)}
This \fo deals with a large domain size $d$ by first using a random hash function to map an input value into a smaller domain of size $g$, and then applying randomized response to the hash value in the smaller domain.
In \olh, the reporting protocol is
\[
	\perturb_{\olh(\epsilon)}(v)\coloneqq\tuple{H,\;\perturb_{\grr(\epsilon, g)}(H(v))},
\]
where $H$ is randomly chosen from a family of hash functions that hash each value in $\Domain$ to $\{1\ldots g\}$,
and $\perturb_{\mathsf{GRR}(\epsilon, g)}$ is given in~\eqref{eq:grr}, while operating on the domain $\{1\ldots g\}$.
The hash family should have the property that the distribution of each $v$'s hashed result is uniform over $\{1\ldots g\}$ and independent from the distributions of other input values in $\Domain$.  Since $H$ is chosen independently of the user's input $v$, $H$ by itself carries no meaningful information.  Such a report $\tuple{H,r}$ can be represented by the set $Y=\{ y \in \Domain \mid H(y)=r \}$.
The use of a hash function can be viewed as a compression technique, which results in constant size encoding of a set.  For a user with value $v$, the probability that $v$ is in the set $Y$ represented by the randomized report $\tuple{H,r}$ is $p=\frac{e^\epsilon - 1}{e^\epsilon + g - 1}$ and the probability that a user with value $\ne v$ is in $Y$ is $q=\frac{1}{g}$.

For each value $x\in \Domain$, the aggregator first computes the vector $\mathbf{c}$ of how many times each value is in the reported set.
More precisely, let $Y_i$ denote the set defined by the user $i$, then $c_v = |\{i\mid H(v) \in Y_i\}|$.  The aggregator then scales it:
\begin{align}
	\label{eq:olh_aggregate}
	\tilde{f}_v =
	\frac{\frac{c_v}{n} - 1/g}{p - 1/g}
\end{align}
In \olh, both the hashing step and the randomization step result in information loss. The choice of the parameter $g$ is a tradeoff between losing information during the hashing step and losing information during the randomization step.  It is found that the estimation variance when viewed as a continuous function of $g$ is minimized when $g = e^\epsilon + 1$ (or the closest integer to $e^\epsilon + 1$ in practice)~\cite{uss:WangBLJ17}.

\subsection{Other \fo Protocols}

Several other \fo protocols have been proposed.  While they take different forms when originally proposed, in essence, they all have the user report some encoding of a subset $Y \subseteq \Domain$, so that the user's true value has a probability $p$ to be included in $Y$ and any other value has a probability $q<p$ to be included in $Y$.  The estimation method used in \grr and \olh (namely, $\tilde{f}_v = \frac{c_v/n - q}{p - q}$) equally applies.

\begin{comment}
\begin{table}
	\centering
	\begin{tabular}{|c|l|}\hline
		\textbf{Notion} & \textbf{Description}                 \\ \hline
		$n$             & Number of users.                     \\ \hline
		$\Domain$       & Domain of all possible input values. \\ \hline
		$d$             & Size of Domain $\Domain$             \\ \hline
		$c_v$           & Count of value $v$ in reported sets. \\ \hline
		$f_v$           & True frequency of $v$.               \\ \hline
		$\tilde{f}_v$   & Estimated frequency of $v$ by $\fo$. \\ \hline
		$f'_v$          & Post-processed frequency of $v$.     \\ \hline
		$p$             & High probabilities used in $\fo$.    \\ \hline
		$q$             & Low probabilities used in $\fo$.     \\ \hline
	\end{tabular}
	\caption{Summary of Notations.}\label{tab:notation_summary}
\end{table}
\end{comment}

\begin{comment}
\mypara{Random Matrix Projection~\cite{stoc:BassilyS15}} motivated \olh, and is equivalent to using local hash with $g$ fixed at $2$.  That is, each reported set consists of approximatly half of elements in $\Domain$.  When $\epsilon < \ln 2 \approx 0.69$, \olh is exactly the same as this method.
With larger $\epsilon$, \olh performs better.
\end{comment}

\mypara{Optimized Unary Encoding~\cite{uss:WangBLJ17}} encodes a value in a size-$d$ domain using a length-$d$ binary vector, and then perturbs each bit independently.  The resulting bit vector encodes a set of values.  It is found in~\cite{uss:WangBLJ17} that when $d$ is large, one should flip the $1$ bit with probability $1/2$, and flip a $0$ bit with probability $1/e^{\epsilon}$.  This results in the same values of $p,q$ as \olh, and has the same estimation variance, but has higher communication cost (linear in domain size $d$).

\mypara{Subset Selection~\cite{tiot:YeB18,corr:WangHWNXYLQ16}} method reports a randomly selected subset of a fixed size $k$.
The sensitive value $v$ is included in the set with probability $p=1/2$.  For any other value, it is included with probability $q = p \cdot \frac{k-1}{d-1} + (1 - p) \cdot \frac{k}{d-1}$.
To minimize estimation variance, $k$ should be an integer equal or close to $d /( e^\epsilon + 1)$.
Ignoring the integer constraint, we have $q = \frac{1}{2}\cdot \frac{2k - 1}{d - 1} = \frac{1}{2}\cdot \frac{2\frac{d}{e^\epsilon + 1} - 1}{d - 1} = \frac{1}{e^\epsilon + 1}\cdot \frac{d - (e^\epsilon + 1)/2}{d - 1} < \frac{1}{e^\epsilon + 1}$.
Its variance is smaller than that of \olh.  However, as $d$ increases, the term $\frac{d - (e^\epsilon + 1)/2}{d - 1}$ gets closer and closer to $1$.  For a larger domain, this offers essentially the same accuracy as \olh, with higher communication cost (linear in domain size $d$).

\mypara{Hadamard Response~\cite{nips:BassilyNST17, aistats:AcharyaSZ18}} is similar to Subset Selection with $k=d/2$, where the Hadamard transform is used to compress the subset.  The benefit of adopting this protocol is to reduce the communication bandwidth (each user's report is of constant size).  While it is similar to \olh with $g=2$, its aggregation part $\Phi$ faster, because evaluating a Hadamard entry is practically faster than evaluating hash functions.  However, this FO is sub-optimal when $g=2$ is sub-optimal.

\subsection{Accuracy of Frequency Oracles}
In~\cite{uss:WangBLJ17}, it is proved that $\tilde{f}_v = \frac{c_v/n - q}{p - q}$ produces unbiased estimates.  That is, $\forall v\in\Domain, \;\EV{\tilde{f}_v} = f_v$.  Moreover, $\tilde{f}_v$ has variance
\begin{align}
	\sigma^2_v= & \frac{q(1-q) + {f}_v(p-q)(1-p-q)}{n(p-q)^2}\label{eq:var_general}
\end{align}
As $c_v$ follows Binomial distribution, by the central limit theorem, the estimate $\tilde{f}_v$ can be viewed as the true value $f_v$ plus a Normally distributed noise:
\begin{align}
	\tilde{f}_v \approx f_v + \mathcal{N}(0, \sigma_v).\label{eq:normal_general}
\end{align}

When $d$ is large and $\epsilon$ is not too large, ${f}_v(p-q)(1-p-q)$ is dominated by $q(1-q)$.  Thus, one can approximate Equation~\eqref{eq:var_general} and~\eqref{eq:normal_general} by ignoring the $f_v$.  Specifically,
\begin{align}
	\sigma^2\approx     & \;\frac{q(1-q)}{n(p-q)^2},\label{eq:var_appox}           \\
	\tilde{f}_v \approx & \; f_v + \mathcal{N}(0, \sigma).\label{eq:normal_approx}
\end{align}

As the probability each user's report support each value is independent, we focus on post-processing $\mathbf{\tilde{f}}$ instead of $\mathbf{Y}$.

\begin{comment}

\subsection{A Uniform View}

For the purpose of designing post-processing algorithms that work for all \fo protocols, we now present a uniform view of frequency oracles.
In all protocols, each user is reporting a set $Y \subset \Domain$ of values.  Such a report means that given two values $v_1\in Y$ and $v_2\in \Domain \setminus Y$, the user's true value is more likely to be $v_1$ than $v_2$.

In GRR, the user is reporting a singleton $Y$.  In OUE, the set $Y$ is encoded in the form of a bit vector over all values.  In OLH, the set $Y$ is encoded using a (randomly chosen) hash function $H$ and a hash value.  The reported set $Y$ is the set of all values in $\Domain$ hashed to the reported value by $H$.
\end{comment}

%

\section{Towards Consistent Frequency Oracles}
\label{sec:approach}

While existing state-of-the-art frequency oracles are designed to provide unbiased estimations while minimizing the variance, it is possible to further reduce the variance by performing post-processing steps that use prior knowledge to adjust the estimations.
For example, exploiting the property that all frequency counts are non-negative can reduce the variance; however, simply turning all negative estimations to 0 introduces a systematic positive bias in all estimations.  By also ensuring the property that the sum of all estimations must add up to 1, one ensures that the sum of the biases for all estimations is 0.  However, even though the biases cancel out when summing over the whole domain, they still exist.
There are different post-processing methods that were explicitly proposed or implicitly used.  They will result in different combinations of variance reduction and bias distribution.  Selecting a post-processing method is similar to considering the bias-variance tradeoff in selecting a machine learning algorithm.

We study the property of several post-processing methods, aiming to understand how they compare under different settings, and how they relate to each other.
Our goal is to identify efficient post-processing methods that can give accurate estimations for a wide variety of queries.
We first present the baseline method that does not do any post-processing.

\begin{itemize}[leftmargin=*]
	\item
	      \underline{\textbf{Base}}:
	      \textit{We use the standard \fo as presented in Section~\ref{sec:fo} to obtain estimations of each value. }
\end{itemize}

Base has no bias, and its variance can be analytically computed (e.g., using~\cite{uss:WangBLJ17}).

\subsection{Baseline Methods}

When the domain is large, there will be many values in the domain that have a zero or very low true frequency; the estimation of them may be negative.  To overcome negativity, we describe three methods: Base-Pos, Post-Pos, and Base-Cut.

\begin{itemize}[leftmargin=*]
	\item
	      \underline{\textbf{Base-Pos}}:
	      \textit{After applying the standard \fo, we convert all negative estimations to $0$.}
\end{itemize}
This satisfies non-negativity, but the sum of all estimations is likely to be above 1.
This reduces variance, as it turns erroneous negative estimations to 0, closer to the true value.  As a result, for each individual value, Base-Pos results in an estimation that is at least as accurate as the Base method.  However, this introduces systematic positive bias, because some negative noise are removed or reduced by the process, but the positive noise are never removed.  This positive bias will be reflected when answering subset queries, for which Base-Pos results in biased estimations.  For larger-range queries, the bias can be significant.
\begin{lemma}
	\label{lem:base-pos-bias}
	Base-Pos will introduce positive bias to all values.
\end{lemma}
\begin{proof}
	The outputs of standard \FO are unbiased estimation, which means for any $v$,
	\begin{align*}
		f_v = \EV{\tilde{f}_v} = \EV{\tilde{f}_v \cdot \mathbf{1}[\tilde{f}_v \ge 0]} + \EV{\tilde{f}_v \cdot \mathbf{1}[\tilde{f}_v < 0]}
	\end{align*}
	As Base-Pos changes all negative estimated frequencies to 0, we have
	\begin{align*}
		\EV{f'_v} = \EV{\tilde{f}_v \cdot \mathbf{1}[\tilde{f}_v \ge 0]}
	\end{align*}
	After enforcing non-negativity constraints, the bias will be $\EV{f'_v} - f_v > 0$.
\end{proof}

\begin{itemize}[leftmargin=*]
	\item
	      \underline{\textbf{Post-Pos}}:
	      \textit{For each query result, if it is negative, we convert it to $0$.}
\end{itemize}
This method does not post-process the estimated distribution.  Rather, it post-processes each query result individually.
For subset queries, as the results are typically positive, Post-Pos is similar to Base.  On the other hand, when the query is on a single item, Post-Pos is equivalent to Base-Pos.

Post-Pos still introduces a positive bias, but the bias would be smaller for subset queries.  However, Post-Pos may give inconsistent answers in the sense that the query result on $A\cup B$, where $A$ and $B$ are disjoint, may not equal the addition of the query results for $A$ and $B$ separately.

\begin{itemize}[leftmargin=*]
	\item
	      \underline{\textbf{Base-Cut}}:
	      \textit{After standard \fo, convert everything below some sensitivity threshold to 0.}
\end{itemize}
The original design goal for frequency oracles is to recover frequencies for \emph{frequent} values, and oftentimes there is a sensitivity threshold so that only estimations above the threshold are considered.
Specifically, for each value, we compare its estimation with a threshold
\begin{align}
	T = F^{-1}\left(1-\frac{\alpha}{d}\right) \sigma,\label{eq:sig_threshold}
\end{align}
where $d$ is the domain size, $F^{-1}$ is the inverse of cummulative distribution function of the standard normal distribution, and $\sigma$ is the standard deviation of the LDP mechanism (i.e., as in Equation~\eqref{eq:var_appox}).
By Base-Cut, estimations below the threshold are considered to be noise.
When using such a threshold, for any value $v\in \Domain$ whose original count is $0$, the probability that it will have an estimated frequency  above $T$ (or the probability a zero-mean Gaussian variable with standard deviation $\delta$ is above $T$) is at most $\frac{\alpha}{d}$.
Thus when we observe an estimated frequency above $T$, the probability that the true frequency of the value is $0$ is (by union bound) at most $d\times \frac{\alpha}{d}=\alpha$.   In~\cite{ccs:ErlingssonPK14}, it is recommended to set $\alpha=5\%$, following conventions in the statistical community.

Empirically we observe that $\alpha=5\%$ performs poorly, because such a threshold can be too high when the population size is not very large and/or the $\epsilon$ is not large.  A large threshold results in all except for a few estimations to be below the threshold and set to 0.
We note that the choice of $\alpha$ is trading off false positives with false negatives.  Given a large domain, there are likely between several and a few dozen values that have quite high frequencies, with most of the remaining values having low true counts.  We want to keep an estimation if it is a lot more likely to be from a frequent value than from a very low frequency one.  In this paper, we choose to set $\alpha=2$, which ensures that the expected number of \emph{false positives}, i.e., values with very low true frequencies but estimated frequencies above $T$, to be around $2$.  If there are around 20 values that are truly frequent and have estimated frequencies above $T$, then ratio of true positives to false positives when using this threshold is 10:1.

This method ensures that all estimations are non-negative.  It does not ensure that the sum of estimations is 1.
The resulting estimations are either high (above the chosen threshold) or zero.  The estimation for each item with non-zero frequency is subject to two bias effects.  The negative bias effect is caused by the situation when the estimations are cut to zero.  The positive effect is when large positive noise causes the estimation to be above the threshold, the resulting estimation is higher than true frequency.

\subsection{Normalization Method}

We now explore several methods that normalize the estimated frequencies of the whole domain to ensure that the sum of the estimates equals $1$.  When the estimations are normalized to sum to 1, the sum of the biases over the whole domain has to be 0.

\begin{lemma}
	\label{lem:general-norm-unbias}
	If a normalization method adjusts the unbiased estimates so that they add up to $1$, the sum of biases it introduces over the whole domain is $0$.
\end{lemma}
\begin{proof}
	Denote $f'_v$ as the estimated frequency of value $v$ after post-processing.  By linearity of expectations, we have
	\begin{align*}
		\sum_{v\in\Domain}\left(\EV{f'_v} - f_v\right)
		= \EV{\sum_{v\in\Domain} f'_v} - \sum_{v\in\Domain}f_v
		= \EV{1} - 1 = 0
	\end{align*}
\end{proof}

One standard way to do such normalization is through additive normalization:

\begin{itemize}[leftmargin=*]
	\item
	      \underline{\textbf{Norm}}:
	      \textit{After standard \fo, add $\delta$ to each estimation so that the overall sum is 1.}
\end{itemize}
The method is formally proposed for the centralized setting~\cite{pvldb:HayRMS10} of DP and is used in the local setting, e.g.,~\cite{icde:WangXYHSSY18,vldb:KulkarniCD18}.  Note the method does not enforce non-negativity.  For \grr, Hadamard Response, and Subset Selection, this method actually does nothing, since each user reports a single value, and the estimations already sum to 1.  For \olh, however, each user reports a randomly selected subset whose size is a random variable, and Norm would change the estimations.
It can be proved that  Norm is unbiased:
\begin{lemma}
	\label{lem:norm-unbias}
	Norm provides unbiased estimation for each value.
\end{lemma}
\begin{proof}
	By the definition of Norm, we have $\sum_{v \in D}f'_v = \sum_{v \in D}(\tilde{f}_v + \delta) = 1$.
	As the frequency oracle outputs unbiased estimation, i.e., $\EV{\tilde{f}_v} = f_v$, we have
	\begin{align*}
		         & \EV{\sum_{v\in D}f'_v} = 1 = \EV{\sum_{v\in D}(\tilde{f}_v + \delta)}       \\
		=        & \sum_{v\in D}\EV{\tilde{f}_v} + d\cdot \EV{\delta} = 1 + d\cdot \EV{\delta} \\
		\implies & \EV{\delta} = 0
	\end{align*}
	Thus $\EV{f'_v} = \EV{\tilde{f}_v + \delta} = \EV{\tilde{f}_v} + 0 = f_v.$
\end{proof}

Besides sum-to-one, if a method also ensures non-negativity, we first state that it introduces positive bias to values whose frequencies are close to 0.
\begin{lemma}
	\label{lem:norm-nonneg-pos-bias}
	If a normalization method adjusts the unbiased estimates so that they add up to $1$ and are non-negative, then it introduces positive biases to values that are sufficiently close to $0$.
\end{lemma}
\begin{proof}
	As the estimates are non-negative and sum up to $1$, some of the estimates must be positive.  For a value close to $0$, there exists some possibility that its estimation is positive; but the possibility its estimation is negative is $0$.  Thus the expectation of its estimation is positive, leading to a positive bias.
\end{proof}
Lemma~\ref{lem:norm-nonneg-pos-bias} shows the biases for any method that ensures both constraints cannot be all zeros.  Thus different methods are essentially different ways of distributing the biases.  Next we present three such normalization methods.

\begin{itemize}[leftmargin=*]
	\item
	      \underline{\textbf{Norm-Mul}}:
	      \textit{After standard \fo, convert negative value to 0.  Then multiply each value by a multiplicative factor so that the sum is 1.}
\end{itemize}
More precisely, given estimation vector $\mathbf{\tilde{f}}$, we find $\gamma$ such that
\begin{align*}
	\sum_{v\in\Domain} \max(\gamma \times \tilde{f}_v,0) = 1,
\end{align*}
and assign $f'_v = \max(\gamma \times \tilde{f}_v,0)$ as the estimations.  This results in a consistent \fo.  Kairouz et al.~\cite{icml:KairouzBR16} evaluated this method and it performs well when the underlying dataset distribution is smooth.
This method results in positive biases for low-frequency items, but negative biases for high-frequency items.  Moreover, the higher an item's true frequency, the larger the magnitude of the negative bias.  The intuition is that here $\gamma$ is typically in the range of $[0, 1]$; and multiplying by a factor may result in the estimation of high frequency values to be significantly lower than their true values.
When the distribution is skewed, which is more interesting in the LDP case, the method performs poorly.

\begin{itemize}[leftmargin=*]
	\item
	      \underline{\textbf{Norm-Sub}}:
	      \textit{After standard \fo, convert negative values to 0, while maintaining overall sum of 1 by adding $\delta$ to each remaining value.}
\end{itemize}
More precisely, given estimation vector $\mathbf{\tilde{f}}$, we want to find $\delta$ such that
$$ \sum_{v\in\Domain} \max(\tilde{f}_v+\delta,0) = 1$$
Then the estimation for each value $v$ is $f'_v = \max(\tilde{f}_v+\delta,0)$. This extends the method Norm and results in consistency.  Norm-Sub was used by Kairouz et al.~\cite{icml:KairouzBR16} and Bassily~\cite{aistats:Bassily19} to process results for some \fo's.
Under Norm-Sub, low-frequency values have positive biases, and high-frequency items have negative biases.  The distribution of biases, however, is more even when compared to Norm-Mul.

\begin{comment}
\says{tianhao}{maybe we can present this in a way similar to Norm-Cut:}
Given estimation vector $\mathbf{\tilde{f}}$, we want to find the smallest $\theta$ such that
$$ \sum_{v\in\Domain | \tilde{f}_v \ge \theta} \tilde{f}_v - \theta = 1$$
Then the estimation for each value $v$ is $0$ if $\tilde{f}_v<\theta$ and $\tilde{f}_v - \theta$ if $\tilde{f}_v \ge \theta$.
\end{comment}

\begin{itemize}[leftmargin=*]
	\item
	      \underline{\textbf{Norm-Cut}}:
	      \textit{After standard \fo, convert negative and small positive values to 0 so that the total sums up to 1.}
\end{itemize}
We note that under Norm-Sub, higher frequency items have higher negative biases.
One natural idea to address this is to turn the low estimations to $0$ to ensure consistency, without changing the estimations of high-frequency values.  This is the idea of Norm-Cut.  More precisely, given the estimation vector $\mathbf{\tilde{f}}$, there are two cases.  When $\sum_{v\in\Domain} \max(\tilde{f}_v,0) \le 1$, we simply change each negative estimations to 0.
When $\sum_{v\in\Domain} \max(\tilde{f}_v,0) > 1$, we want to find the smallest $\theta$ such that
$$ \sum_{v\in\Domain | \tilde{f}_v \ge \theta} \tilde{f}_v \le 1$$
Then the estimation for each value $v$ is $0$ if $\tilde{f}_v<\theta$ and $\tilde{f}_v$ if $\tilde{f}_v \ge \theta$.
This is similar to Base-cut in that both methods change all estimated values below some thresholds to 0.  The differences lie in how the threshold is chosen.  This results in non-negative estimations, and typically results in estimations that sum up to 1, but might result in a sum $<1$.

\begin{comment}
\mypara{Norm-WSub}\textit{}

After standard \fo, convert negative value to 0, we then subtract a small delta from all other positive values while maintaining a sum of 1.  Basically, we want to find $\delta$ such that after subtracting $\delta/\hat{a}_x$ from each estimation for $x$, and then turning negative values to 0, the sum is 1.
\end{comment}
\subsection{Constrained Least Squares}

From a more principled point of view, we note that what we are doing here is essentially solving a Constraint Inference (CI) problem, for which CLS (Constrained Least Squares) is a natural solution.  This approach was proposed in~\cite{pvldb:HayRMS10} but without the constraint that the estimates are non-negative (and it leads to Norm).  Here we revisit this approach with the consistency constraint (i.e., both requirements in Definition~\ref{def:consistency}).

\begin{itemize}[leftmargin=*]
	\item
	      \underline{\textbf{CLS}}:
	      \textit{After standard \fo, use least squares with constraints (summing-to-one and non-negativity) to recover the values.}
\end{itemize}
Specifically, given the estimates $\mathbf{\tilde{f}}$ by \fo, the method outputs $\mathbf{f}'$ that is a solution of the following problem:
\begin{align*}
	\mbox{minimize: } & ||\mathbf{f}' - \mathbf{\tilde{f}}||_2 \\
	\mbox{subject to: }
	                  & \forall_v f'_v \geq 0                  \\
	                  & \sum_v f'_v = 1
\end{align*}

We can use the KKT condition~\cite{book:kuhn2014nonlinear,disertation:karush1939minima} to solve the problem.  The process is presented in Appendix~\ref{app:kkt_cls}.  In the solution, we partition the domain $D$ into $D_0$ and $D_1$, where $D_0\cap D_1 = \emptyset$ and $D_0\cup D_1 = D$.  For $v\in D_0$, assign $f'_v = 0$.  For $v \in D_1$,
\begin{align*}
	f'_v = & \tilde{f}_v - \frac{1}{|D_1|}\left(\sum_{v\in D_1}\tilde{f}_v - 1\right)
\end{align*}

Norm-Sub is the solution to the Constraint Least Square (CLS) formulation to the problem, and $\delta = - \frac{1}{|D_1|}\left(\sum_{v\in D_1}\tilde{f}_v - 1\right)$ is the $\delta$ we want to find in Norm-Sub.

\begin{table*}[!ht]
	\centering
	\begin{tabular}{|c|l|c|c|c|}\hline
		\textbf{Method} & \textbf{Description}                                                & \textbf{Non-neg} & \textbf{Sum to 1} & \textbf{Complexity}  \\\hline
		Base-Pos        & Convert negative est. to 0                                          & Yes              & No                & $O(d)$               \\ \hline
		Post-Pos        & Convert negative query result to 0                                  & Yes              & No                & N/A                  \\ \hline
		Base-Cut        & Convert est. below threshold $T$ to 0                               & Yes              & No                & $O(d)$               \\ \hline
		Norm            & Add $\delta$ to est.                                                & No               & Yes               & $O(d)$               \\ \hline
		Norm-Mul        & Convert negative est. to 0, then multiply $\gamma$ to positive est. & Yes              & Yes               & $O(d)$               \\ \hline
		Norm-Cut        & Convert negative and small positive est. below $\theta$ to 0.       & Yes              & Almost            & $O(d)$               \\ \hline
		Norm-Sub        & Convert negative est. to 0 while adding $\delta$ to positive est.   & Yes              & Yes               & $O(d)$               \\ \hline
		MLE-Apx         & Convert negative est. to 0, then add $\delta$ to positive est.      & Yes              & Yes               & $O(d)$               \\ \hline
		Power           & Fit Power-Law dist., then minimize expected squared error           & Yes              & No                & $O(\sqrt{n}\cdot d)$ \\ \hline
		PowerNS         & Apply Norm-Sub after Power                                          & Yes              & Yes               & $O(\sqrt{n}\cdot d)$ \\ \hline
	\end{tabular}
	\caption{Summary of Methods.}
	\label{tab:method_summary}
	\vspace{-0.8cm}
\end{table*}

\subsection{Maximum Likelihood Estimation}

Another more principled way of looking into this problem is to view it as recovering distributions given some LDP reports.  For this problem, one standard solution is Bayesian inference.
In particular,
we want to find the $\mathbf{{f'}}$ such that
\begin{align}
	\Pr{\mathbf{{f'}}|\mathbf{\tilde{f}}}=
	\frac{\Pr{\mathbf{\tilde{f}}|\mathbf{{f'}}}\cdot\Pr{\mathbf{{f'}}}}{\Pr{\mathbf{\tilde{f}}}}\label{eq:bayesian_inference}
\end{align}
is maximized.
Note that we require $\mathbf{f'}$ satisfies $\forall_v f'_v \geq 0$ and $\sum_v f'_v = 1$.
In~\eqref{eq:bayesian_inference}, $\Pr{\mathbf{{f'}}}$ is the prior, and the prior distribution influence the result.  In our setting, as we assume there is no such prior, $\Pr{\mathbf{{f'}}}$ is uniform.  That is, $\Pr{\mathbf{{f'}}}$ is a constant.  The denominator $\Pr{\mathbf{\tilde{f}}}$ is also a constant that does not influence the result.  As a result, we are seeking for $\mathbf{{f'}}$
which is the maximal likelihood estimator (MLE), i.e., $\Pr{\mathbf{\tilde{f}}|\mathbf{{f'}}}$ is maximized.

For this method, Peter et al.~\cite{icml:KairouzBR16} derived the exact MLE solution for \grr and RAPPOR~\cite{ccs:ErlingssonPK14}.  We compute $\Pr{\mathbf{\tilde{f}}|\mathbf{{f'}}}$ using the general form of Equation~\eqref{eq:normal_general}, which states that, given the original distribution $\mathbf{{f'}}$, the vector $\tilde{\mathbf{f}}$ is a set of independent random variables, where each component $\tilde{f}_v$ follows Gaussian distribution with mean ${f'}_v$ and variance $\sigma_v^{\prime 2} $.
The likelihood of $\tilde{\mathbf{f}}$ given $\mathbf{f'}$ is thus
\begin{align}
	        & \Pr{\mathbf{\tilde{f}}|\mathbf{{f'}}} = \prod_v\Pr{\tilde{f}_v|f'_v}\nonumber                               \\
	\approx & \prod_v\frac{1}{\sqrt{2\pi\sigma_v^{\prime 2}}}\cdot e^{-\frac{(f'_v-\tilde{f}_v)^2}{2\sigma_v^{\prime 2}}}
	=  \frac{1}{\sqrt{2\pi \prod_v \sigma_v^{\prime 2}}}\cdot e^{-\sum_v\frac{(f'_v-\tilde{f}_v)^2}{2\sigma_v^{\prime 2}}}.\label{eq:mle_apx}
\end{align}
To differentiate from~\cite{icml:KairouzBR16}, we call it MLE-Apx.

\begin{itemize}[leftmargin=*]
	\item
	      \underline{\textbf{MLE-Apx}}:
	      \textit{First use standard \fo, then compute the MLE with constraints (summing-to-one and non-negativity) to recover the values.}
\end{itemize}
In Appendix~\ref{app:kkt_mle}, we use the KKT condition~\cite{book:kuhn2014nonlinear,disertation:karush1939minima} to obtain an efficient solution.
In particular, we partition the domain $D$ into $D_0$ and $D_1$, where $D_0\cap D_1 = \emptyset$ and $D_0\cup D_1 = D$.  For $v\in D_0$, $f'_v = 0$; for $v \in D_1$,
\begin{align}
	f'_v = & \frac{q(1 - q)x_v + \tilde{f}_v(p - q)}{p - q - (p-q)(1-p-q)x_v}\label{eq:mle_sol}
\end{align}
where
\begin{align*}
	x_v = & \frac{\sum_{x\in D_1} \tilde{f}_v(p - q)  - (p - q)}{ (p-q)(1-p-q) - |D_1| q(1-q)}
\end{align*}
We can rewrite Equation~\eqref{eq:mle_sol} as
\begin{align*}
	{f'}_v = & \tilde{f}_v  \cdot \gamma + \delta,
\end{align*}
where
\begin{align*}
	\gamma & = \frac{p - q}{p - q + (p-q)(1-p-q)x_v}       \\
	\delta & = \frac{q(1 - q)x_v}{p - q + (p-q)(1-p-q)x_v}
\end{align*}
Hence MLE-Apx appears to represent some hybrid of Norm-Sub and Norm-Mul.
In evaluation, we observe that Norm-Sub and MLE-Apx give very close results, as $\gamma \sim 1$.
Furthermore, when the $f_v$ component in variance is dominated by the other component (as in Equation~\eqref{eq:var_appox}), the CLS formulation is equivalent to our MLE formulation.

\begin{figure*}[ht]
	\centering

	\subfigure[Base (Post-Pos)]{
		\includegraphics[width=0.31\textwidth]{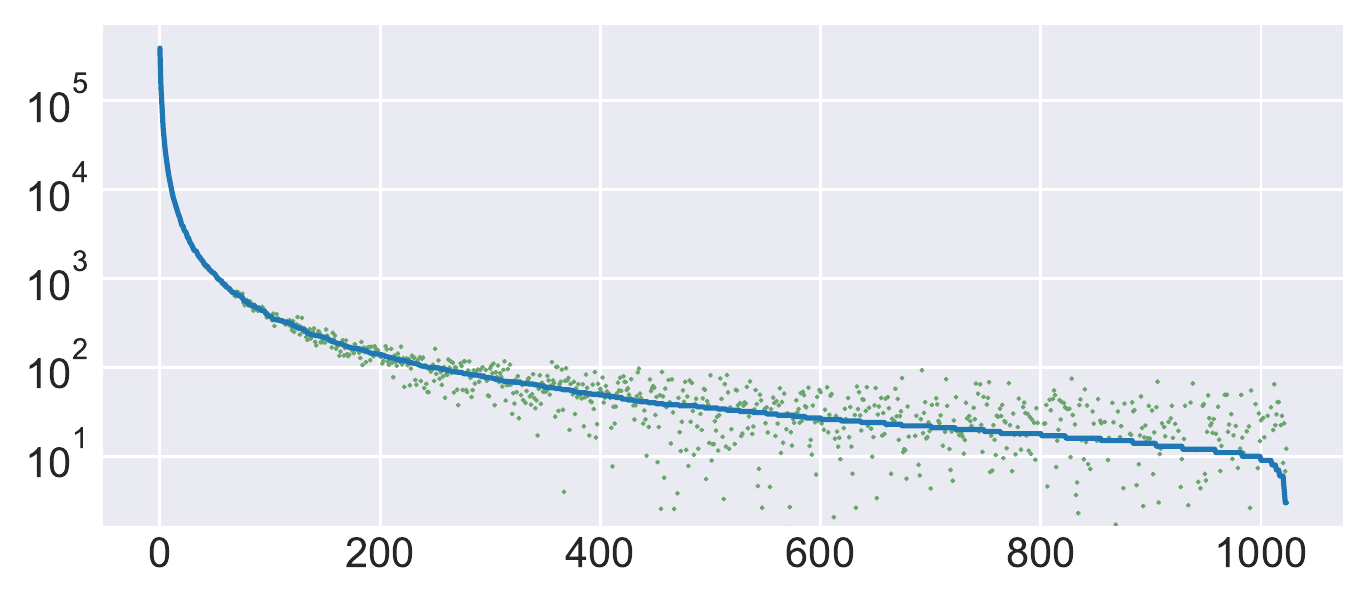}
	}
	\subfigure[Base-Pos]{
		\includegraphics[width=0.31\textwidth]{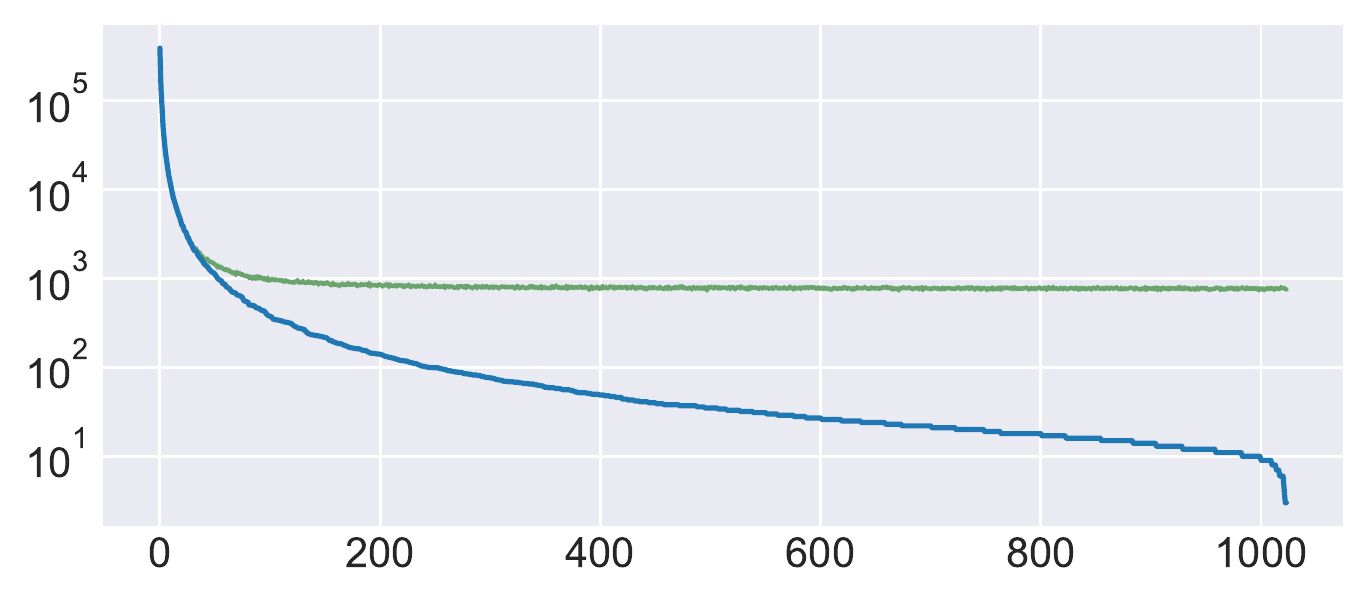}
	}
	\subfigure[Base-Cut]{
		\includegraphics[width=0.31\textwidth]{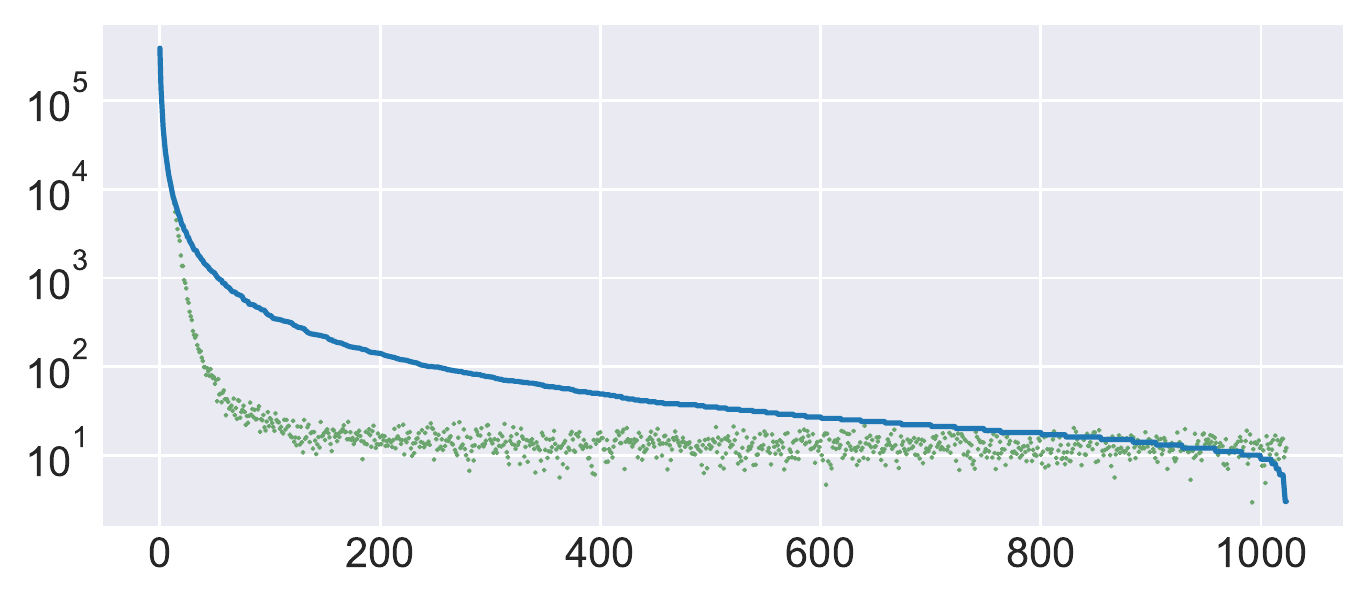}
	}
	\subfigure[Norm]{
		\includegraphics[width=0.31\textwidth]{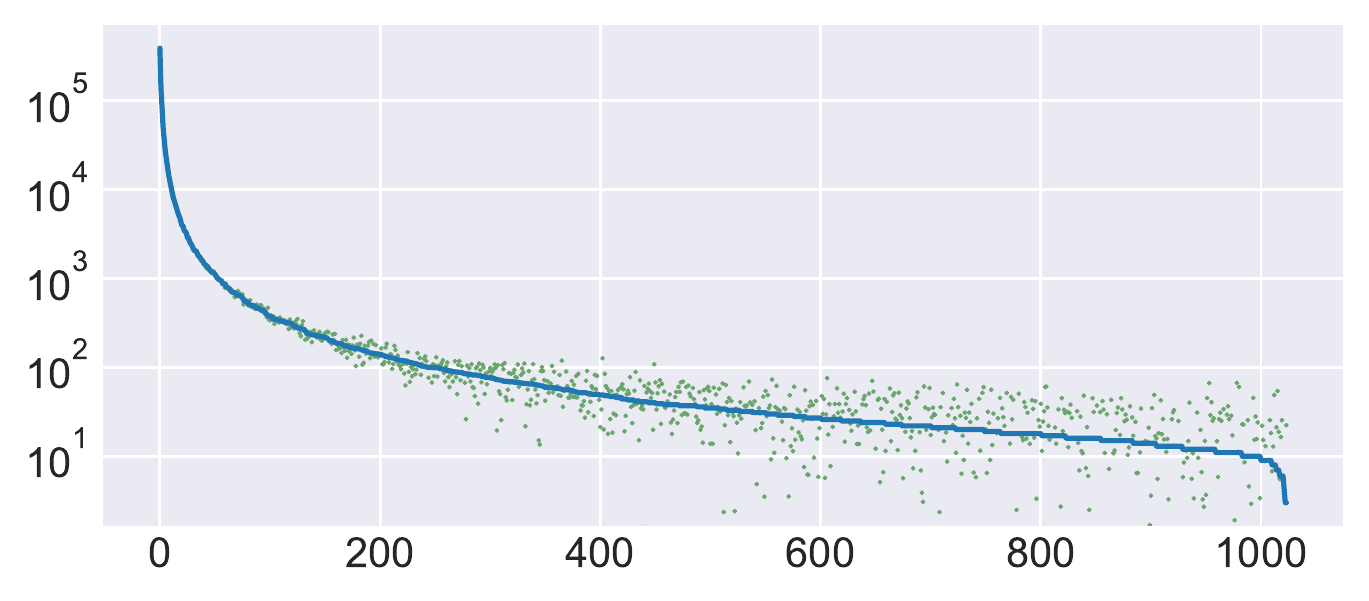}
	}
	\subfigure[Norm-Mul]{
		\includegraphics[width=0.31\textwidth]{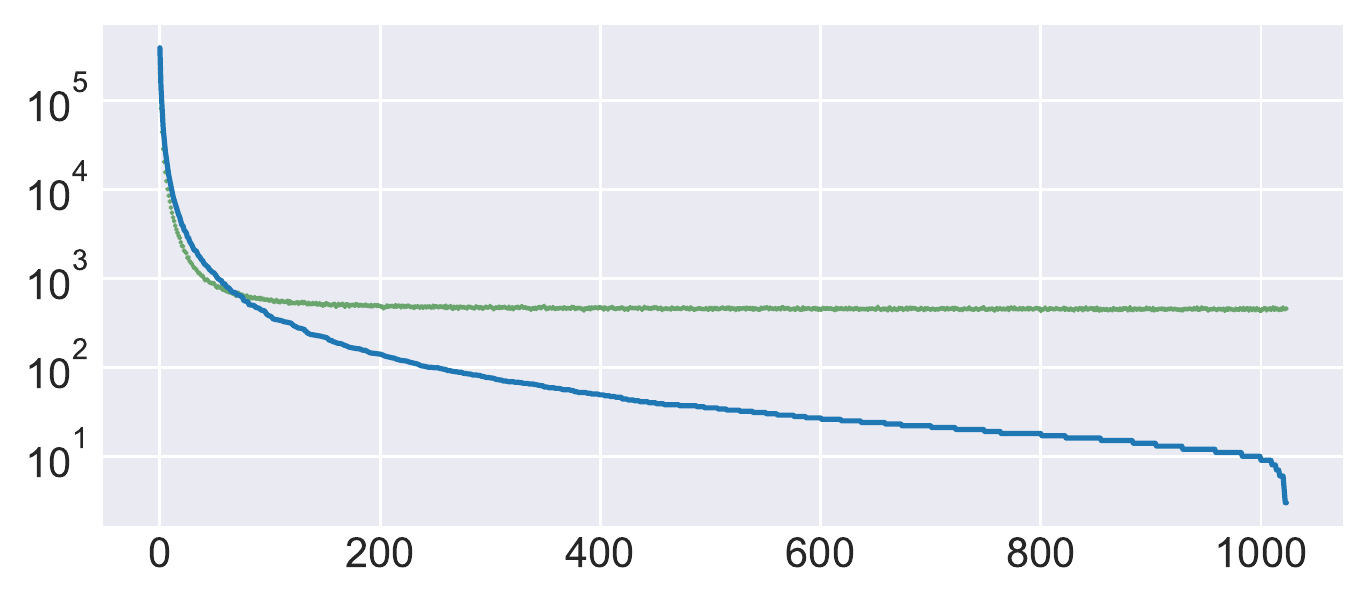}
	}
	\subfigure[Norm-Cut]{
		\includegraphics[width=0.31\textwidth]{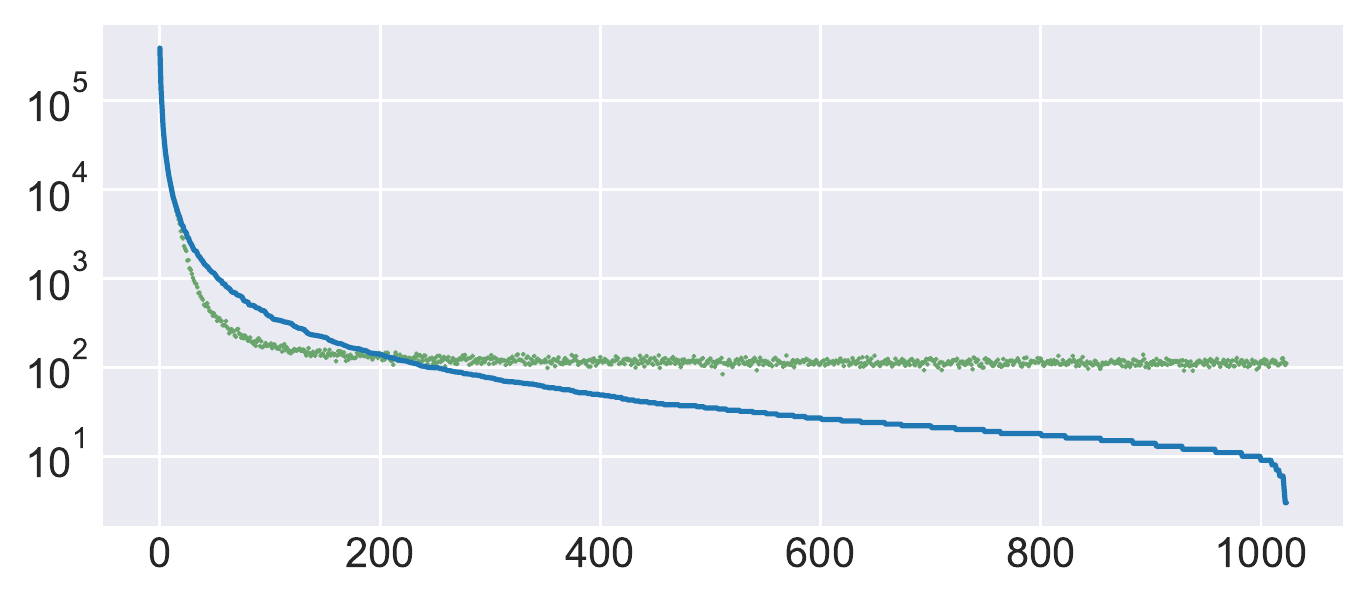}
	}
	\subfigure[Norm-Sub]{
		\includegraphics[width=0.31\textwidth]{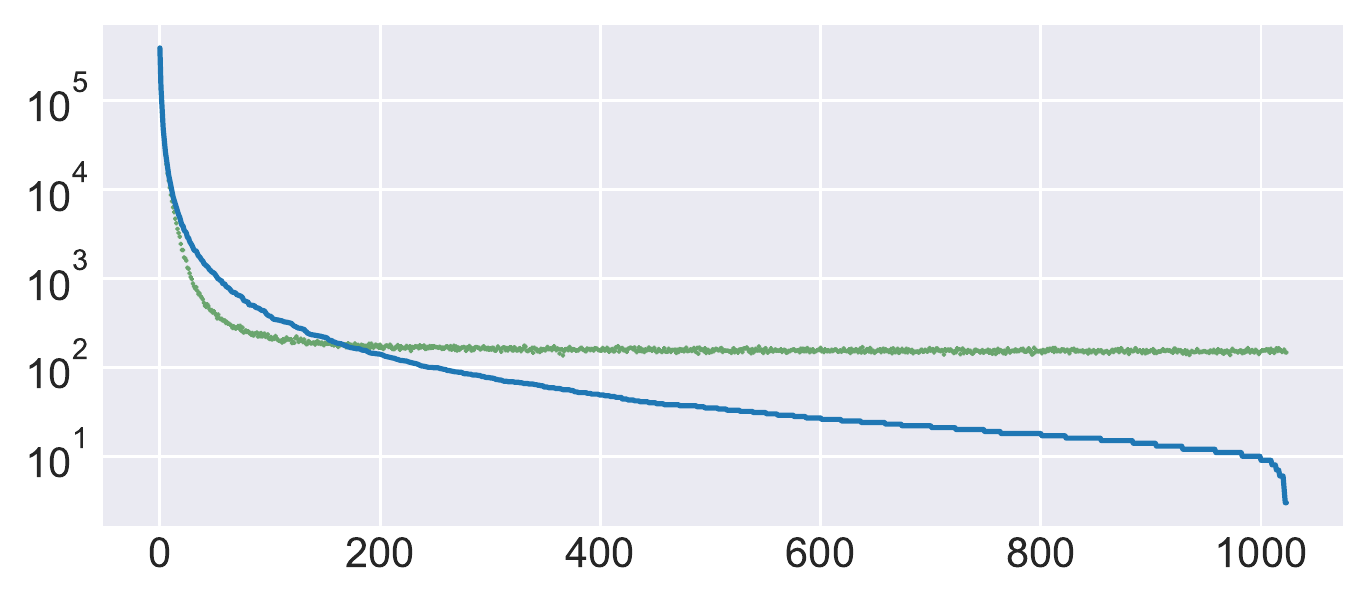}
	}
	\subfigure[Power]{
		\includegraphics[width=0.31\textwidth]{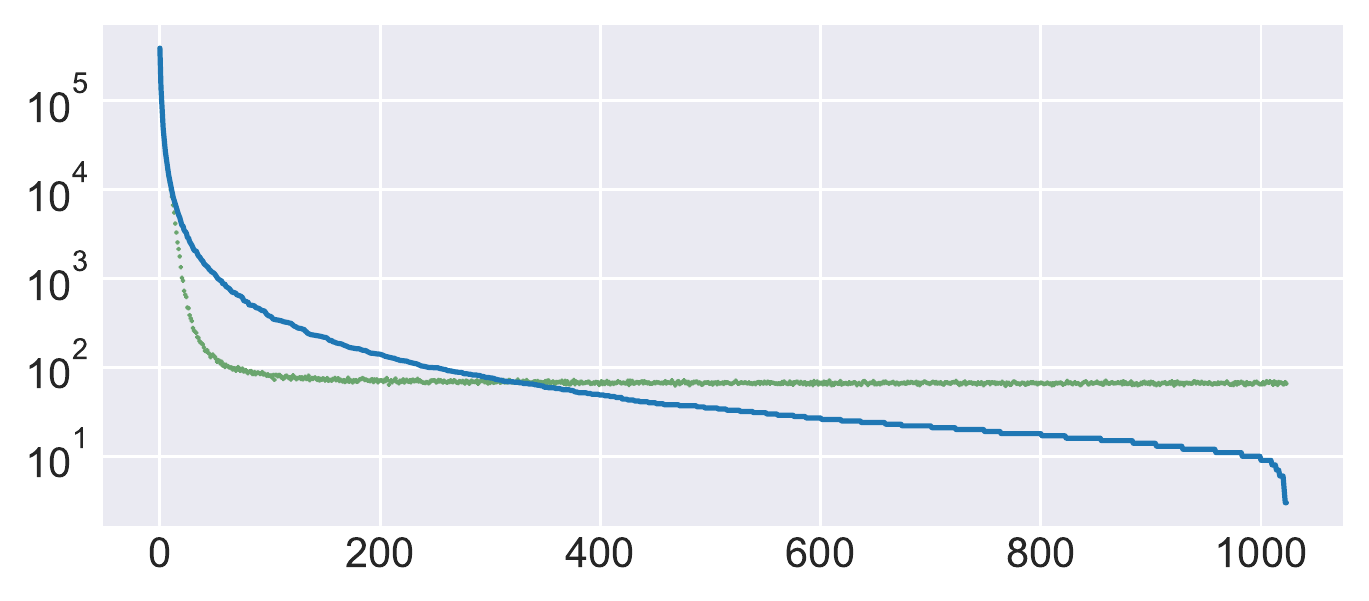}
	}
	\subfigure[PowerNS]{
		\includegraphics[width=0.31\textwidth]{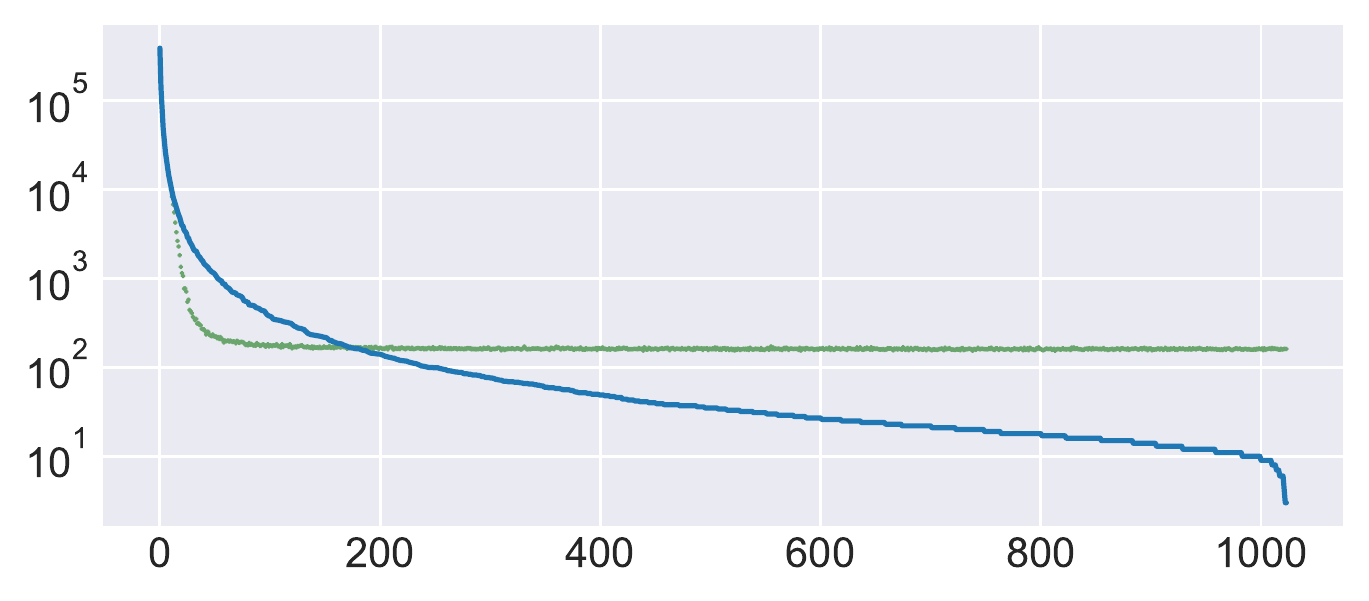}
	}
	\caption{Log-scale distribution of the Zipf's dataset fixing $\epsilon=1$, the $x$-axes indicates the sorted value index and the $y$-axes is its count.  The blue line is the ground truth; the green dots are estimations by different methods.}
	\label{fig:dist_est}
\end{figure*}

\begin{figure*}[ht]
	\centering

	\subfigure[Base (Post-Pos),  bias sum: $-1405$]{
		\includegraphics[width=0.31\textwidth]{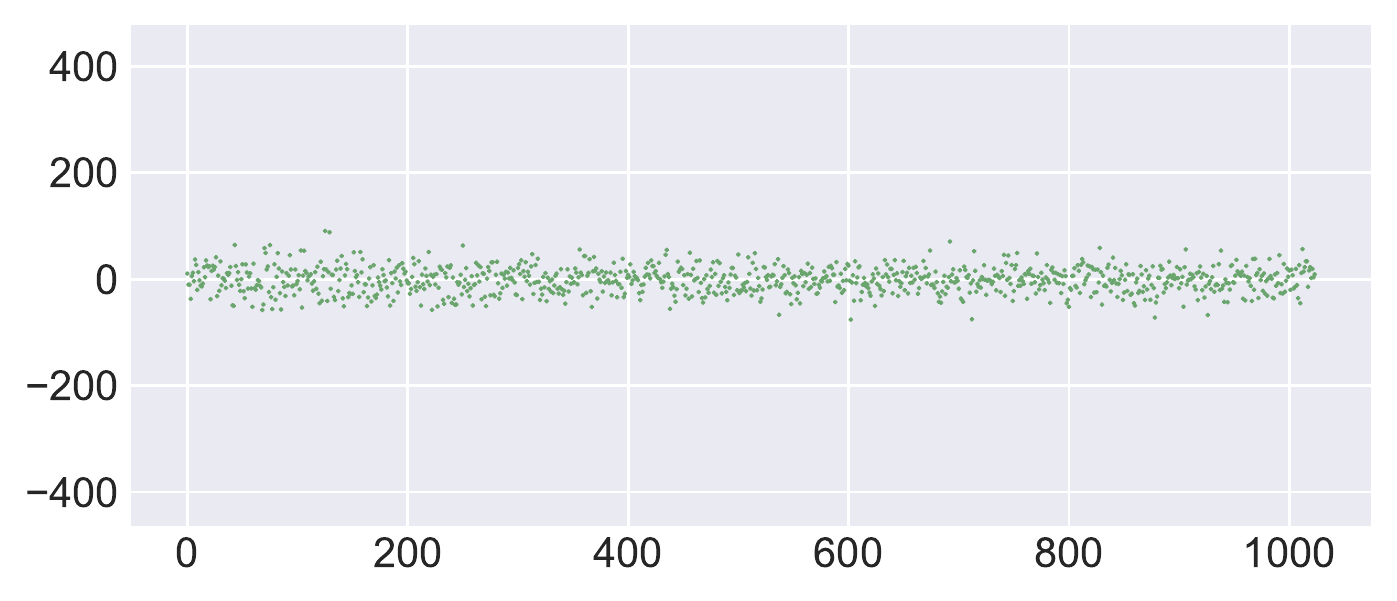}
	}
	\subfigure[Base-Pos,  bias sum: $711932$]{
		\includegraphics[width=0.31\textwidth]{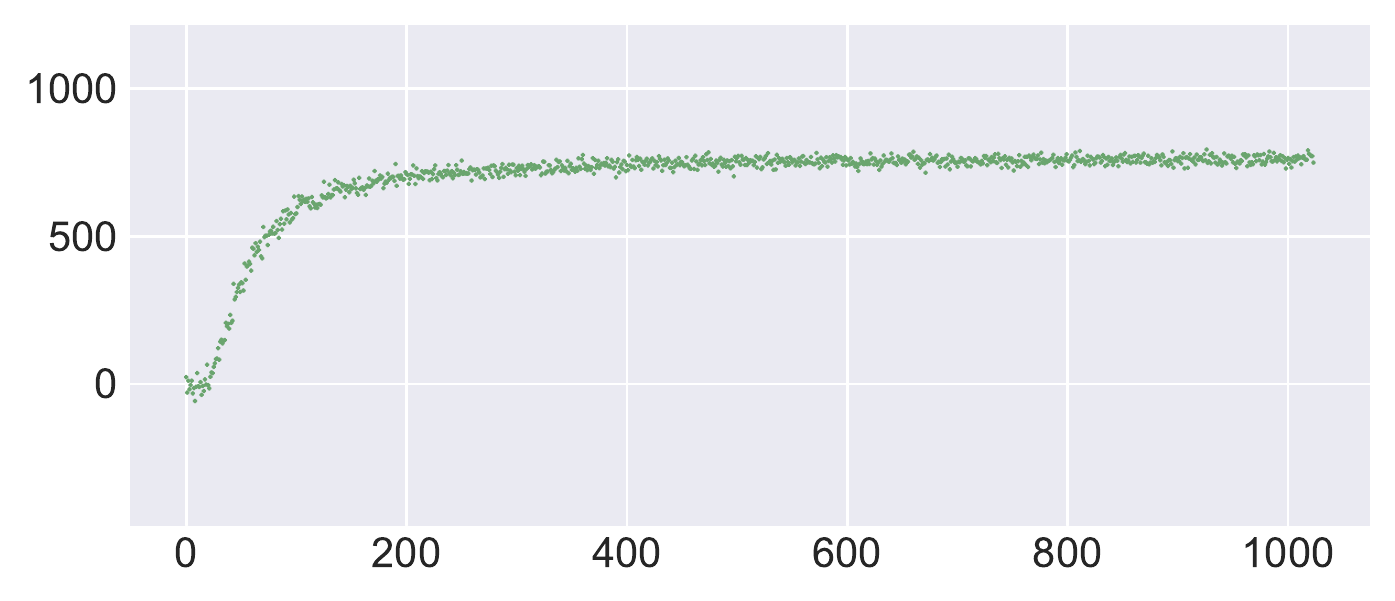}
	}
	\subfigure[Base-Cut,  bias sum: $-137449$]{
		\includegraphics[width=0.31\textwidth]{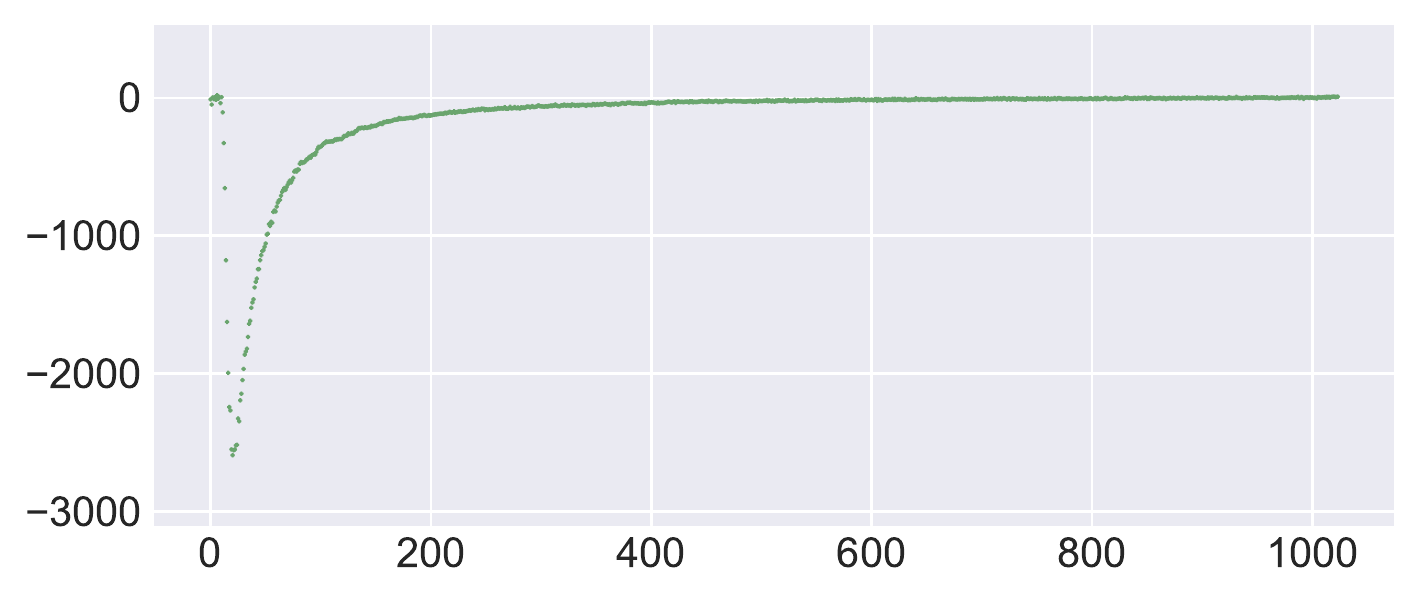}
	}
	\subfigure[Norm,  bias sum: $0$]{
		\includegraphics[width=0.31\textwidth]{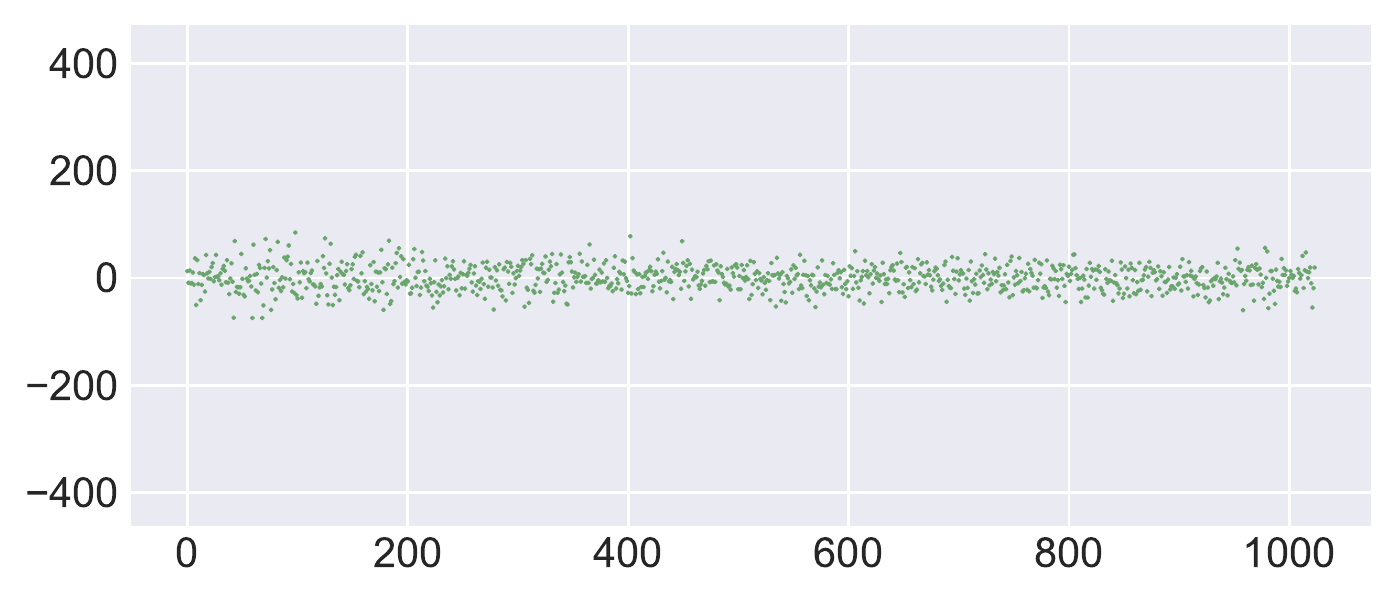}
	}
	\subfigure[Norm-Mul,  bias sum: $0$]{
		\includegraphics[width=0.31\textwidth]{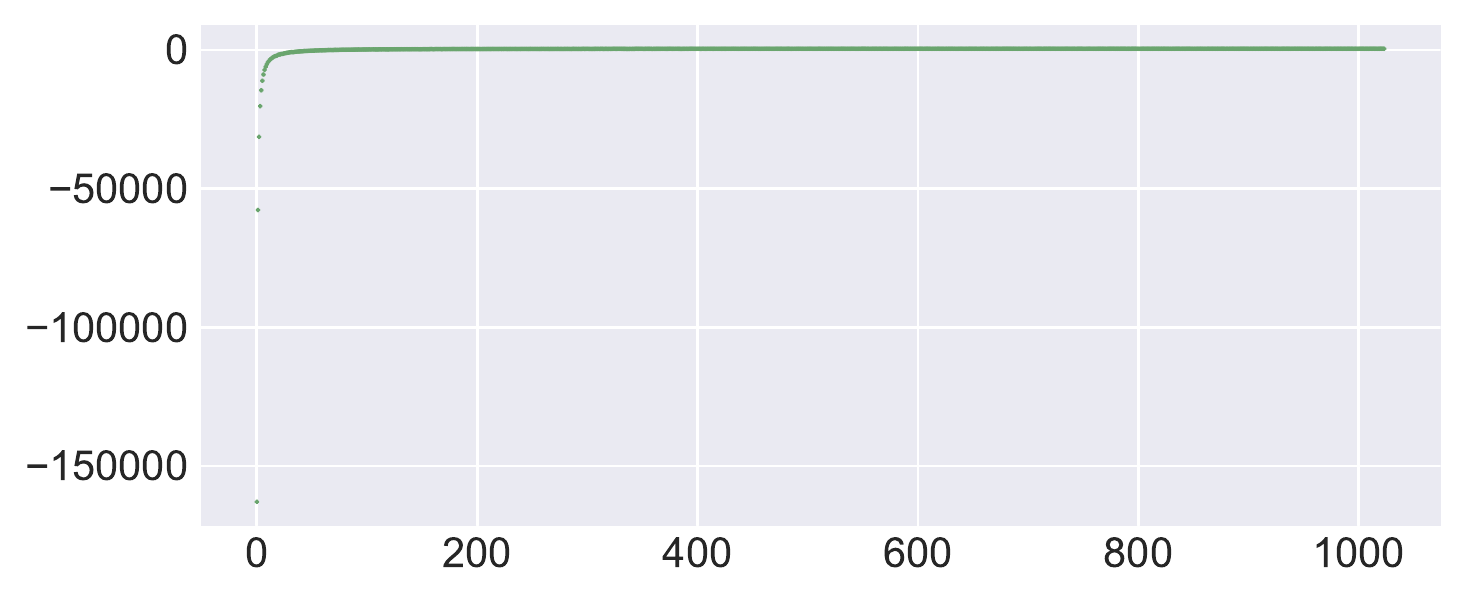}
	}
	\subfigure[Norm-Cut,  bias sum: $0$]{
		\includegraphics[width=0.31\textwidth]{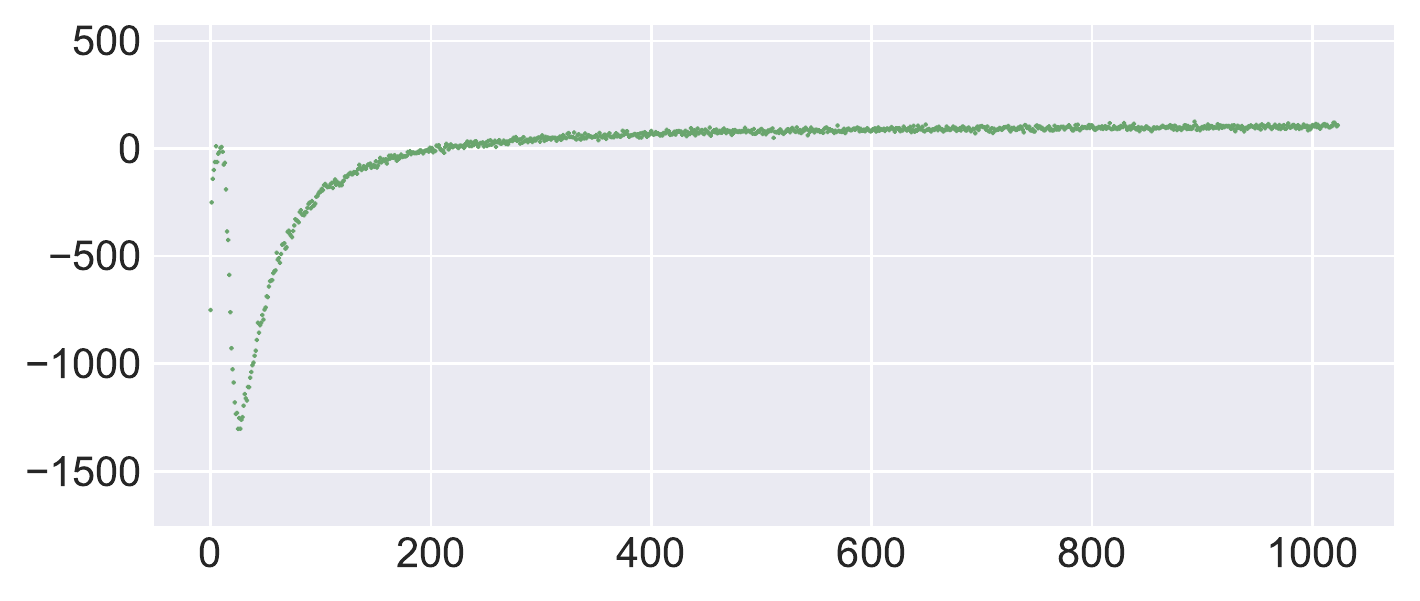}
	}
	\subfigure[Norm-Sub,  bias sum: $0$]{
		\includegraphics[width=0.31\textwidth]{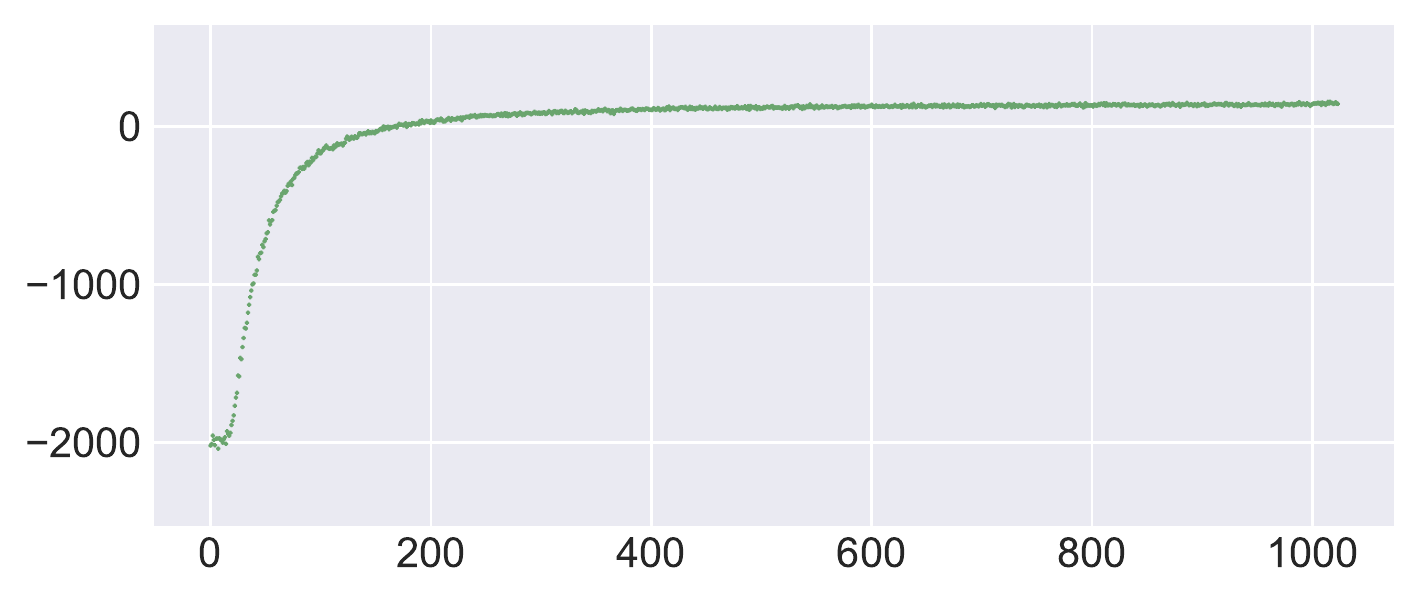}
	}
	\subfigure[Power,  bias sum: $-96332$]{
		\includegraphics[width=0.31\textwidth]{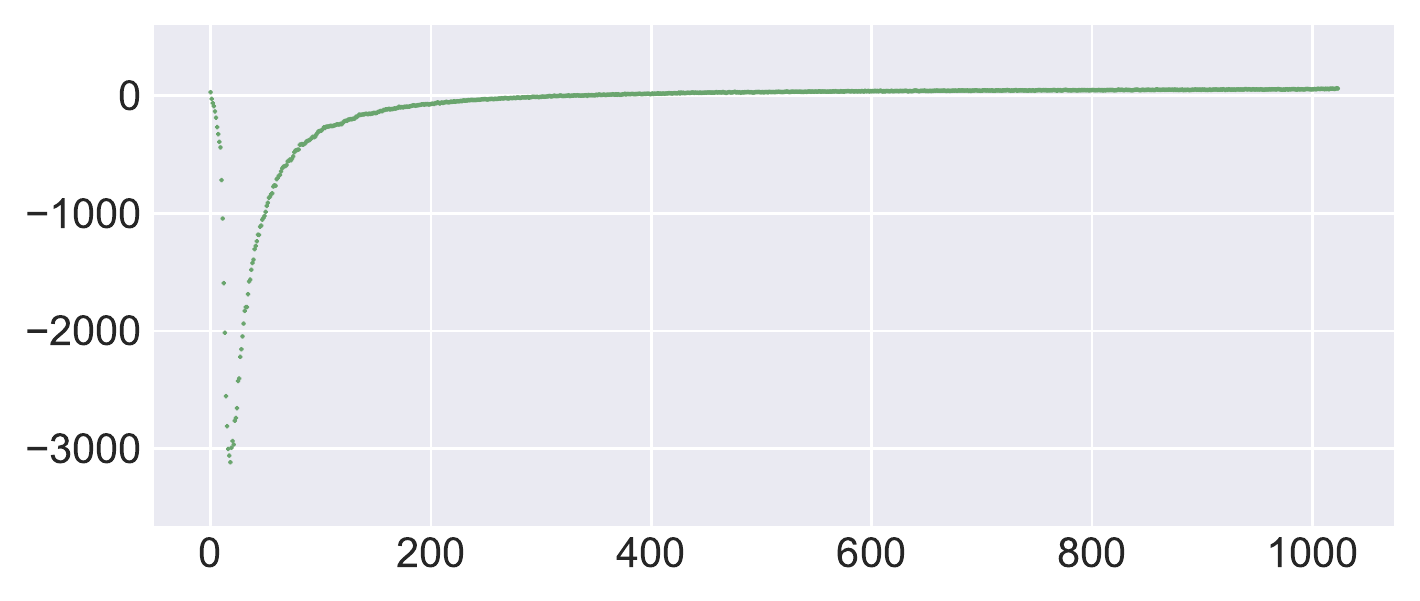}
	}
	\subfigure[PowerNS,  bias sum: $0$]{
		\includegraphics[width=0.31\textwidth]{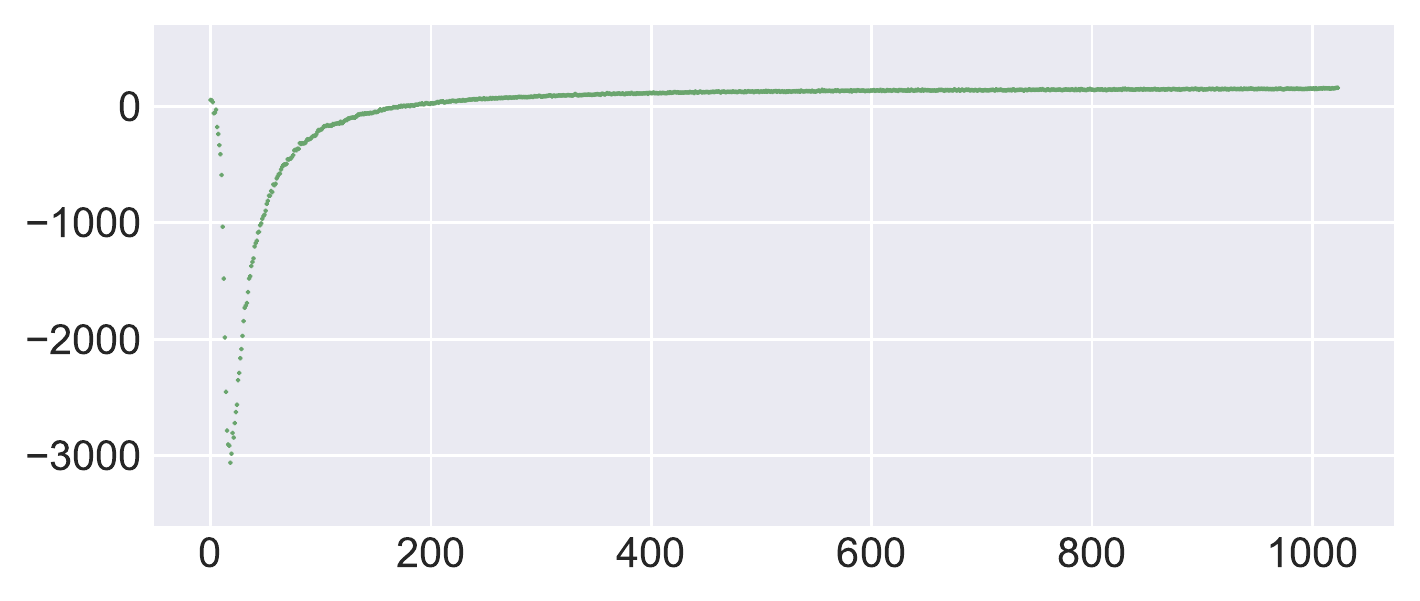}
	}
	\caption{Bias of count estimation for the Zipf's dataset fixing $\epsilon=1$.}
	\label{fig:dist_diff}
\end{figure*}

\begin{figure*}[ht]
	\centering

	\subfigure[Base (Post-Pos)]{
		\includegraphics[width=0.31\textwidth]{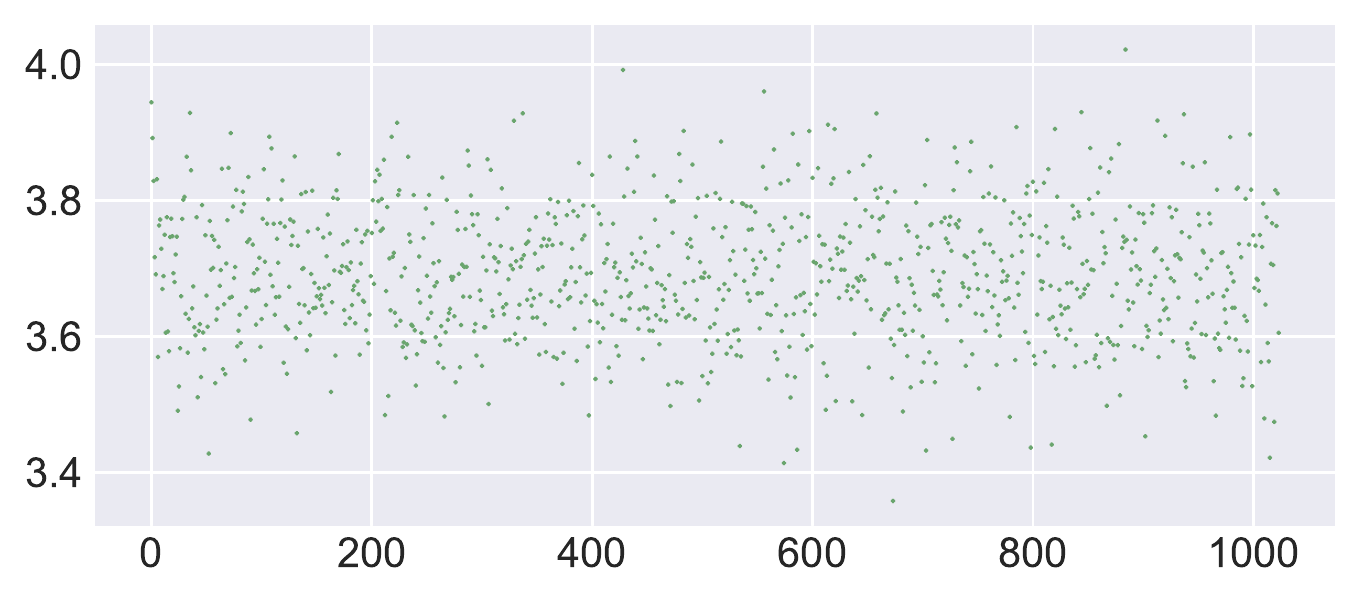}
	}
	\subfigure[Base-Pos]{
		\includegraphics[width=0.31\textwidth]{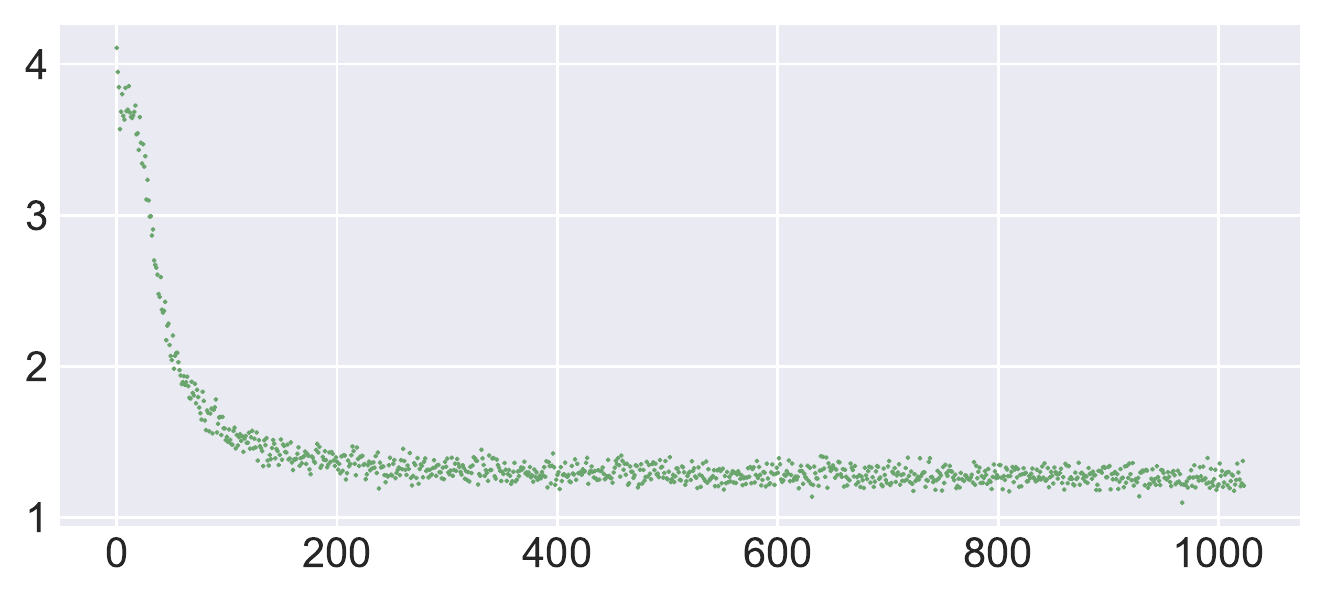}
	}
	\subfigure[Base-Cut]{
		\includegraphics[width=0.31\textwidth]{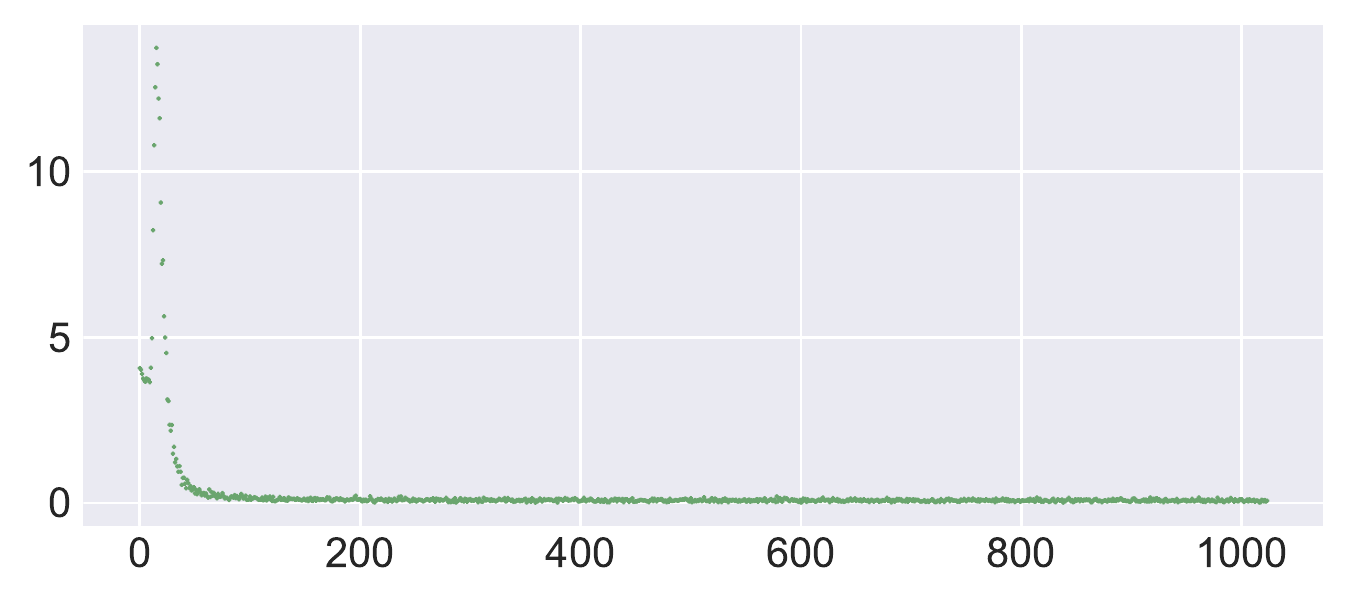}
	}
	\subfigure[Norm]{
		\includegraphics[width=0.31\textwidth]{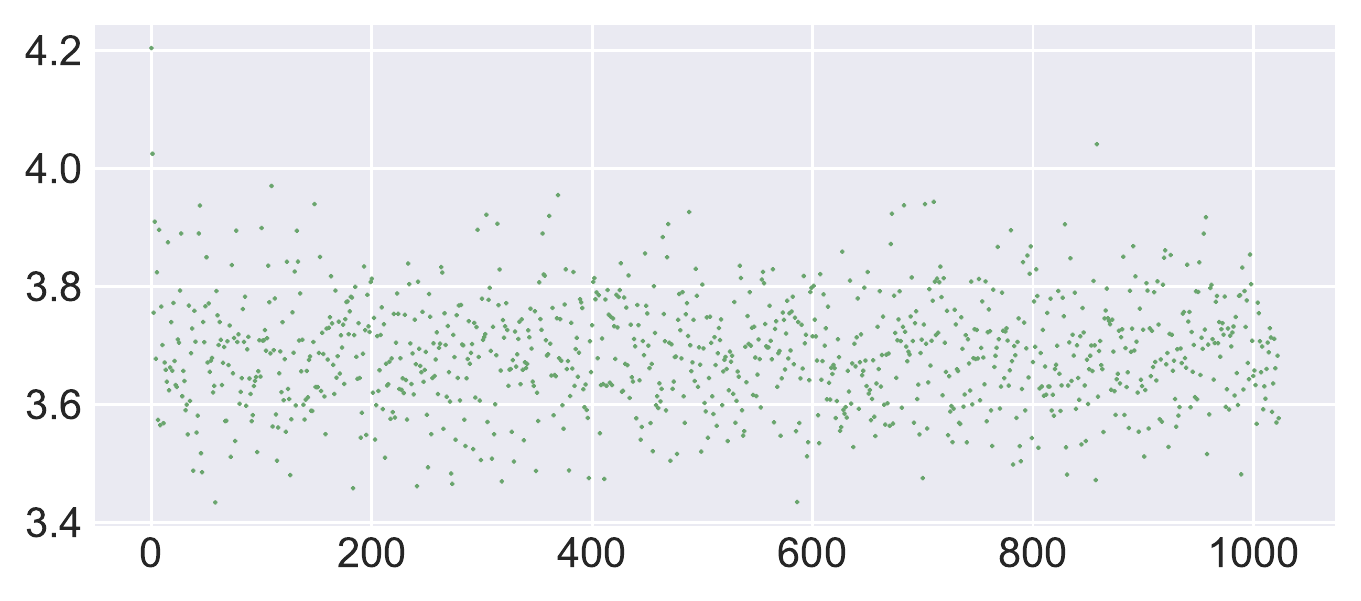}
	}
	\subfigure[Norm-Mul]{
		\includegraphics[width=0.31\textwidth]{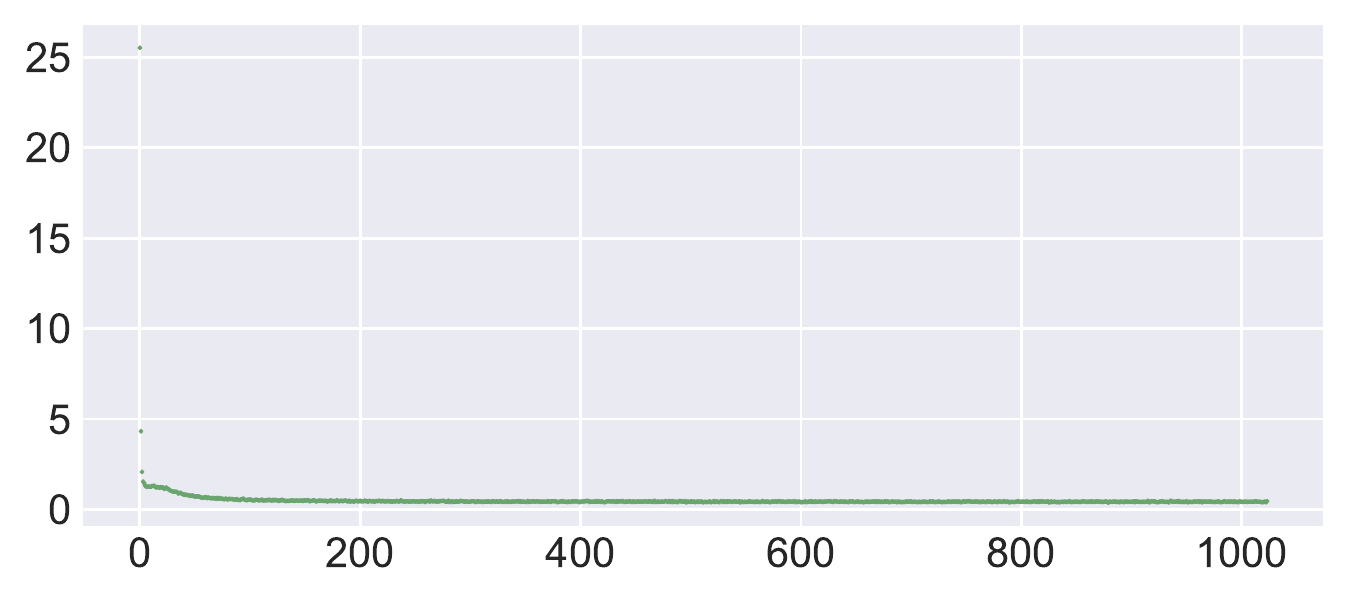}
	}
	\subfigure[Norm-Cut]{
		\includegraphics[width=0.31\textwidth]{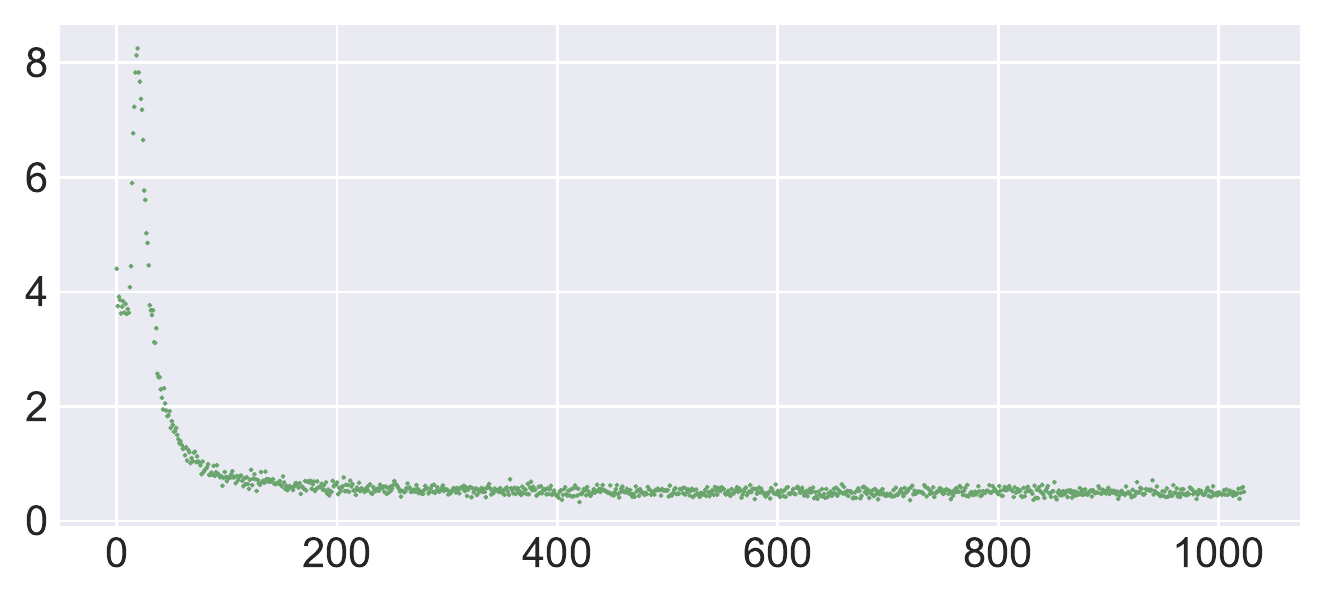}
	}
	\subfigure[Norm-Sub]{
		\includegraphics[width=0.31\textwidth]{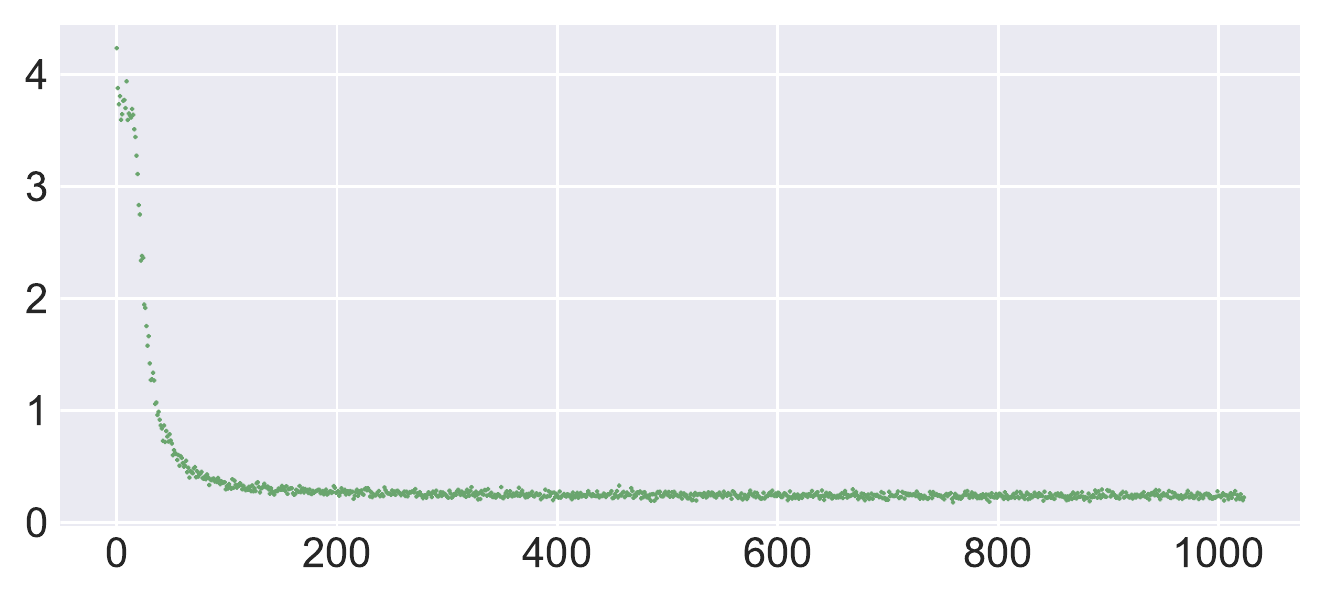}
	}
	\subfigure[Power]{
		\includegraphics[width=0.31\textwidth]{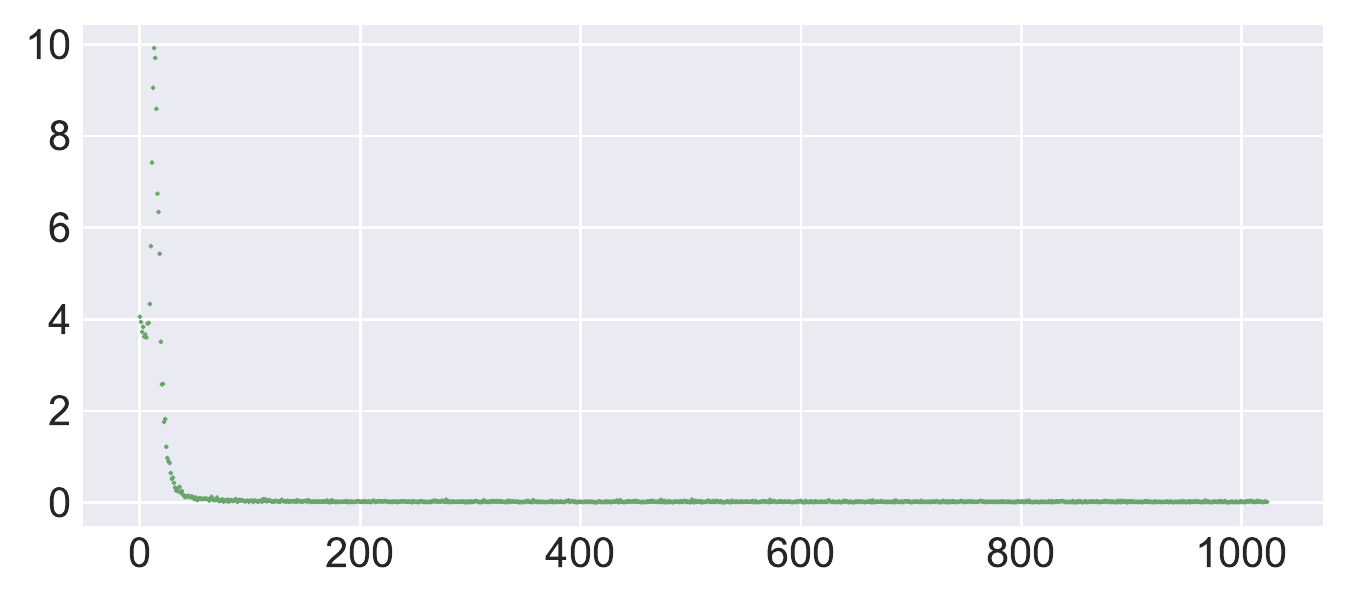}
	}
	\subfigure[PowerNS]{
		\includegraphics[width=0.31\textwidth]{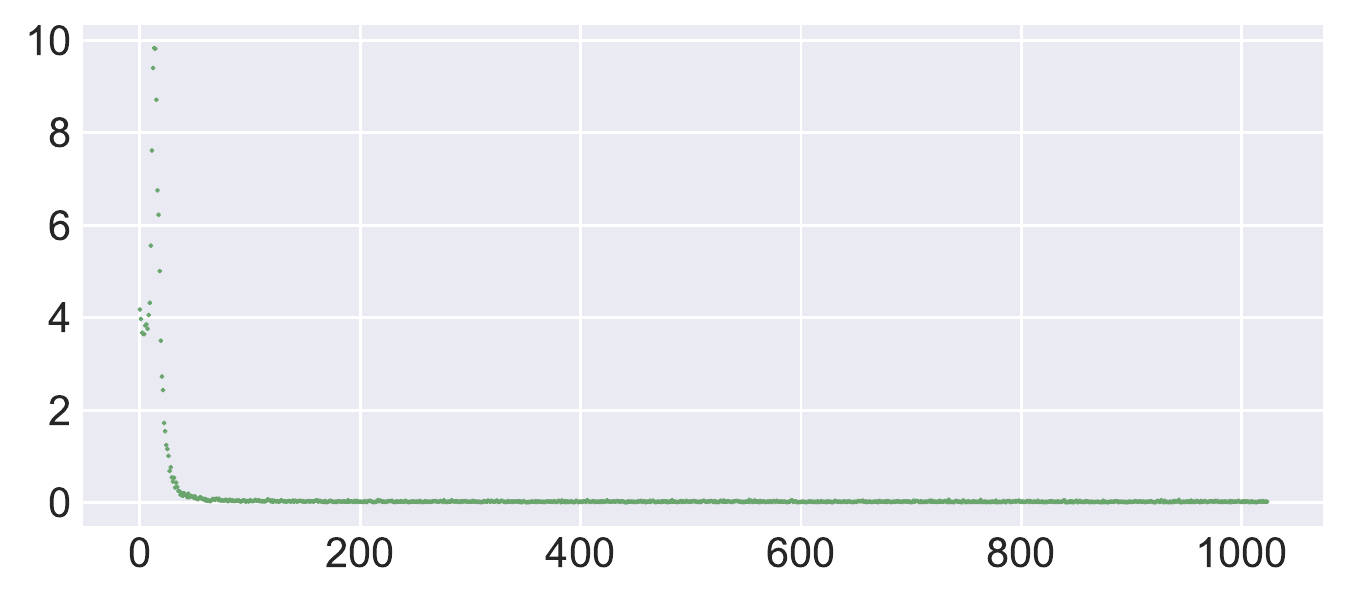}
	}
	\caption{Variance of count estimation of the Zipf's dataset fixing $\epsilon=1$.  The $y$-axes are scaled down by $n=10^6$ (a value $a$ in the figure represents $a\cdot 10^6$).}
	\label{fig:dist_var}
\end{figure*}

\subsection{Least Expected Square Error}

Jia et al.~\cite{infocom:JiaG18} proposed a method in which one first assumes that the data follows some type of distribution (but the parameters are unknown), then uses the estimates to fit the parameters of the distribution, and finally updates the estimates that achieve expected least square.

\begin{itemize}[leftmargin=*]
	\item
	      \underline{\textbf{Power}}:
	      \textit{Fit a distribution, and then minimize the expected squared error.}
\end{itemize}
Formally, for each value $v$, the estimate $\tilde{f}_v$ given by \fo is regarded as the addition of two parts: the true frequency $f_v$ and noise following the normal distribution (as shown in Equation~\eqref{eq:normal_approx}).  The method then finds $f'_v$ that minimizes $\EV{(f_v - f'_v)^2|\tilde{f}_v}$.  To solve this problem, the authors estimate the true distribution $f_v$ from the estimates $\mathbf{\tilde{f}}$ (where $\mathbf{\tilde{f}}$ is the vector of the $\tilde{f}_v$'s).

In particular, it is assume in~\cite{infocom:JiaG18} that the distribution follows Power-Law or Gaussian.  The distributions can be determined by one or two parameters, which can be fitted from the estimation $\mathbf{\tilde{f}}$.
Given $\Pr{x}$ as the probability $f_v = x$ from the fitted distribution, and
$\Pr{ x \sim \mathcal{N}(0, \sigma)}$ as the pdf of $x$ drawn from the Normal distribution with $0$ mean and standard deviation $\sigma$ (as in Equation~\eqref{eq:normal_approx}), one can then minimize the objective.  Specifically, for each value $v\in\Domain$, output
\begin{align}
	f'_v & =\int_0^1\frac{\Pr{(\tilde{f}_v - x) \sim \mathcal{N}(0, \sigma)}\cdot \Pr{x}\cdot x}{\int_0^1 \Pr{ (\tilde{f}_v - y) \sim \mathcal{N}(0, \sigma)}\cdot \Pr{y}dy}dx.\label{eq:prior}
\end{align}

We fit $\Pr{x}$ with the Power-Law distribution and call the method Power.
Using this method requires knowledge and/or assumption of the distribution to be estimated.  If there are too much noise, or the underlying distribution is different, forcing the observations to fit a distribution could lead to poor accuracy.
Moreover, this method does not ensure the frequencies sum up to 1, as Equation~\eqref{eq:prior} only considers the frequency of each value $v$ independently.
To make the result consistent, we use Norm-Sub to post-process results of Power, since Power is close to CLS, and Norm-Sub is the solution to CLS.  We call it PowerNS.

\begin{itemize}[leftmargin=*]
	\item
	      \underline{\textbf{PowerNS}}:
	      \textit{First use standard \fo, then use Power to recover the values, finally use Norm-Sub to further process the results.}
\end{itemize}

\subsection{Summary of Methods}

In summary, Norm-Sub is the solution to the Constraint Least Square (CLS) formulation to the problem.  Furthermore, when the $f_v$ component in variance is dominated by the other component (as in Equation~\eqref{eq:var_appox}), the CLS formulation is equivalent to our MLE formulation.  In that case, Norm-Sub is equivalent to MLE-Apx.

Table~\ref{tab:method_summary} gives a summary of the methods.  First of all, all of the methods preserve the frequency order of the value, i.e., $f'_{v_1}\le f'_{v_2}$ iff $\tilde{f}_{v_1}\le \tilde{f}_{v_2}$.
The methods can be classifies into three classes: First, enforcing non-negativity only.  Base-Pos, Post-Pos, Base-Cut, and Power fall in this category.  Second, enforcing summing-to-one only.  Only Norm is in this class.  Third, enforcing the two requirement simultaneously.  Norm-Mul, Norm-Cut, Norm-Sub, and PowerNS satisfy both requirements.

\begin{figure*}[ht]
	\centering
	\subfigure{
		\includegraphics[width=0.48\textwidth]{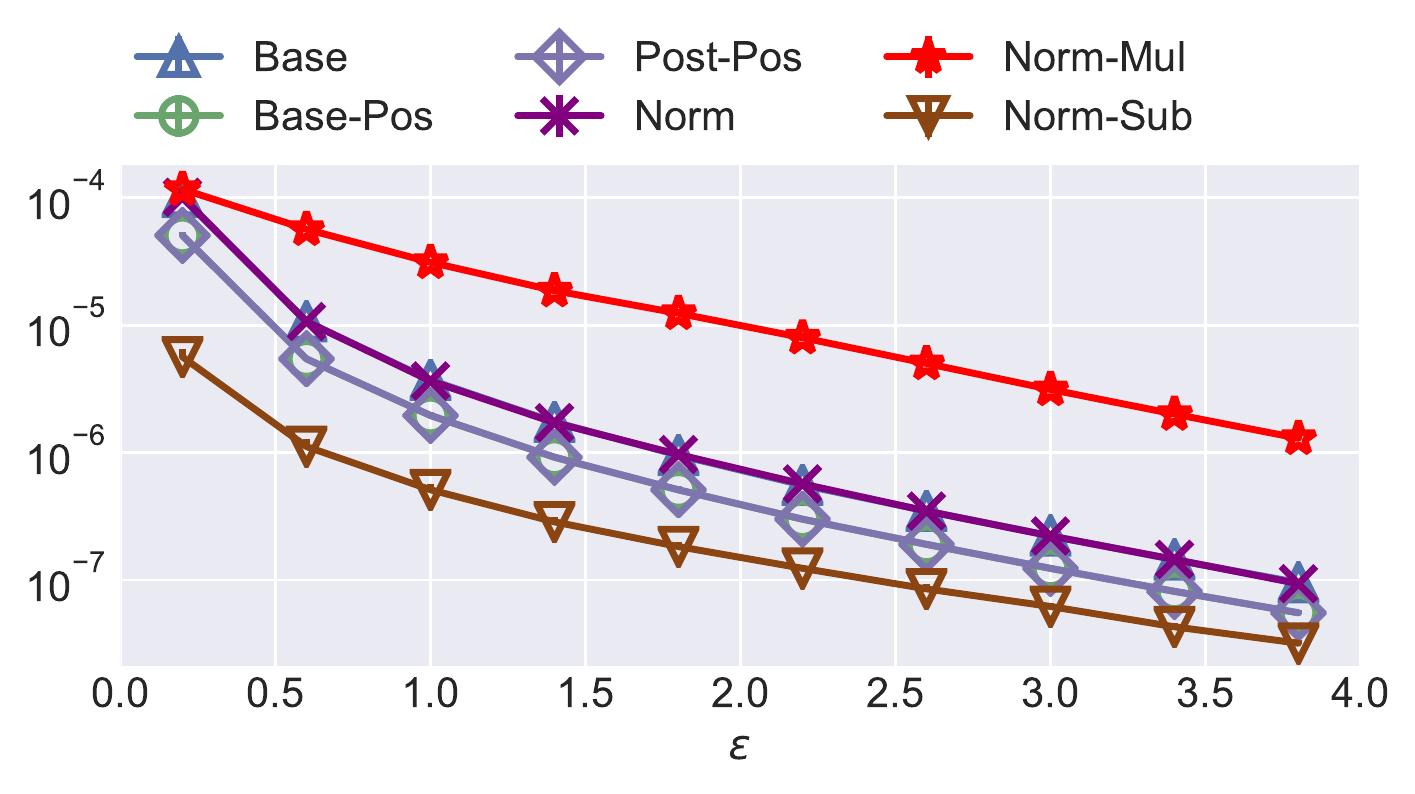}
	}
	\subfigure{
		\includegraphics[width=0.48\textwidth]{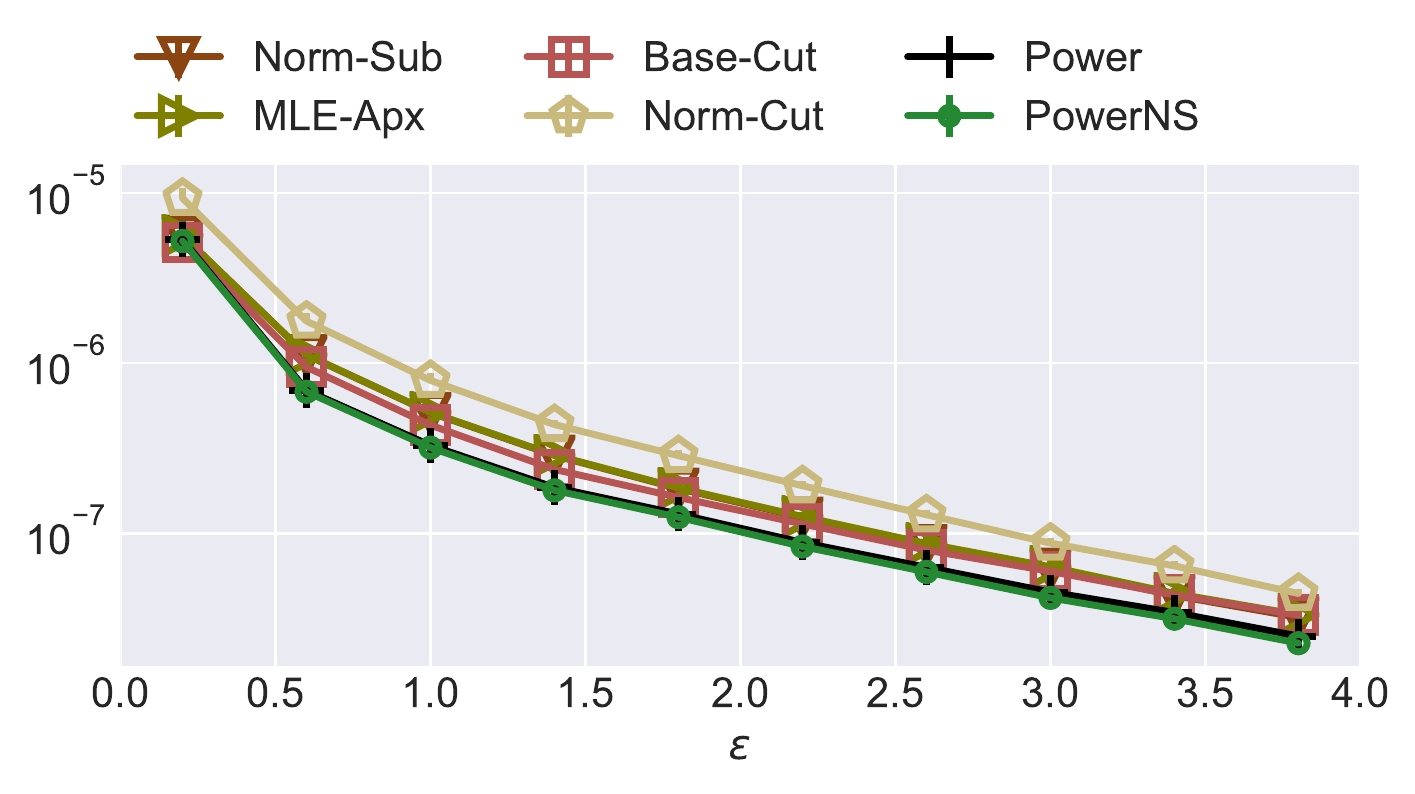}
	}
	Zipf's

	\subfigure{
		\includegraphics[width=0.48\textwidth]{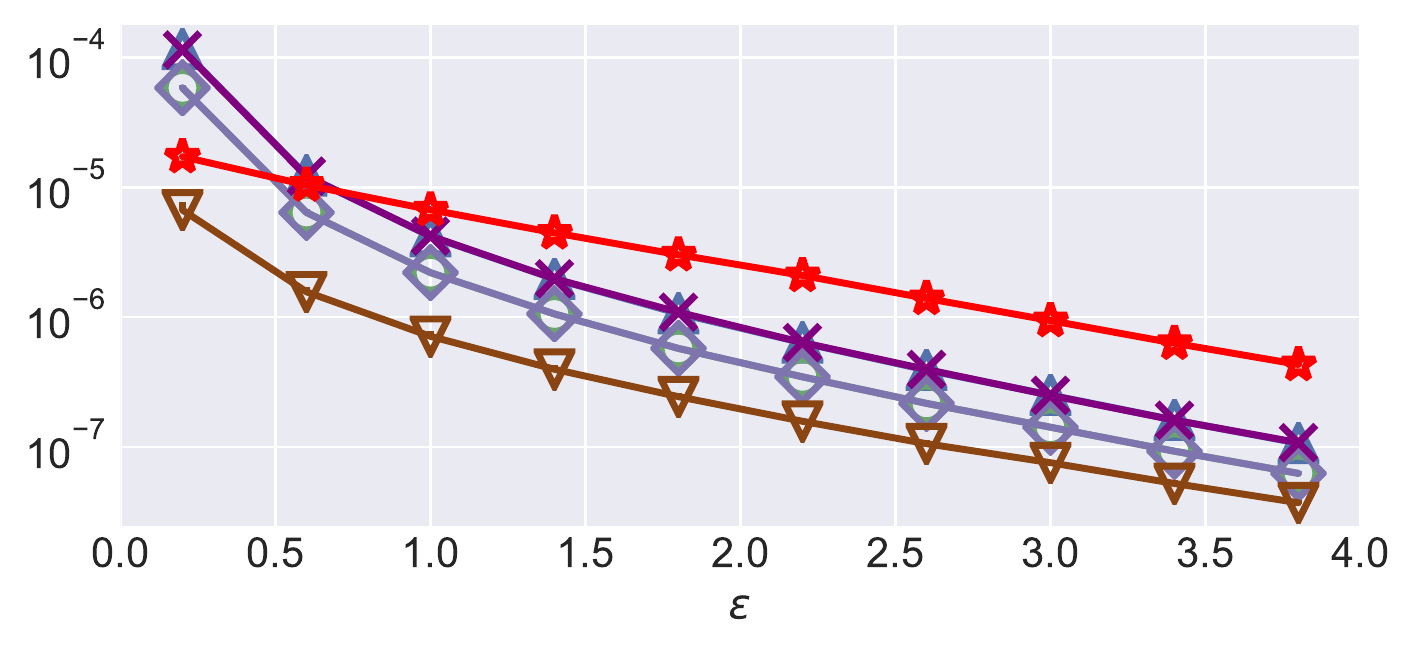}
	}
	\subfigure{
		\includegraphics[width=0.48\textwidth]{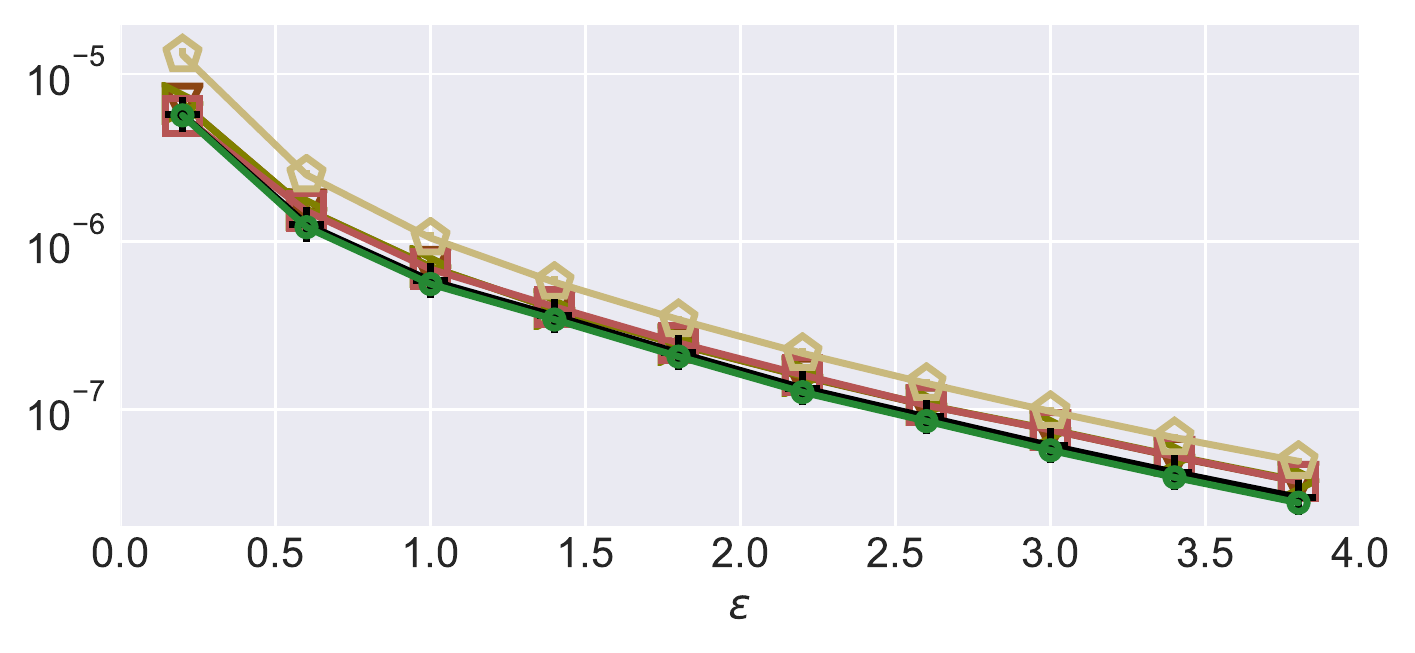}
	}
	Emoji
	\caption{MSE results on full-domain estimation, varying $\epsilon$ from $0.2$ to $4$.}
	\label{fig:eps_full}
\end{figure*}

\section{Evaluation}
\label{sec:exp}
As we are optimizing multiple utility metrics together, it is hard to theoretically compare different methods.  In this section, we run experiments to empirically evaluate these methods.

At the high level, our evaluations show that different methods perform differently in different settings, and to achieve the best utility, it may or may not be necessary to exploit all the consistency constraints.
As a result, we conclude that for full-domain query, Base-Cut performs the best; for set-value query, PowerNS performs the best; and for high-frequency-value query, Norm performs the best.

\subsection{Experimental Setup}\label{ssec:experimental_setup}

\mypara{Datasets.}
We run experiments on two datasets (one synthetic and one real).
\begin{itemize}
	\item Synthetic Zipf's distribution with 1024 values and 1 million reports.  We use $s=1.5$ in this distribution.
	\item Emoji: The daily emoji usage data.  We use the average emoji usage of an emoji keyboard~\footnote{\url{http://www.emojistats.org/}, accessed 12/15/2019 10pm ET}, which gives the total count of $n=884427$ with $d=1573$ different emojis.
\end{itemize}

\mypara{Setup.}
The \fo protocols and post-processing algorithms are implemented in Python 3.6.6 using Numpy 1.15; and all the experiments are conducted on a PC with Intel Core i7-4790 3.60GHz and 16GB memory.  Although the post-processing methods can be applied to any \fo protocol, we focus on simulating \olh as it provides near-optimal utility with reasonable communication bandwidth.

\mypara{Metrics.}
We evaluate three scenarios 1) estimate the frequency of every value in the domain (full-domain), 2) estimate the aggregate frequencies of a subset of values (set-value), and 3) estimate the frequencies of the most frequent values (frequent-value).

We use the metrics of \emph{Mean of Squared Error} (MSE).  MSE measures the mean of squared difference between the estimate and the ground truth for each (set of) value.  For full-domain, we compute
\begin{align*}
	\text{MSE} = \frac{1}{d} \sum_{v \in D} (f_v - f'_v)^2.
\end{align*}
For frequent-value, we consider the top $k$ values with highest $f_v$ instead of the whole domain $D$; and for set-value, instead of measuring errors for singletons, we measure errors for sets, that is, we first sum the frequencies for a set of values, and then measure the difference.

\begin{figure*}[ht]
	\centering
	\subfigure{
		\includegraphics[width=0.48\textwidth]{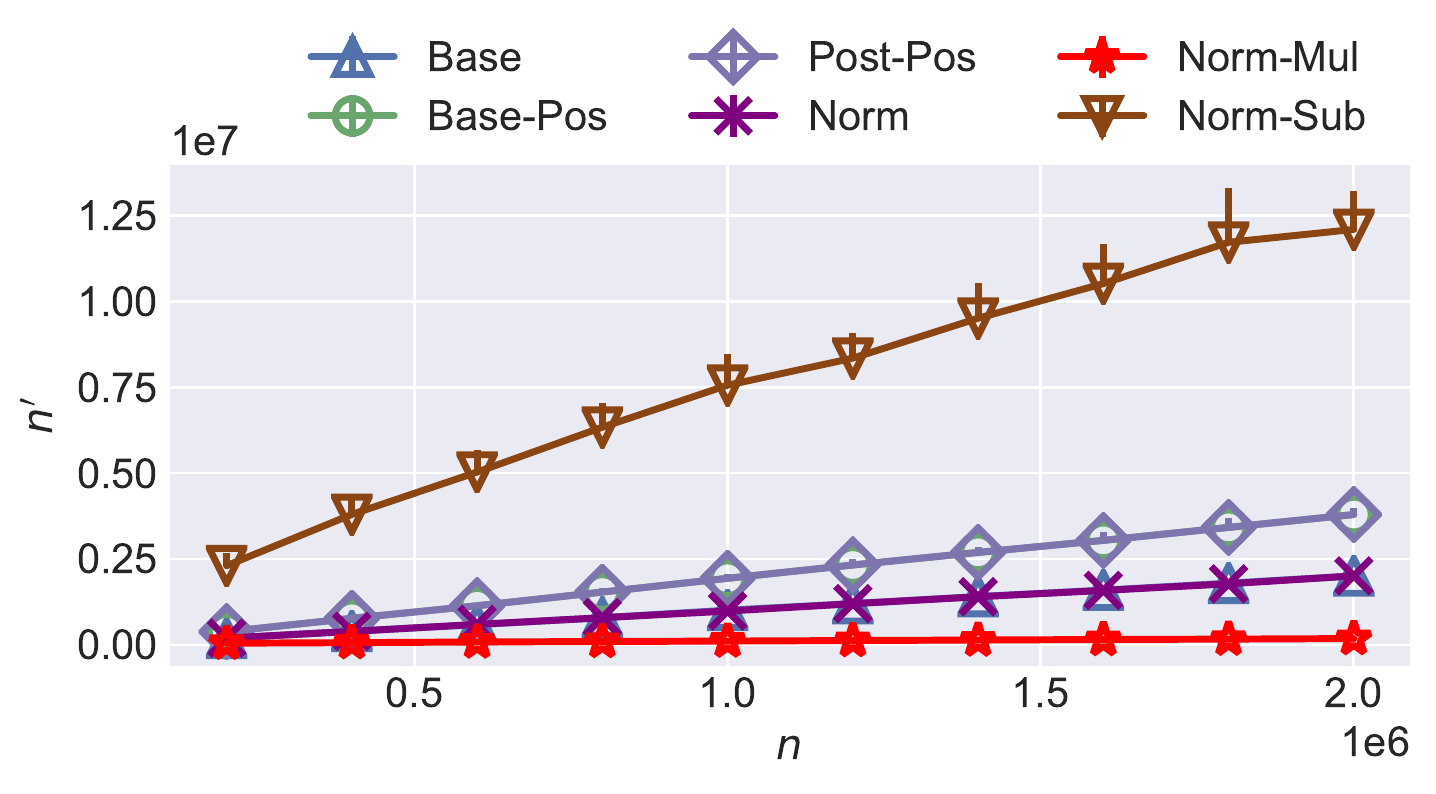}
	}
	\subfigure{
		\includegraphics[width=0.48\textwidth]{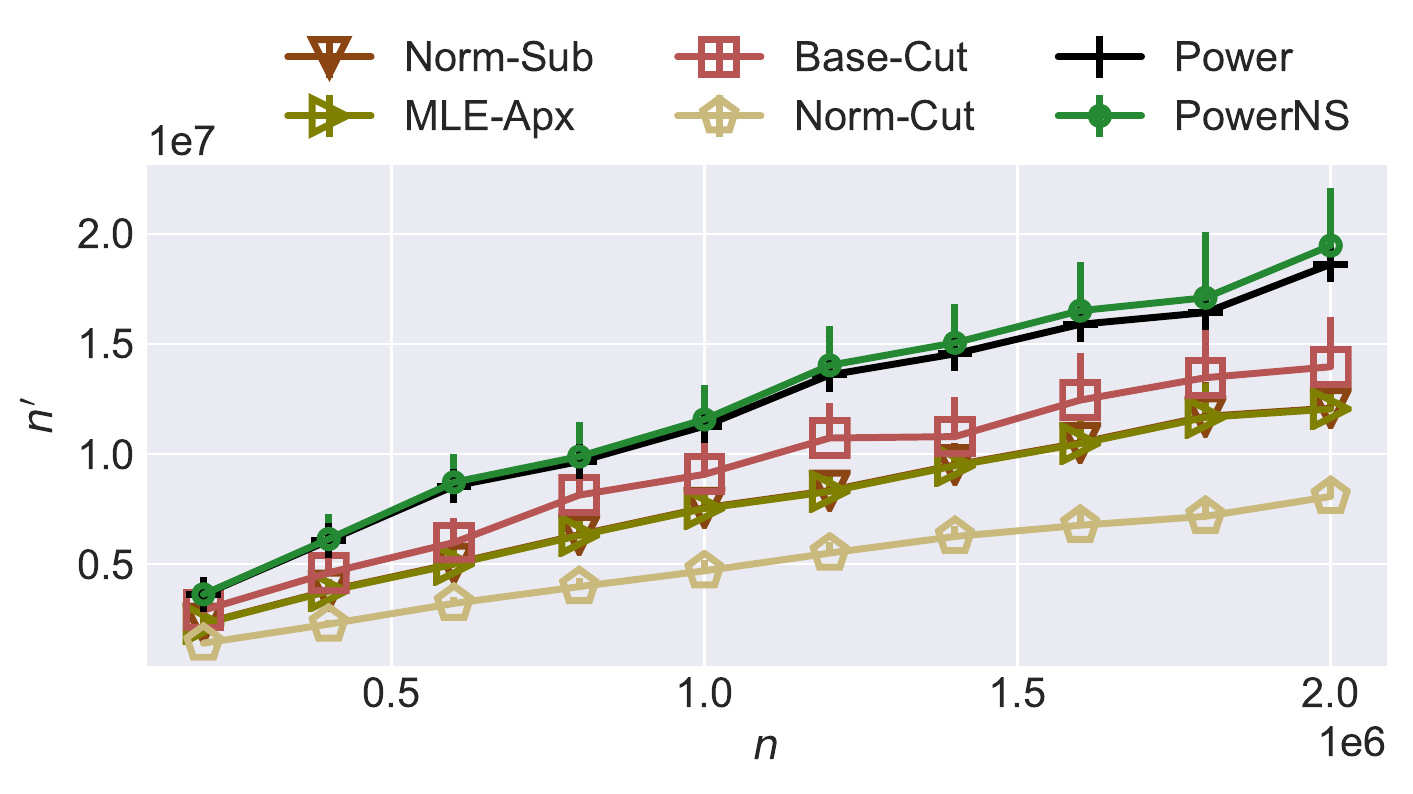}
	}
	\vspace{-0.3cm}
	\caption{MSE results on full-domain estimation on Zipfs dataset, comparing $n$ with $n'$, fixing $\epsilon=1$ while varying $n$ from $0.2\times 10^6$ to $2.0\times 10^6$.  Three pairs of methods have similar performance: Base and Norm, Base-Pos and Post-Pos, Norm-Sub and MLE-Apx.
	}
	\label{fig:eps_vary_n}
\end{figure*}

\begin{figure*}[ht]
	\centering
	\subfigure{
		\includegraphics[width=0.475\textwidth]{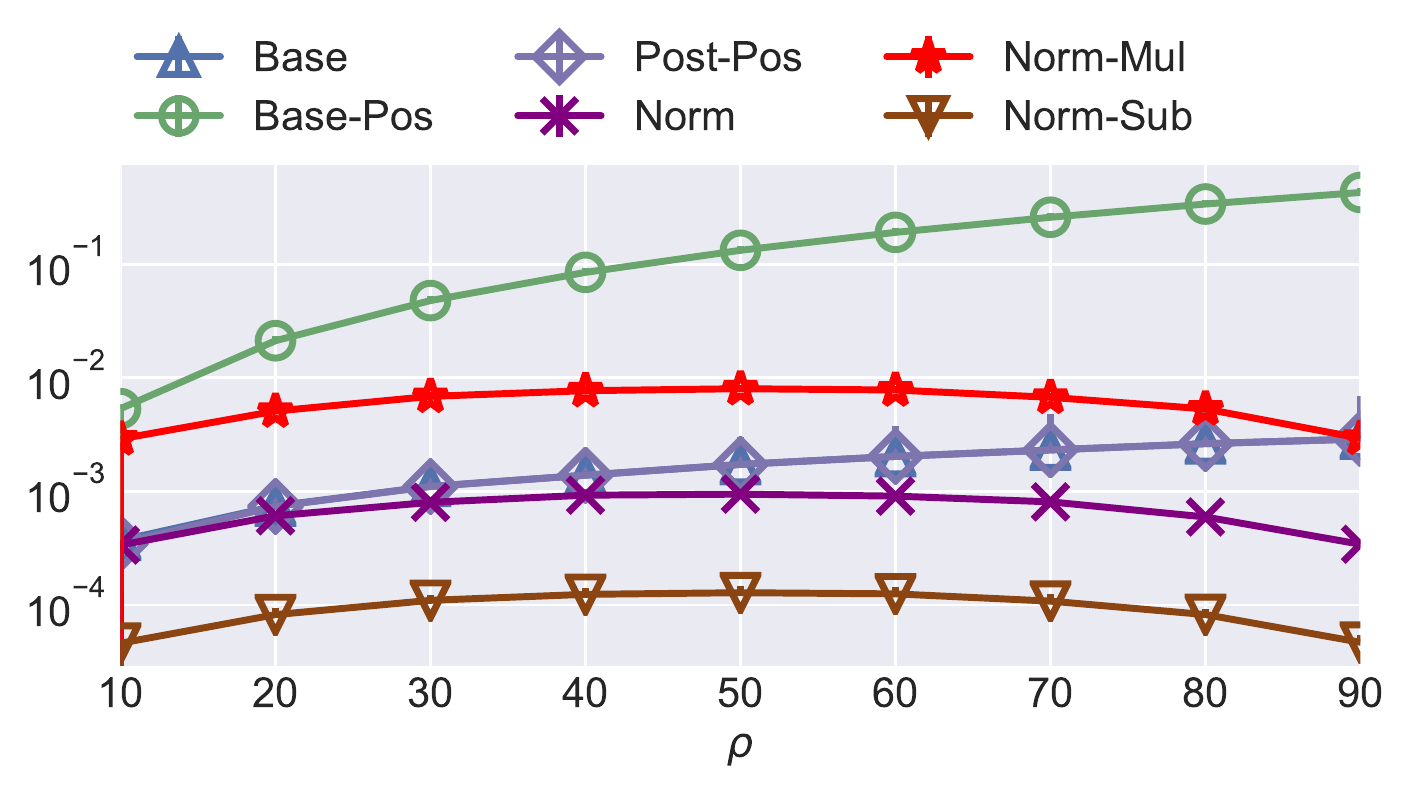}
	}
	\subfigure{
		\includegraphics[width=0.475\textwidth]{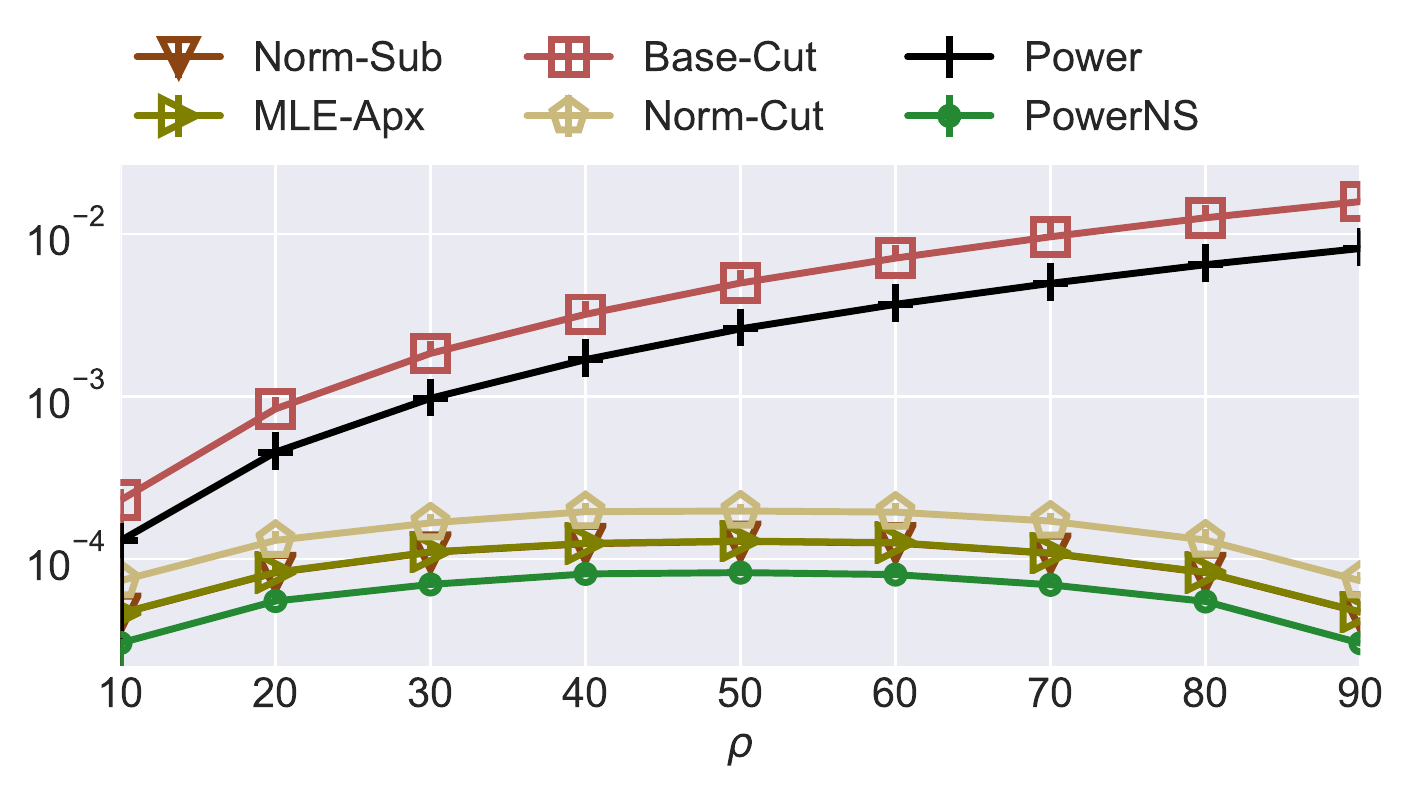}
	}
	Zipf's

	\subfigure{
		\includegraphics[width=0.475\textwidth]{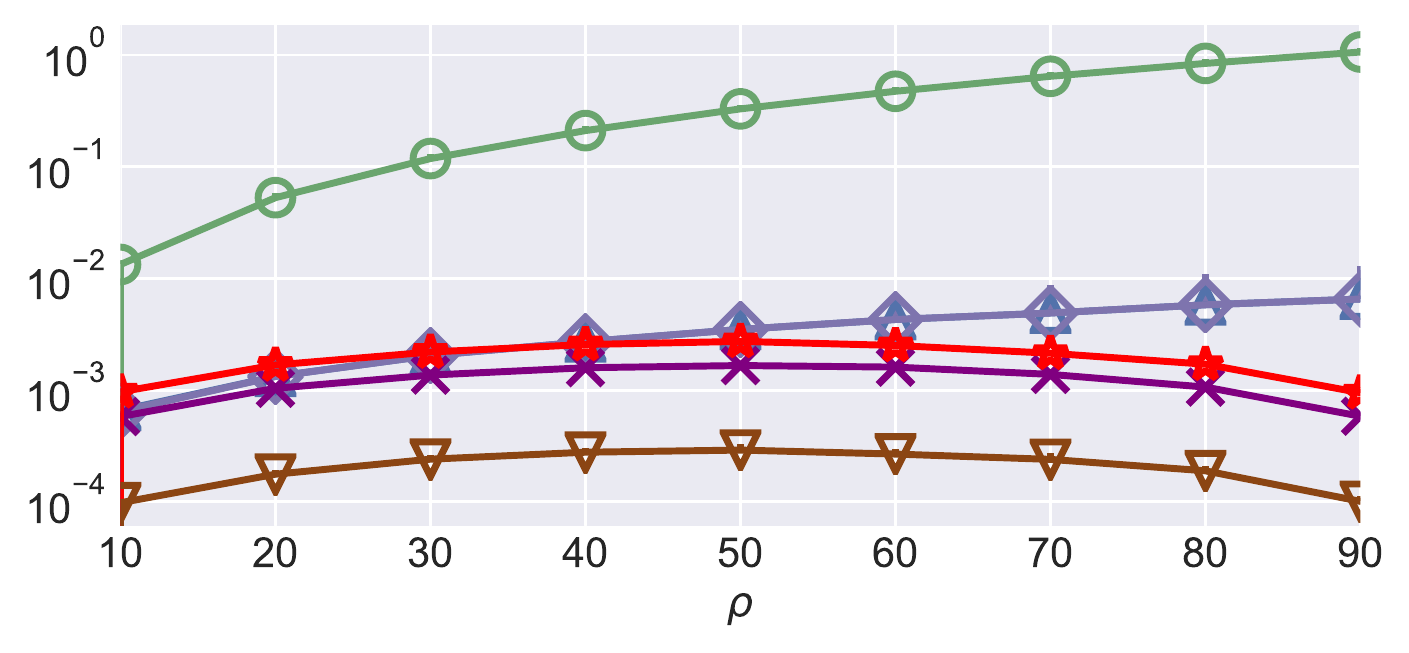}
	}
	\subfigure{
		\includegraphics[width=0.475\textwidth]{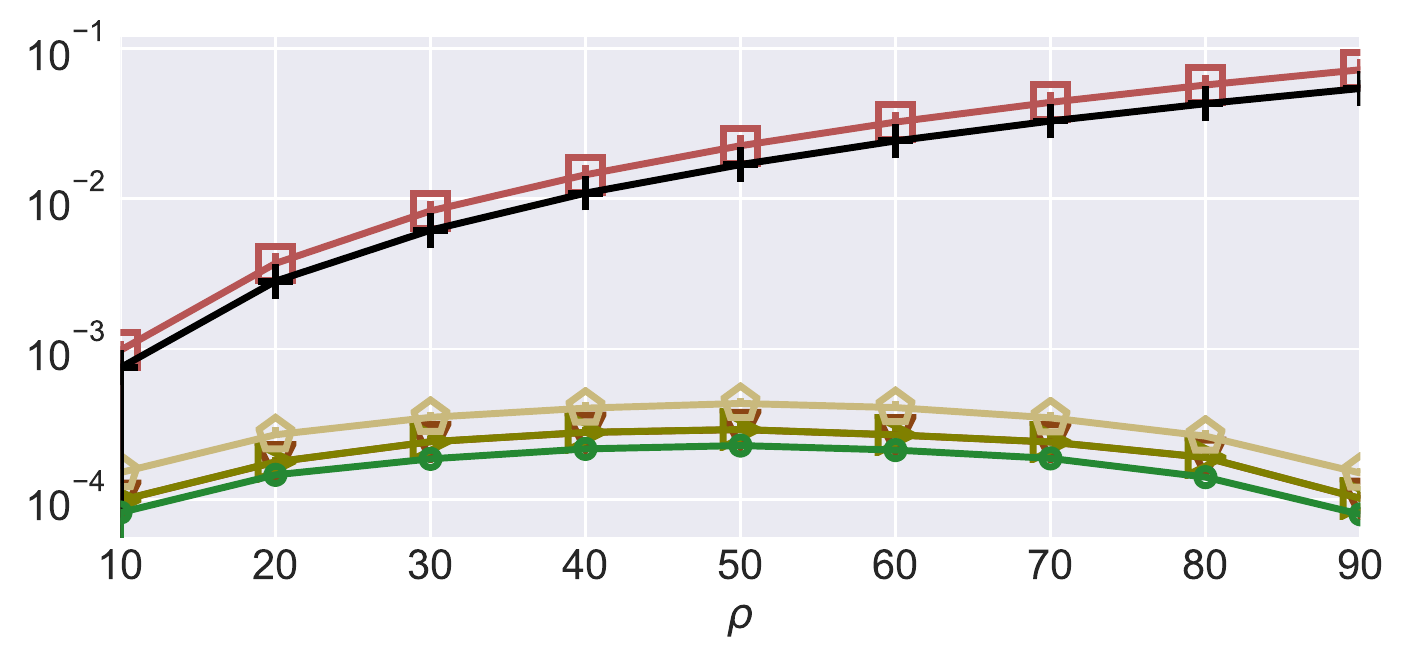}
	}
	Emoji
	\caption{MSE results on set-value estimation, varying set size percentage $\rho$ from $10$ to $90$, fixing $\epsilon=1$.}
	\label{fig:set_size}
\end{figure*}

\mypara{Plotting Convention.}
Unless otherwise specified, for each dataset and each method, we repeat the experiment $30$ times, with result mean and standard deviation reported.  The standard deviation is typically very small, and barely noticeable in the figures.

Because there are 11 algorithms (10 post-processing methods plus Base), and for any single metric there are often multiple methods that perform very similarly, resulting their lines overlapping.  To make Figures 4--8 readable, we plot results on two separate figures on the same row.  On the left, we plot 6 methods, Base, Base-Pos, Post-Pos, Norm, Norm-Mul, and Norm-Sub.  On the right, we plot Norm-Sub with the remaining 5 methods, MLE-Apx, Base-Cut, Norm-Cut, Power and PowerNS.  We mainly want to compare the methods in the right column.

\subsection{Bias-variance Evaluation}
Figure~\ref{fig:dist_est} shows the true distribution of the synthetic Zipf's dataset and the mean of the estimations.  As we plot the count estimations (instead of frequency estimations), the variance is larger (a $n^2=10^{12}$ multiplicative factor than the frequency estimations).  We thus estimate $5000$ times in order to make the mean stabilize.
In Figure~\ref{fig:dist_diff}, we subtract the estimation mean by the ground truth and plot the difference, which representing the empirical bias.
It can be seen that Base and Norm are unbiased.  Base-Pos introduces systematic positive bias.  Base-Cut gives unbiased estimations for the first few most frequent values, as their true frequencies are much greater than the threshold $T$ used to cut off estimation below it to $0$.  As the noise is close to normal distribution, the possibility that a high-frequency value is estimated to be below $T$ is exponentially small.  The similar analysis also holds for the low-frequency values, whose estimates are unlikely to be above $T$.  On the other hand, for values in between, the two biases compete with each other.  At some point, the two effects cancel out with each other, leading to unbiased estimations.  But this point is dependent on the whole distribution, and thus is hard to be found analytically.
For Norm-Cut, the similar reasoning also applies, with the difference that the threshold in Norm-Cut is typically smaller.
For Norm-Sub, each value is influenced by two factors: subtraction by a same amount; and converting to $0$ if negative.  For the high-frequency values, we mostly see the first factor; for the low-frequency values, they are mostly affected by the second factor; and for the values in between, the two factors compete against each other.  We see an increasing line for Norm-Sub.
Finally, Power changes little to the top estimations; but more to the low ones, thus leading to a similar shape as Norm-Cut.  The shape of PowerNS is close to Power because PowerNS applies Norm-Sub, which subtract some amount to the estimations, after Power.

Figure~\ref{fig:dist_var} shows the variance of the estimations among the 5000 runs.  First of all, the variance is similar for all the values in Base and Norm, with Norm being slightly better (smaller) than Base.  For all other methods, the variance drops with the rank, because for low-frequency values, their estimates are mostly zeros.

\subsection{Full-domain Evaluation}

Figure~\ref{fig:eps_full} shows MSE when querying the frequency of every value in the domain.  Note that The MSE is composed of the (square of) bias shown in Figure~\ref{fig:dist_diff} and variance in Figure~\ref{fig:dist_var}.  We vary $\epsilon$ from $0.2$ to $4$.     Let us fist focus on the figures on the left.
Base performs very close to Norm, since the adjustment of Norm can be either positive or negative as the expected value of the estimation sum is 1.
As Base-Pos (which is equivalent to Post-Pos in this setting) converts negative results to $0$, its MSE is around half that of Base (note the y-axis is in log-scale).
Norm-Sub is able to reduce the MSE of Base by about a factor of 10 and 100 in the Zipfs and Emoji dataset respectively.
Norm-Mul behaves differently from other methods.  In particular, the MSE decreases much slower than other methods.
This is because Norm-Mul multiplies the original estimations by the same factor.  The higher the estimate, the greater the adjustment.  Since the estimations are individually unbiased, this is not the correct adjustment.

For the right part of Figure~\ref{fig:eps_full}, we observe that,
Norm-Sub and MLE-Apx perform almost exactly the same, validating the prediction from theoretical analysis.
Norm-Sub, MLE-Apx, Power, PowerNS, and Base-Cut perform very similarly.  In these two datasets, PowerNS performs the best.
Note that PowerNS works well when the distribution is close to Power-Law.  For an unknown distribution, we still recommend Base-Cut.  This is because if one considers average accuracy of all estimations, the dominating source of errors comes from the fact many values have true frequencies close or equal to $0$ are randomly perturbed. And Base-Cut maintains the high-frequency values unchanged, and
converts results below a threshold $T$ to $0$.
Norm-Cut also converts low estimations to 0, but the threshold $\theta$ is likely to be lower than $T$, because $\theta$ is chosen to achieve a sum of 1.

\mypara{Benefit of Post-Processing.}
We demonstrate the benefit of post-processing by measuring the relationship between $n$ and $n'$, so that $n$ records with post-processing can achieve the same accuracy for $n'$ records without it.  In particular, we vary $n$ and measure the errors for different methods.  We then calculate $n'$ using Equation~\ref{eq:var_general}.  In particular, the analytical MSE for $n'$ records is
\begin{align}
	\frac{1}{d}\sum_v\sigma^2_v= & \frac{q(1-q)}{n'(p-q)^2} + \frac{1}{d}\sum_v\frac{{f}_v(1-p-q)}{n'(p-q)}\nonumber \\
	=                            & \frac{q(1-q)}{n'(p-q)^2} + \frac{1}{d}\frac{1-p-q}{n'(p-q)}.\nonumber
\end{align}
Given the empirical MSE, we can obtain $n'$ that achieves the same error analytically.  Note that the MSE does not depend on the distribution.  Thus we only evaluate on the Zipf's dataset.  The result is shown in Figure~\ref{fig:eps_vary_n}.  We vary the size of the dataset $n$ and plot the value of $n'$ (note that the $x$-axes are in the scale of $10^6$ and $y$-axes are $10^7$).  The higher the line, the better the method performs.  Base and Norm are two straight lines with the slope of $1$, verifying the analytical variance.  The $y$ value for Norm-Mul grows even slower than Base, indicating the harm of using Norm-Mul as a post-processing method.  The performance of the other methods follow the similar trend of the full-domain MSE (as shown in the upper row of Figure~\ref{fig:eps_full}), with PowerNS gives the best performance, which saves around $90\%$ of users.

\begin{figure*}[ht]
	\centering
	\subfigure{
		\includegraphics[width=0.475\textwidth]{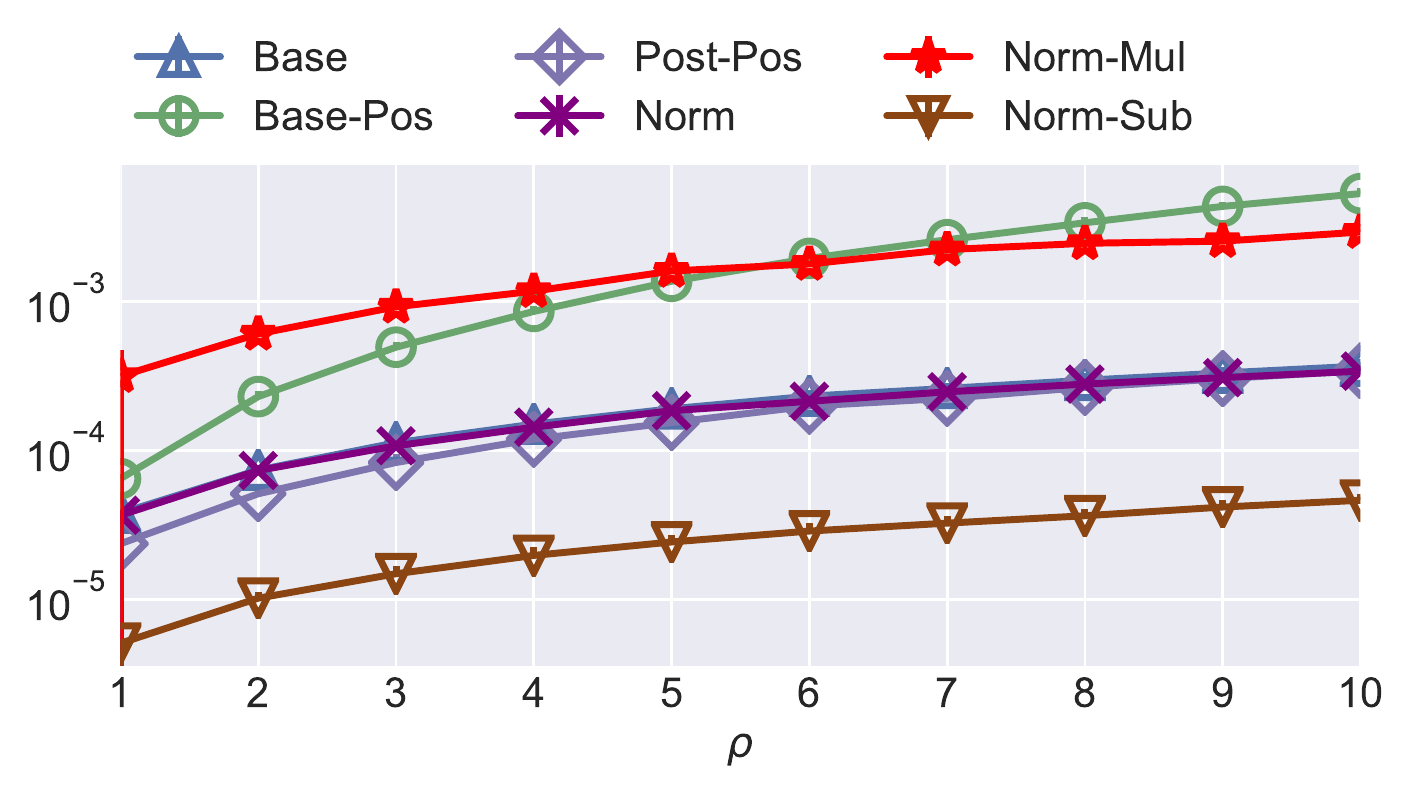}
	}
	\subfigure{
		\includegraphics[width=0.475\textwidth]{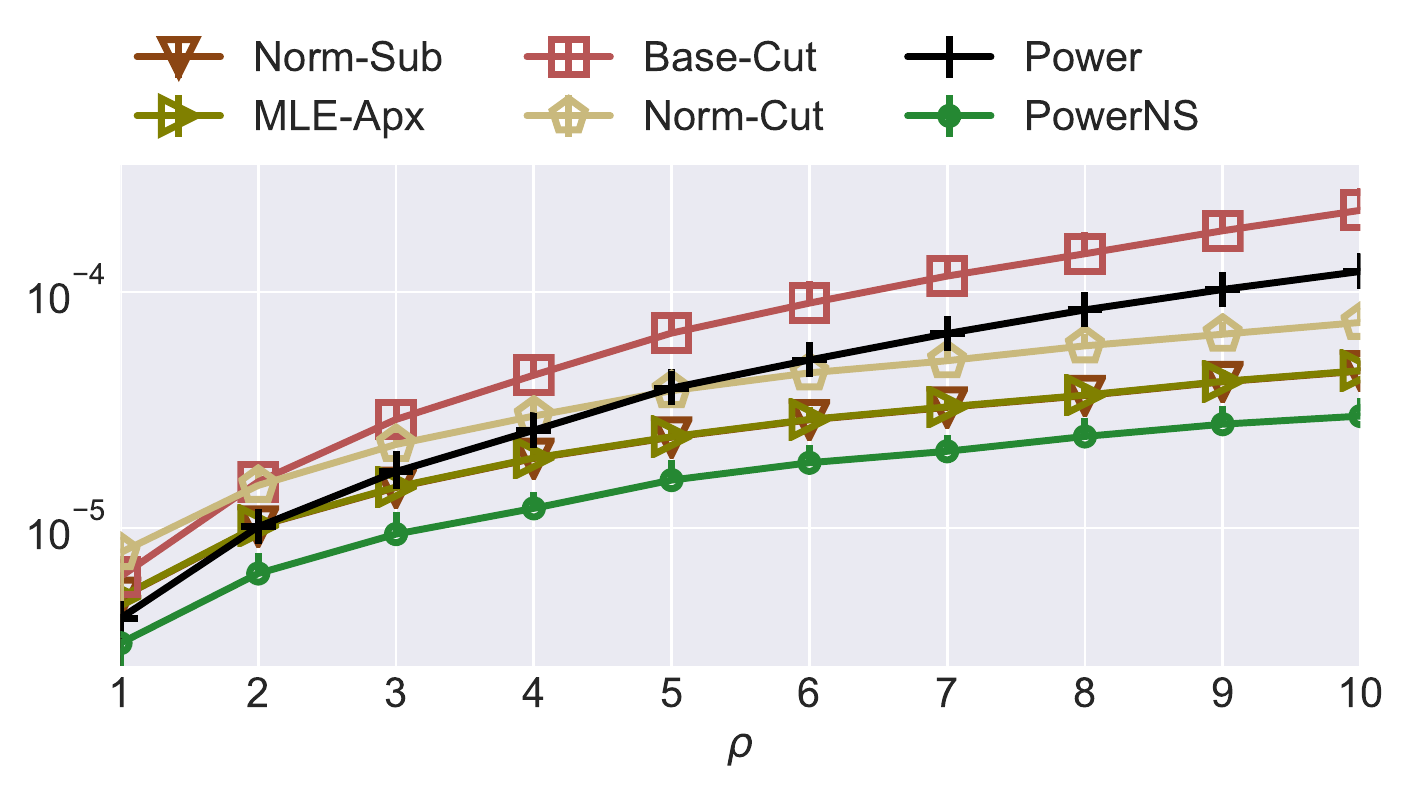}
	}
	Zipf's

	\subfigure{
		\includegraphics[width=0.475\textwidth]{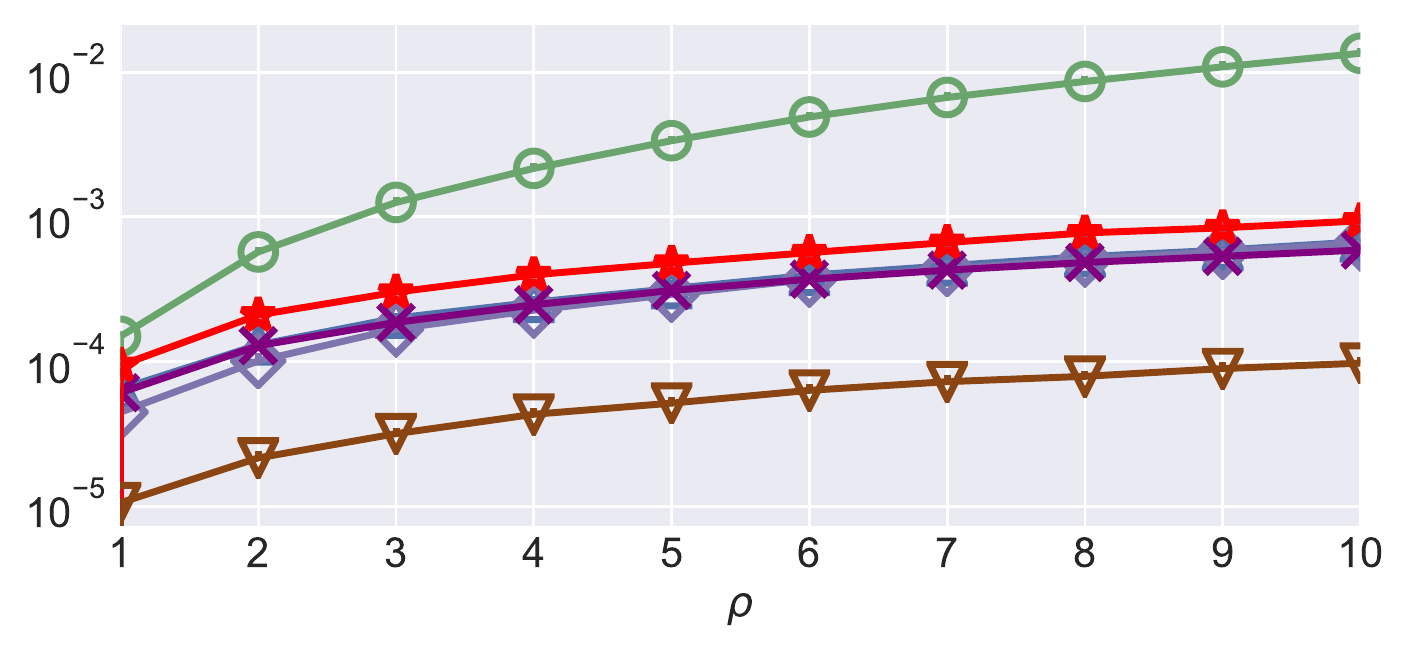}
	}
	\subfigure{
		\includegraphics[width=0.475\textwidth]{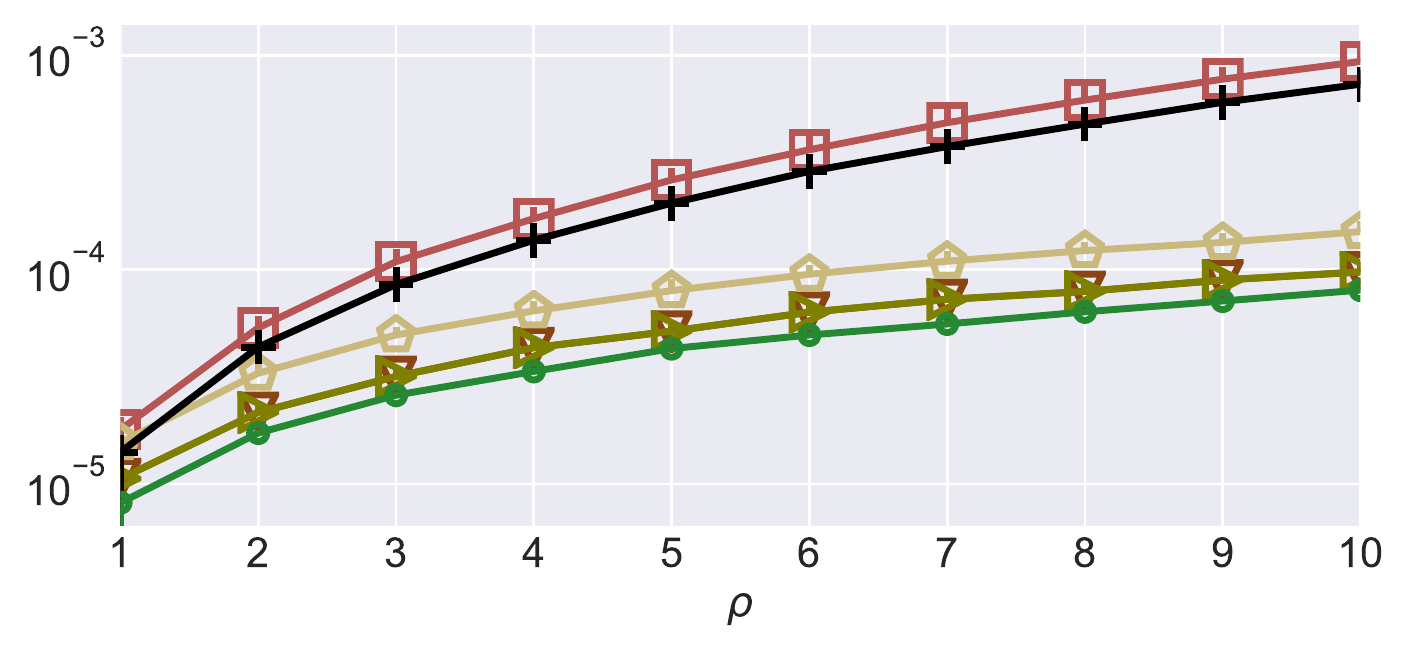}
	}
	Emoji
	\caption{MSE results on set-value estimation, varying set size percentage $\rho$ from $1$ to $10$, fixing $\epsilon=1$.}
	\label{fig:set_size_small}
\end{figure*}

\begin{figure*}[ht]
	\centering
	\subfigure{
		\includegraphics[width=0.475\textwidth]{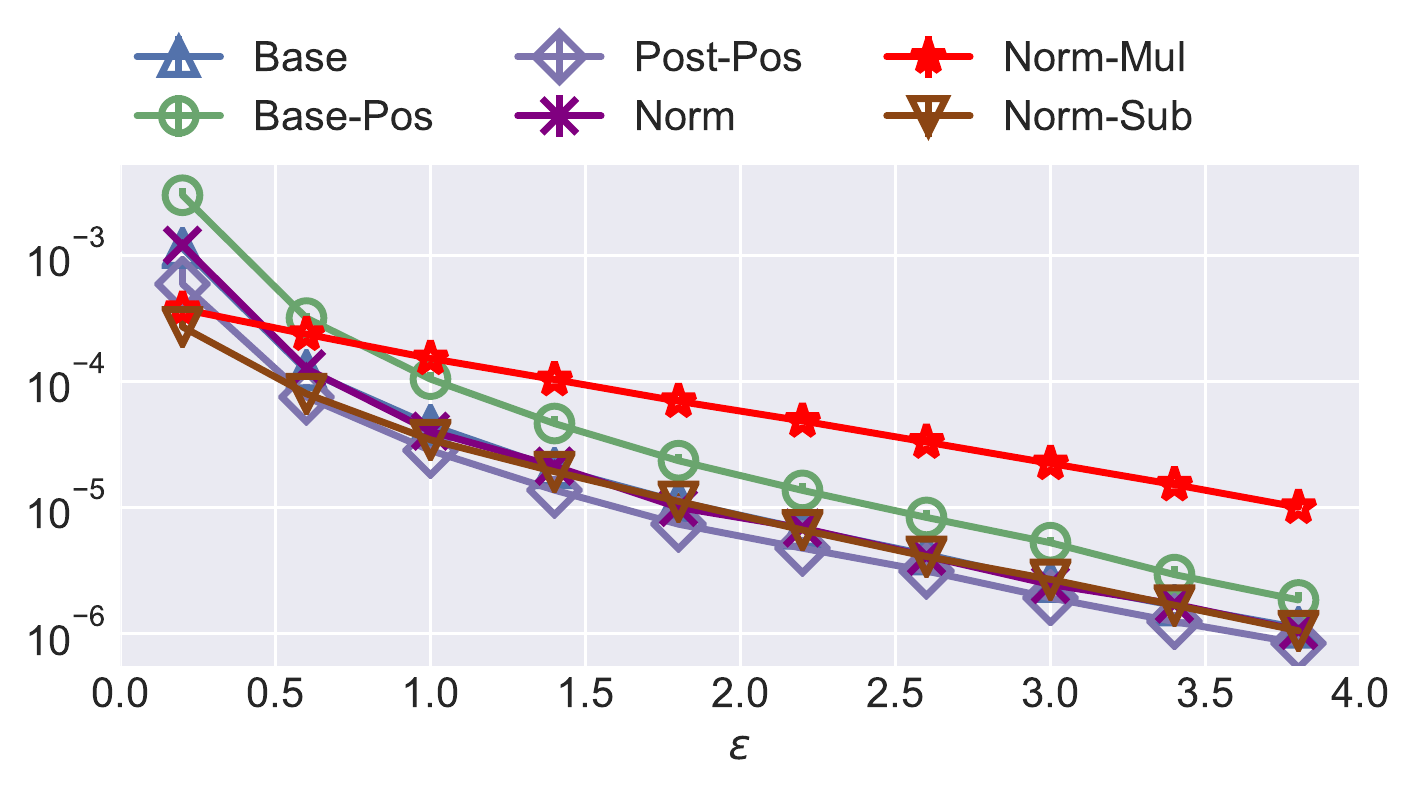}
	}
	\subfigure{
		\includegraphics[width=0.475\textwidth]{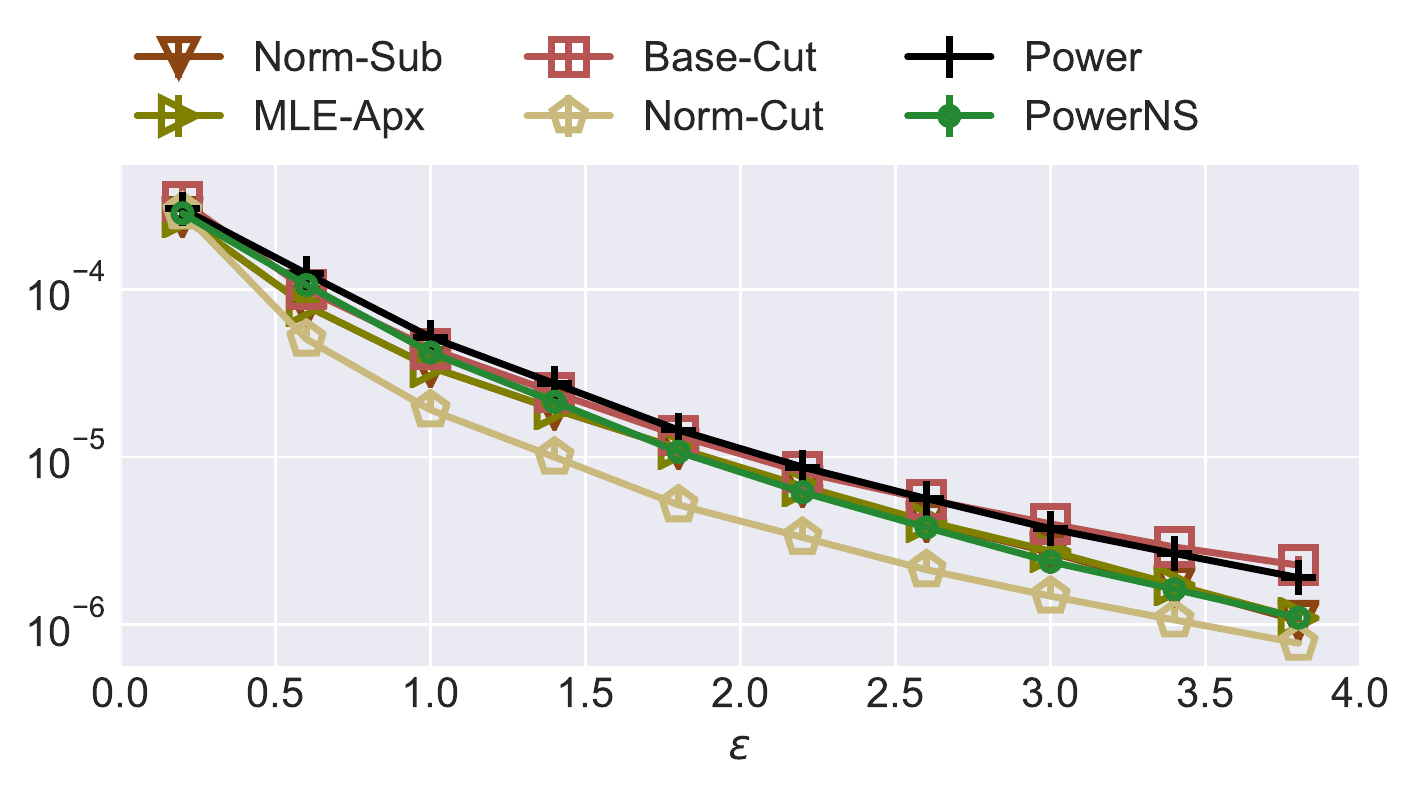}
	}
	\vspace{-0.3cm}
	\caption{MSE results on set-case estimation for the Emoji dataset, varying $\epsilon$ from $0.2$ to $4$.}
	\label{fig:eps_set_case}
\end{figure*}

\subsection{Set-value Evaluation}
Estimating set-values plays an important role
in the interactive data analysis setting (e.g., estimating which category of emoji's is more popular).
Keeping $\epsilon = 1$, we evaluate the performance of different methods by changing the size of the set.
For the set-value queries, we uniformly sample $\rho\%\times |\Domain|$ elements from the domain and evaluate the MSE between the sum of their true frequencies and estimated frequencies.  Formally, define $\Domain_{s\rho}$ as the random subset of $\Domain$ that has $\rho\%\times |\Domain|$ elements; and define $f_{\Domain_{s\rho}} = \sum_{v\in \Domain_{s\rho}} f_v$.  We sample $\Domain_{s\rho}$ multiple times and measure MSE between $f_{\Domain_{s\rho}}$ and $f'_{\Domain_{s\rho}}$.  Overall, the error MSE of set-value queries is greater than that for the full-domain evaluation, because the error for individual estimation accumulates.

\mypara{Vary $\rho$ from $10$ to $90$.}
Following the layout convention, we show results for set-value estimations in Figure~\ref{fig:set_size}, where we first vary $\rho$ from $10$ to $90$.  Overall, the approaches that exploits the summing-to-1 requirement, including Norm, Norm-Mul, Norm-Sub, MLE-Apx, Norm-Cut, and PowerNS, perform well, especially when $\rho$ is large.  Moreover, their MSE is symmetric with $\rho = 50$.  This is because as the results are normalized, estimating set-values for $\rho>50$ equals estimating the rest.  When $\rho = 90$, the best norm-based method, PowerNS, outperforms any of the non-norm based methods by at least $2$ orders of magnitude.

For each specific method, it is observed the MSE for Base-Pos is higher than other methods, because it only turns negative estimates to $0$, introducing systematic bias.  Post-Pos is slightly better than Base, as it turns negative query results to $0$.  In the settings we evaluated, Base-Cut also outperforms Base;  this happens because converting estimates below the threshold $T$ to $0$ is more likely to make the summation $f'_\Domain$ close to one.  Finally, Power only converts negative estimations to be positive, introducing systematic bias; PowerNS further makes them sum to 1, thus achieving better utility than all other methods.

\mypara{Vary $\rho$ from $1$ to $10$.}
Having examined the performance of set-queries for larger $\rho$, we then vary $\rho$ from $1$ to $10$ and demonstrate the results in Figure~\ref{fig:set_size_small}.  Within this $\rho$ range, the errors of all methods increase with $\rho$, which is as expected.  When $\rho$ becomes small, the performance of different methods approaches to that of full-domain estimation.

Norm-Cut varies the threshold so that after cutting, the remaining estimates sum up to one.  Thus the performance of Norm-Cut is better than Base-Cut especially when $\rho\ge 2$.  Intuitively, the norm-based methods should perform better answering set-queries.  But Norm-Mul does not.  This is because the multiplication operation reduces the large estimates a lot, making them biased.  This also demonstrates that enforcing sum-to-one is not enough.  Different approaches perform significantly different.

\mypara{Fixed set queries.}
Besides random set queries, we include a case study of fixed subset queries for the Emoji dataset.
The queries ask the frequency of each category\footnote{\url{https://data.world/kgarrett/emojis}}.  There are 68 categories with the mean of $10.4$ items per set.
The MSE varying $\epsilon$ is reported in Figure~\ref{fig:eps_set_case}.  It is interesting to see that the Post-Pos works best in the left sub-figure, and Norm-Cut from the right performs even better, especially when $\epsilon<3$.  This indicates the set-queries contain values that are infrequent.

\begin{figure}[ht]
	\centering
	\subfigure{
		\includegraphics[width=0.48\textwidth]{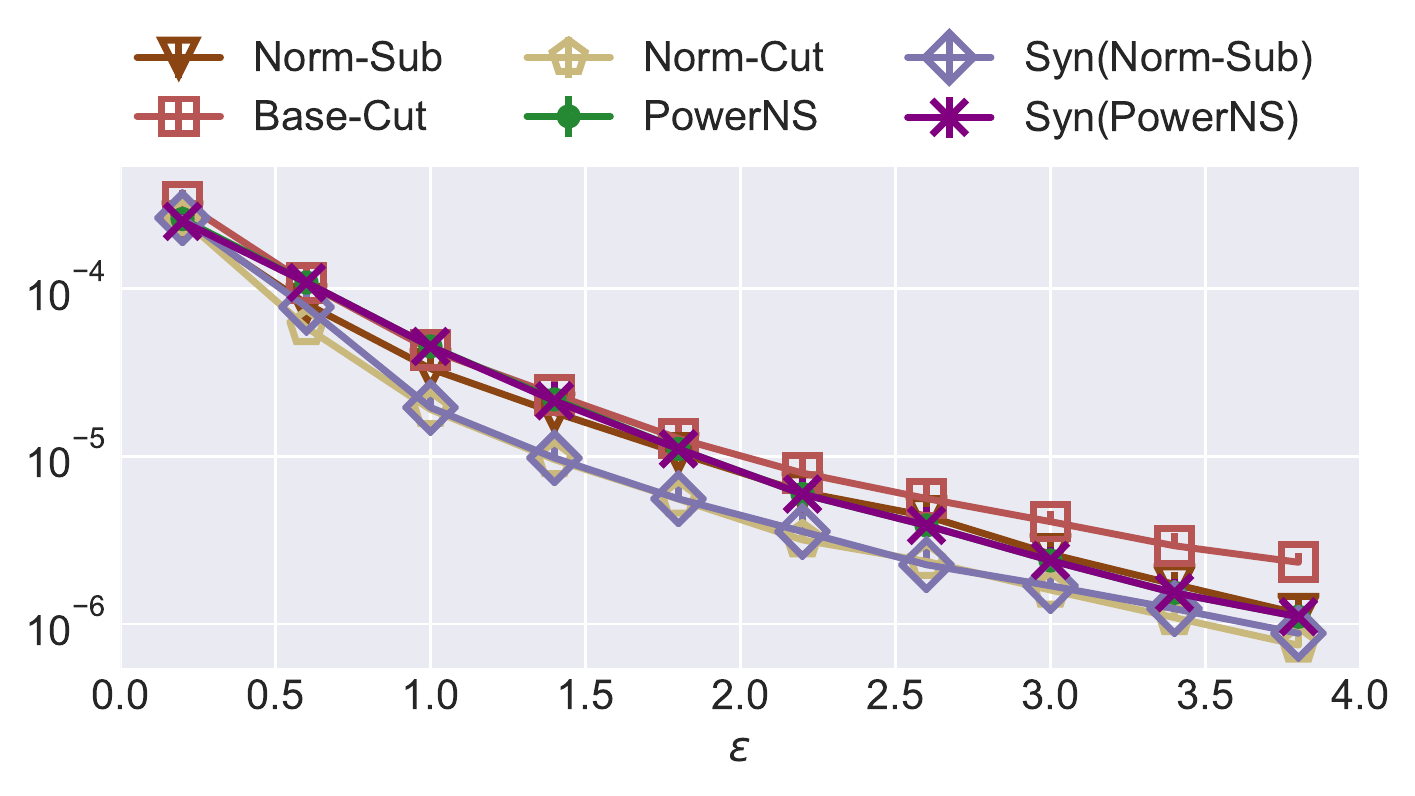}
	}
	\vspace{-0.8cm}
	\caption{Synthetic estimation for set-case query on the Emoji dataset.
	}
	\label{fig:eps_syn_full}
\end{figure}

\mypara{Choosing the method on synthetic dataset.}
As the optimal method in fixed set-values (as shown in Figure~\ref{fig:eps_set_case}) is different from random set-values (shown in Figure~\ref{fig:set_size} and~\ref{fig:set_size_small}), we investigate whether we can select the optimal post-processing method given the query and the LDP reports.  In particular, we first fit a synthetic dataset from the estimation, then we simulate the data collection and estimation process multiple times, with different post-processing methods, and we calculate the errors taking the synthesized dataset as the ground truth.  Figure~\ref{fig:eps_syn_full} shows the result.  Note that as we generate the synthetic dataset from the estimated distribution, the distribution itself should be consistent (non-negative and sum up to 1).  We select Norm-Sub and PowerNS to process the estimated distribution first.  These two methods perform well on full-domain and random set-value queries.

From the figure we can see that if the results are processed by Norm-Sub, the optimal method can be find quite accurately; if PowerNS is used, PowerNS will be selected.  The reason is that PowerNS makes the distribution more close to the prior of Power-Law distribution, while Norm-Sub does not.

\begin{figure*}[ht]
	\centering
	\subfigure{
		\includegraphics[width=0.475\textwidth]{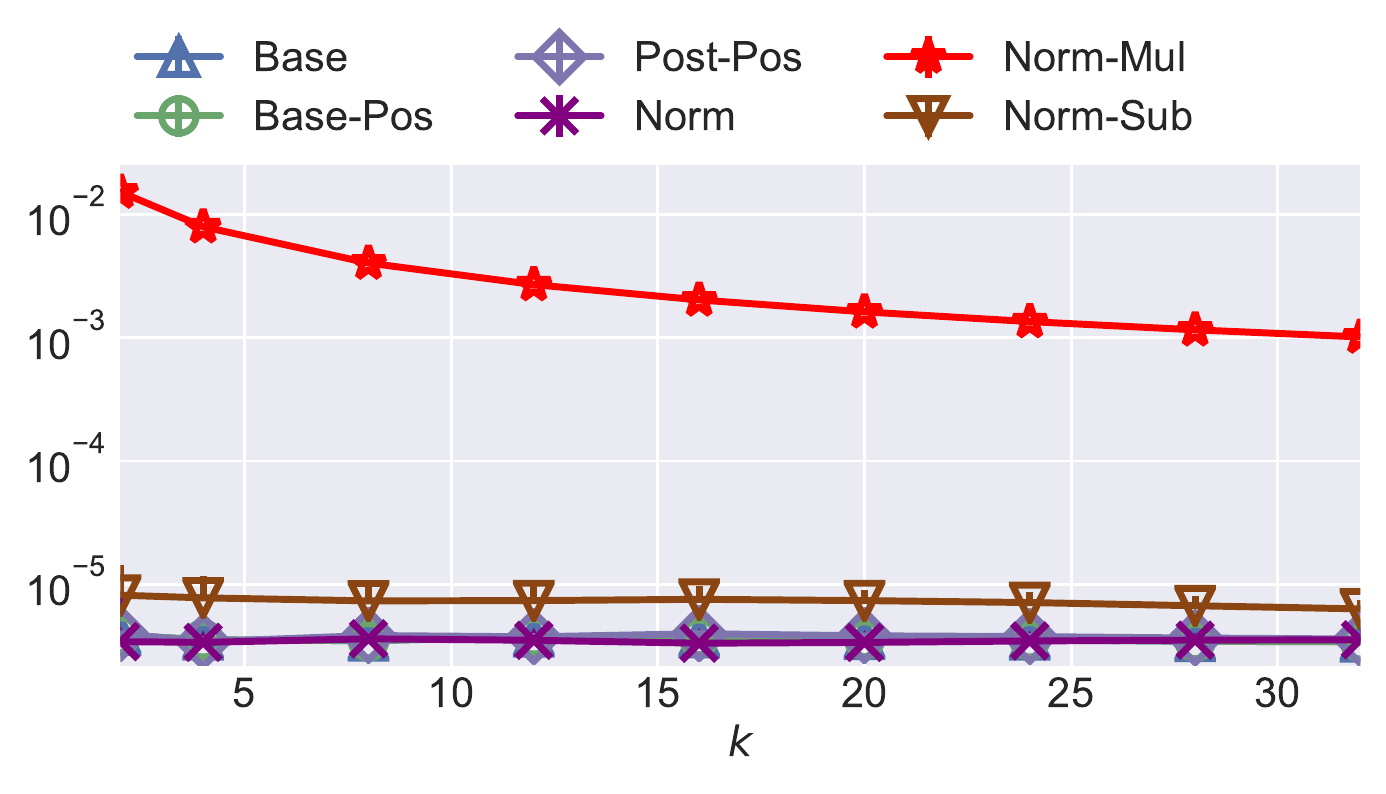}
	}
	\subfigure{
		\includegraphics[width=0.475\textwidth]{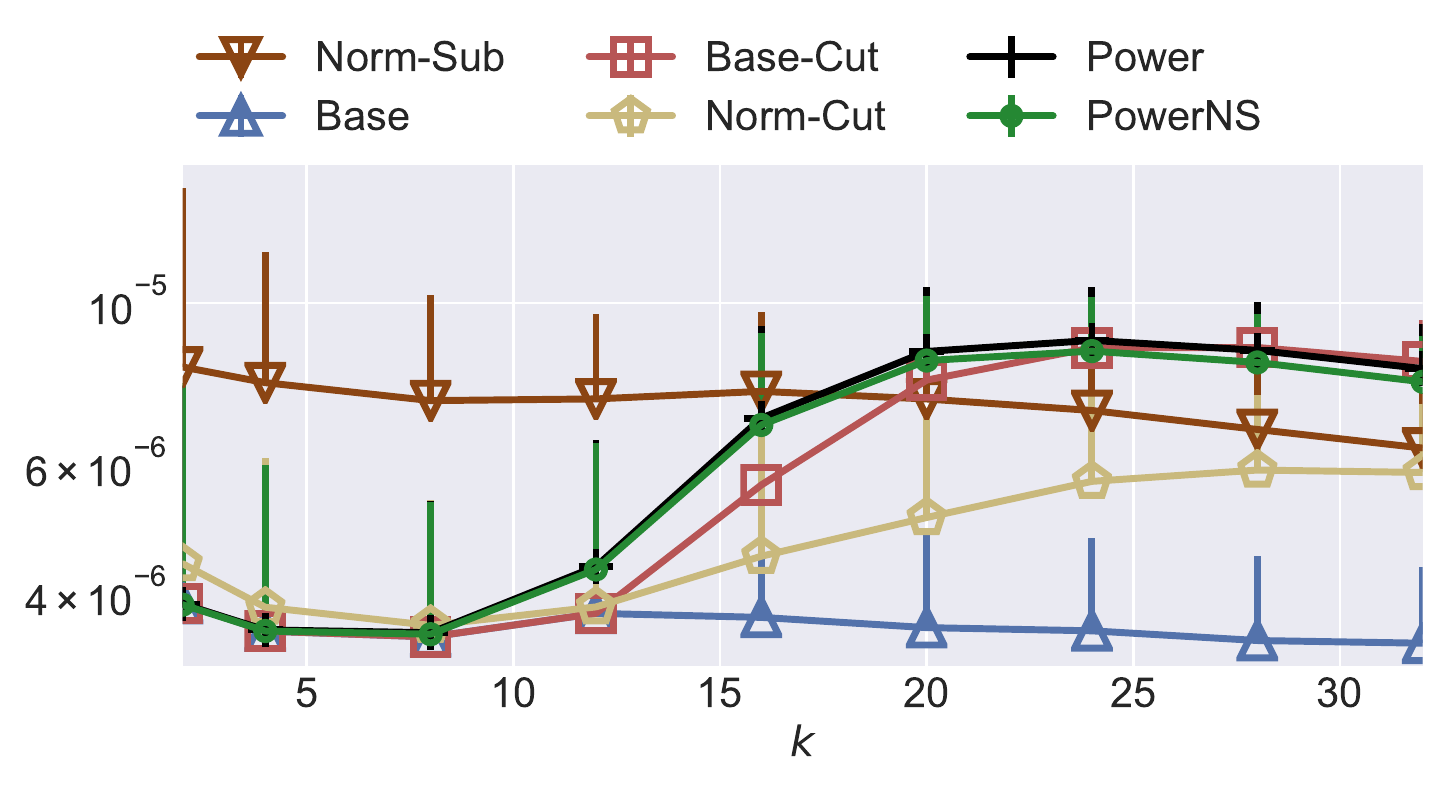}
	}
	Zipf's

	\subfigure{
		\includegraphics[width=0.475\textwidth]{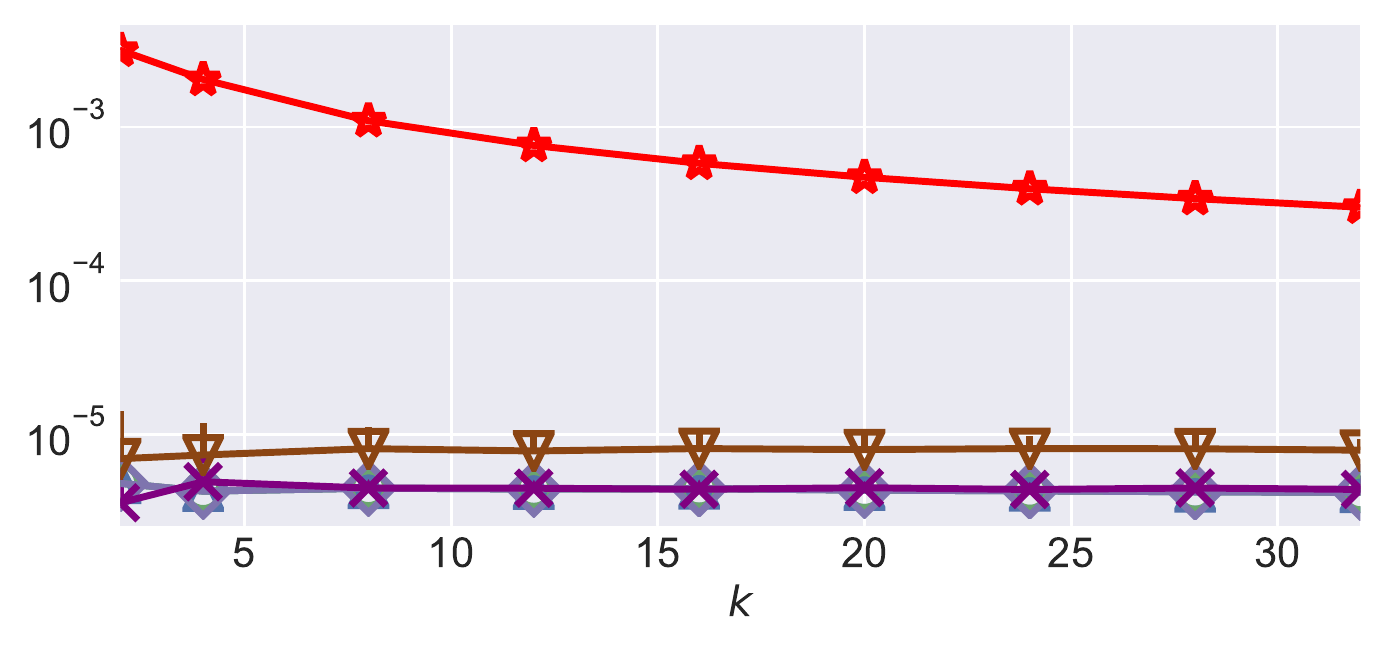}
	}
	\subfigure{
		\includegraphics[width=0.475\textwidth]{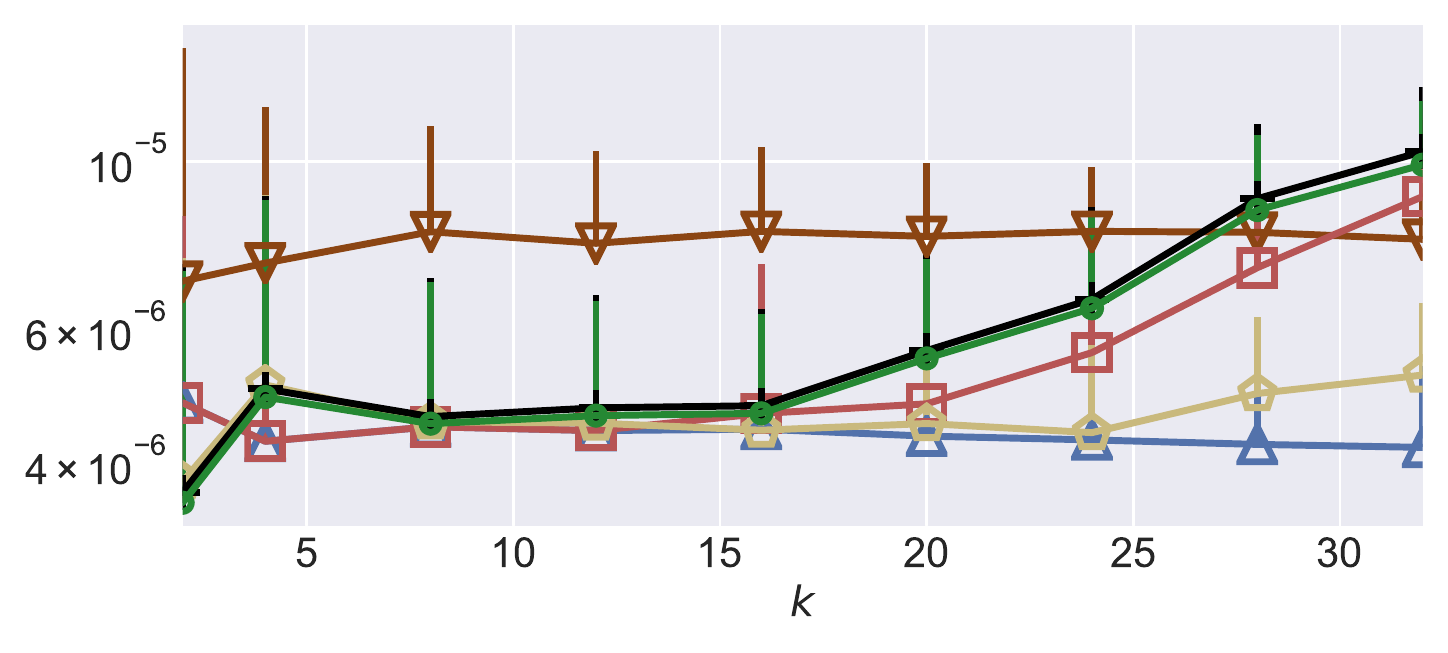}
	}
	Emoji
	\caption{MSE results on top-$k$ value estimation varying $k$ from $2$ to $32$, fixing $\epsilon=1$.}
	\label{fig:top_size}
\end{figure*}

\subsection{Frequent-value Evaluation}
\label{subsec:eval_frequent}
Finally, we evaluate different methods varying the top values to be considered.
Define $\Domain_{tk}$ as $\{v\in \Domain\mid f_v\mbox{ ranks top } k\}$.  We measure MSE between $(f'_v)_{v\in \Domain_{tk}}$ and $(f_v)_{v\in \Domain_{tk}}$ for different values of $k$ (from $2$ to $32$), fixing $\epsilon=1$.  Note that neither the frequency oracle nor the subsequent post-processing operation is aware of $\Domain_{tk}$.

From the left column of Figure~\ref{fig:top_size}, we observe that Base, Base-Pos, Post-Pos, and Norm perform consistently well for different $k$, as the first three methods do nothing to the top values, and Norm touches them in an unbiased way.
Norm-Mul performs at least $10\times$ worse than any other methods because it reduces the higher estimations a lot.
Norm-Sub performs worse than Base, but better than Norm-Mul, because the same amount is subtracted from every estimate, regardless of $k$.

To give a better comparison, we plot both Base and Norm-Sub to the right (i.e., we ignore MLE-Apx for now, as it performs the same as Norm-Sub).  These two methods have consistent MSE for different $k$.  The rest four methods, Base-Cut, Norm-Cut, Power, and PowerNS, all have MSE that grows with $k$.  In particular, for Base-Cut, a fixed threshold $T$ (in Equation~\eqref{eq:sig_threshold}) is used and estimates below it is converted to $0$.  This also suggests that at $\epsilon=1$, around $10$ values can be reliably estimated.  This also happens to Norm-Cut for the similar reason.  As Norm-Cut is better than Base-Cut, it suggests the threshold used in Norm-Cut is smaller than that in Base-Cut.  If $T$ is reduced, MSE of Base-Cut can be lowered until it matches that of Norm-Cut.  Thus $T$ is actually a tradeoff between frequent values and set-values.  In practice, if the desired $k$ is known in advance, one can set $T$ to be the $k$-th highest estimated value.  Finally, the performances of Power and PowerNS are similar, and they are worse than Base-Cut, especially when $k > 10$.

\subsection{Discussion}
In summary, we evaluate the $10$ post-processing methods on different datasets, for different tasks, and varying different parameters.  We now summarize the findings and present guidelines for using the post-processing methods.

With the experiments, we verify the connections among the methods:  Norm-Sub and MLE-Apx perform similarly, and Base and Norm performs similarly.

The best choice for post-processing method depends on the queries one wants to answer.
If set-value estimation is needed, one should use PowerNS.  When the set is fixed, one can also choose the optimal method using a synthetic dataset processed with Norm-Sub.  The intuition is that PowerNS improves over the approximate MLE (i.e., Norm-Sub, which is a theoretically testified method) by making the estimates closer to the underlying distribution.
If one just want to estimate results for the most frequent values, one can use Norm.  While Base can also be used, Norm reduces variance by utilizing the property that the estimates sum up to $1$.  These two methods do not change any value dramatically.
Finally, if one cares about single value queries only,
Base-Cut should be used.  This is because when many values in the dataset are of low frequency, converting low estimates to $0$ benefit the utility.
Overall, one can follow the guideline for choosing post-processing methods.
\begin{itemize}
	\item When single value queries are desired, use Base-Cut.
	\item When frequent values are desired, use Norm.
	\item When set-value queries are desired, use PowerNS or select one using synthetic datasets.
\end{itemize}

\section{Related Work}
\label{sec:related}

LDP frequency oracle (estimating frequencies of values) is a fundamental primitive in LDP. There have been several mechanisms~\cite{ccs:ErlingssonPK14,stoc:BassilyS15,uss:WangBLJ17,nips:BassilyNST17,aistats:AcharyaSZ18,tiot:YeB18} proposed for this task. Among them,~\cite{uss:WangBLJ17} introduces \olh, which achieves low estimation errors and low communication costs on large domains.  Hadamard Response~\cite{nips:BassilyNST17, aistats:AcharyaSZ18} is similar to \olh in essence, but uses the Hadamard transform instead of hash functions.  The aggregation part is faster because evaluating a Hadamard entry is practically faster; but it only outputs a binary value, which gives higher error than \olh for larger $\epsilon$ setting.  Subset selection~\cite{tiot:YeB18,corr:WangHWNXYLQ16} achieves better accuracy than \olh, but with a much higher communication cost.

LDP frequency oracle is also a building block for other analytical tasks, e.g., finding heavy hitters~\cite{nips:BassilyNST17,pods:BunNS18,tdsc:WangLJ19}, frequent itemset mining~\cite{ccs:qin2016heavy,sp:wang2018locally}, releasing marginals under LDP~\cite{tifs:RenYYYYMY18,sigmod:CormodeKS18,ccs:ZhangWLHC18}, key-value pair estimation~\cite{sp:YeHMZ19,uss:gu2019pckv}, evolving data monitoring~\cite{nips:JosephRUW18,soda:ErlingssonFMRTT18}, and (multi-dimensional) range analytics~\cite{sigmod:wang2019answering,vldb:KulkarniCD18}.
Mean estimation is also a building block in LDP; most of existing work transforms the numerical value to a discrete value using stochastic round, and then apply frequency oracles~\cite{focs:DuchiJW13,icde:WangXYZ19,li2019estimating}.

There exist efforts to post-process results in the setting of centralized DP.  Most of them focus on utilizing the structural information in problems other than the simple histogram, e.g., estimating marginals~\cite{sigmod:DingWHL11,sigmod:qardaji2014priview} and hierarchy structure~\cite{pvldb:HayRMS10}.  The methods do not consider the non-negativity constraint.  Other than that, they are similar to Norm-Sub and minimize $L_2$ distance.
On the other hand, the authors of~\cite{kdd:LeeWK15} started from MLE and propose a method to minimize $L_1$ instead of $L_2$ distance, as the DP noise follows Laplace distribution.

In the LDP setting, Kairouz et al. \cite{icml:KairouzBR16} study exact MLE for \grr and RAPPOR~\cite{ccs:ErlingssonPK14}; and empirically show exact MLE performs worse than Norm-Sub.  In~\cite{aistats:Bassily19}, Bassily proves the error bound of Norm-Sub for the Hadamard Response mechanism.  Jia et al.~\cite{infocom:JiaG18} propose to use external information about the dataset's distribution (e.g., assume the underlying dataset follows Gaussian or Zipf's distribution).
We note that such information may not always be available.
On the other hand, we exploit the basic information in each LDP setting.  That is, first, the total number of users is known; second, negative values are not possible.
We found that in the LDP setting, on the contrary to \cite{icml:KairouzBR16}, minimizing $L_2$ distance achieves MLE under the approximation that the noise is close to the Gaussian distribution.
There are also post-processing techniques proposed for other settings: Blasiok et al. \cite{soda:blasiok2019} study the post-processing for linear queries, which generalizes histogram estimation; but their method only applied to a non-optimal LDP mechanism.
\cite{icde:WangXYHSSY18} and \cite{vldb:KulkarniCD18} consider the hierarchy structure and apply the technique of \cite{pvldb:HayRMS10}. \cite{sp:YeHMZ19} considers mean estimation and propose to project the result into $[0, 1]$.

\section{Conclusion}
\label{sec:conc}

In this paper, we study how to post-process results from existing frequency oracles to make them consistent while achieving high accuracy for a wide range of tasks, including frequencies of individual values, frequencies of the most frequent values, and frequencies of subsets of values.  We considered 10 different methods, in addition to the baseline.  We identified Norm performs similar to Base, and MLE-Apx performs similar to Norm-Sub.  We then recommend that for full-domain estimation, Base-Cut should be used; when estimating frequency of the most frequent values, Norm should be used; when answering set-value queries, PowerNS or the optimal one from synthetic dataset should be used.

\section*{Acknowledgement}
This project is supported by NSF grant 1640374, NWO grant 628.001.026, and NSF grant 1931443. We thank our shepherd Neil Gong and the anonymous reviewers for their helpful suggestions.
	{
		\bibliographystyle{abbrv}
		\bibliography{ref}
	}

\appendices

\section{Solution for CLS}
\label{app:kkt_cls}

Using the KKT condition~\cite{book:kuhn2014nonlinear,disertation:karush1939minima}, we augment the optimization target with the following equations:
\begin{eqnarray}
	\recht{minimize} && \sum_{v} (f'_v - \tilde{f}_v)^2 + a + b \nonumber\\
	\recht{where} && \sum_v f'_v = 1,\;\;\forall v: 0 \leq f'_v \leq 1, \nonumber\\
	&& a=\mu\cdot\sum_v f'_v,  b=\sum_v \lambda_v\cdot f'_v,  \forall v: \lambda_v\cdot f'_v = 0.\nonumber
\end{eqnarray}
Since $b=0$, and $a=\mu$ is a constant, the condition that minimizing the target is unchanged.  Given that the target is convex, we can find the minimum by taking the partial derivative with respect to each variable:
\begin{align}
	         & \frac{\partial \left[\sum_{v} (f'_v - \tilde{f}_v)^2 + a + b \right]}{\partial f'_v} = 0\nonumber \\
	\implies & 2(f'_v - \tilde{f}_v) + \mu + \lambda_v = 0\nonumber                                              \\
	\implies & f'_v = \tilde{f}_v - \frac{1}{2}(\mu + \lambda_v)\nonumber
\end{align}

Now suppose there is a subset of domain $D_0\subseteq D$ s.t., $\forall v\in D_0, f'_v=0$ and $\forall v\in D_1=D\setminus D_0, f'_v> 0\wedge \lambda_v=0$.  By summing up $f'_v$ for all $v\in D_1$, we have
\begin{align}
	1 & = \sum_{v\in D_1}\tilde{f}_v -  \frac{|D_1|\mu}{2}\nonumber
\end{align}

Thus for all $v\in D_1$, we can use the formula
\begin{align*}
	f'_v = & \tilde{f}_v - \frac{1}{|D_1|}\left(\sum_{v\in D_1}\tilde{f}_v - 1\right)
\end{align*}
to derive the estimate $f'_v$ for value $v\in D_1$, and $f'_v = 0$ for $v\in D_0$.  One can also find $D_0$ using a similar approach when dealing with MLE.  And it can also be verified $\sum_v f'_v = 1$.

\section{Solution for MLE-Apx}
\label{app:kkt_mle}
From Equation~\eqref{eq:mle_apx}, we first simplify the exponent plugging in the value of $\sigma'_v$ as in Equation~\eqref{eq:var_general}:
\begin{align*}
	\sum_v \frac{(f'_v - \tilde{f}_v)^2}{2\sigma_v^{\prime 2}}
	= & \frac{n}{2} \sum_v \frac{ (f'_v - \tilde{f}_v)^2 (p-q)^2}{q(1-q) + f'_v(p-q)(1-p-q)}
\end{align*}

The factor $\frac{n}{2}$ in the exponent ensures that for large $n$ the exponent will vary the most with $\mathbf{f}'$, which dominates the coefficient $\frac{1}{\sqrt{2\pi \prod_v \sigma_v^{\prime 2}}}$.  Thus approximately we find $\mathbf{f}'$ that achieves the following optimization goal:
\begin{align}
	\mbox{minimize: }   & \sum_v \frac{ (f'_v - \tilde{f}_v)^2 (p-q)^2}{q(1-q) + f'_v(p-q)(1-p-q)}\nonumber \\
	\mbox{subject to: } & \sum_v f'_v = 1, \nonumber                                                        \\
	                    & \forall v, 0 \leq f'_v \leq 1. \nonumber
\end{align}

Using the KKT condition~\cite{book:kuhn2014nonlinear,disertation:karush1939minima}, we augment the optimization target with the following equations:
\begin{eqnarray}
	\recht{minimize} && \sum_v \frac{ (f'_v - \tilde{f}_v)^2 (p-q)^2}{q(1-q) + f'_v(p-q)(1-p-q)}  + a + b \nonumber\\
	\recht{where} && \sum_v f'_v = 1,\;\; \forall v: 0 \leq f'_v \leq 1, \nonumber\\
	&& a=\mu\cdot\sum_v f'_v, b=\sum_v \lambda_v\cdot f'_v, \forall v: \lambda_v\cdot f'_v = 0.\nonumber
\end{eqnarray}
Since $b=0$, and $a=\mu$ is a constant, the condition for minimizing the target is unchanged.  Given that the target is convex, we can find the minimum by taking the partial derivative with respect to each variable:
\begin{align}
	  & \frac{\partial \left[\sum_{v} \frac{ (f'_v - \tilde{f}_v)^2 (p-q)^2}{q(1-q) + f'_v(p-q)(1-p-q)} + a + b\right]}{\partial f'_v} \nonumber \\
	= & \frac{-(f'_v - \tilde{f}_v)^2 (p-q)^2 \cdot (p-q)(1-p-q)}{(q(1-q) + f'_v(p-q)(1-p-q))^2} \nonumber                                       \\
	+ & \frac{2(f'_v - \tilde{f}_v)(p - q)^2}{q(1-q) + f'_v(p-q)(1-p-q)} + \mu + \lambda_v = 0\nonumber
\end{align}
Define a temporary notation
\begin{align}
	x_v =                 & \frac{(f'_v - \tilde{f}_v)(p - q)}{q(1-q) + f'_v(p-q)(1-p-q)}\nonumber           \\
	\mbox{so that }f'_v = & \frac{q(1 - q)x_v + \tilde{f}_v(p - q)}{p - q - (p-q)(1-p-q)x_v}\label{eq:f_equ}
\end{align}
With $x_v$, we can simplify the previous equation:
\begin{align}
	(p-q)(1-p-q) x_v^2 - 2(p - q)x_v - \mu - \lambda_v = 0\label{eq:x_equ}
\end{align}

Now suppose there is a subset of domain $D_0\subseteq D$ s.t., $\forall v\in D_0, f'_v=0$ and $\forall v\in D_1=D\setminus D_0, f'_v> 0$ and $\lambda_v=0$.  Thus for those $v\in D_1$, solution of $x_v$ in Equation~\eqref{eq:x_equ} does not depend on $v$.  We solve $x_v$ by summing up $f'_v$ for all $v\in D_1$:
\begin{align}
	\sum_{v\in D_1} f'_v = & 1
	=  \sum_{v\in D_1}\frac{q(1 - q)x_v + \tilde{f}_v(p - q)}{p - q - (p-q)(1-p-q)x_v} \nonumber                            \\
	=                      & \frac{|D_1| q(1-q) x_v + \sum_{v\in D_1} \tilde{f}_v(p - q)}{p - q + (p-q)(1-p-q)x_v}\nonumber \\
	\implies x_v =         & \frac{\sum_{x\in D_1} \tilde{f}_v(p - q)  - (p - q)}{ (p-q)(1-p-q) - |D_1| q(1-q)}\nonumber
\end{align}
Given $x_v$, we can compute $f'_v$ from Equation~\eqref{eq:f_equ} for each value $v\in D_1$ efficiently; and $f'_v = 0$ for $v\in D_0$.  It can be verified $\sum_v f'_v = 1$.

Finally, to find $D_0$, one initiates $D_0=\emptyset$ and $D_1=D$, and iteratively tests whether all values in $D_1$ are positive.  In each iteration, for any negative $a_x$, $x$ is moved from $D_1$ to $D_0$.  The process terminates when no negative $a_x$ is found for all $x\in D_1$.

\end{document}